\documentclass[runningheads, envcountsame]{llncs}

\usepackage{hyperref}
\usepackage{transparent}
\usepackage{amsmath,amssymb,amstext}
\usepackage{stmaryrd}
\usepackage[noadjust]{cite}
\usepackage{url}
\usepackage{enumitem}
\usepackage{graphicx}
\usepackage{wrapfig}
\usepackage{bm}

\usepackage{CJKutf8} % for Japanese in references

\usepackage{tikz}
\usetikzlibrary{arrows,arrows.meta,automata,backgrounds,calc,chains,decorations.fractals,decorations.pathreplacing,fadings,fit,folding,mindmap,patterns,plotmarks,positioning,shadows,shapes.geometric,shapes.symbols,through,trees}
\usetikzlibrary{cd}

\colorlet{dblue}{blue!40!black}

\usepackage{enumitem}
\usepackage{verbatim}
\usepackage{float}
\usepackage{xspace}
\usepackage{latexsym}
\usepackage{fontawesome}
\usepackage{refcount}
\usepackage{caption}

\spnewtheorem{convention}[theorem]{Convention}{\bfseries}{\itshape}
\spnewtheorem{notation}[definition]{Notation}{\bfseries}{\itshape}

\newcounter{environment}

\newcommand{\store}[4]{
  \if\relax\detokenize{#2}\relax
  \expandafter\def\csname store#3\endcsname{\begin{#1}\label{#3}#4\end{#1}}%
  \else
  \expandafter\def\csname store#3\endcsname{\begin{#1}[#2]\label{#3}#4\end{#1}}%
  \fi
  \csname store#3\endcsname%
}

\newcommand{\use}[1]{{
  \renewcommand{\label}[1]{}%
  \renewcommand{\qedappendix}{}%
  \renewcommand{\footnote}[1]{}%
  \setcounter{environment}{\arabic{theorem}}%
  \setcounterref{theorem}{#1}%
  \addtocounter{theorem}{-1}%
  \csname store#1\endcsname%
  \setcounter{lemma}{\arabic{environment}}%
}}

\newcommand{\pbpostrong}{PBPO$^{+}$\xspace}

\newcommand{\catname}[1]{{\normalfont\textbf{#1}}\xspace}
\newcommand{\Set}{\catname{Set}}

\newcommand{\SimpleGraph}{\catname{SimpleGraph}}
\newcommand{\Graph}{\catname{Graph}}

\newcommand{\GraphLattice}[1]{\normalfont\textbf{Graph}$^{#1}$\xspace}

\newcommand{\labels}{\mathcal{L}}

\newcommand{\id}[1]{\mathrm{1}_{#1}}

\newcommand{\obj}[1]{\mathrm{Obj}(#1)}

\newcommand{\mono}{\rightarrowtail}

\newcommand{\epi}{\twoheadrightarrow}
\newcommand{\leftto}{\leftarrow}

\newcommand{\iso}{\cong}

\newcommand{\MM}{\mathcal{M}}
\newcommand{\CC}{\catname{C}}

\tikzset{graphNode/.style={circle,draw=black,inner sep=.5mm,outer sep=1mm}}

\tikzset{np/.style={n,fill=black}}
\tikzset{nc/.style={n,minimum size=2mm,fill=none,densely dotted}}
\tikzset{ep/.style={-red>}}
\tikzset{ec/.style={-red>,densely dotted}}

\tikzset{default/.style={
  thick,
  every node/.style={circle},
  level distance=12mm, 
  inner sep=.5mm}}
\tikzset{smallCircle/.style={circle,fill=black,inner sep=0mm,outer sep=1mm,minimum size=1mm}}
\tikzset{paint/.style={very thick,draw=#1!50!black,fill=#1,opacity=.4}}
\tikzset{paintopaque/.style={very thick,draw=#1!50!black!60,fill=#1!60}}
\tikzset{loopl/.style={out=140,in=210,looseness=10}}
\tikzset{loopr/.style={out=-40,in=30,looseness=10}}
\tikzset{loopb/.style={out=-40-90,in=30-90,looseness=10}}
\tikzset{loopt/.style={out=-40+90,in=30+90,looseness=10}}
\tikzset{loop/.style={out=-30+#1,in=30+#1,distance=6mm}}
\tikzset{thinloop/.style={out=-20+#1,in=20+#1,looseness=20}}

\tikzset{emptystep/.style={-,dotted,line cap=round,dash pattern=on 0 off 3.00000}}

\tikzset{b/.style={anchor=north,at=(#1.south)}}
\tikzset{br/.style={anchor=north west,at=(#1.south east)}}
\tikzset{bl/.style={anchor=north east,at=(#1.south west)}}
\tikzset{bw/.style={anchor=north west,at=(#1.south west)}}
\tikzset{be/.style={anchor=north east,at=(#1.south east)}}
\tikzset{a/.style={anchor=south,at=(#1.north)}}
\tikzset{ar/.style={anchor=south west,at=(#1.north east)}}
\tikzset{al/.style={anchor=south east,at=(#1.north west)}}
\tikzset{aw/.style={anchor=south west,at=(#1.north west)}}
\tikzset{ae/.style={anchor=south east,at=(#1.north east)}}
\tikzset{r/.style={anchor=west,at=(#1.east)}}
\tikzset{l/.style={anchor=east,at=(#1.west)}}
\tikzset{rn/.style={anchor=north west,at=(#1.north east)}}

\tikzset{eta/.style={very thick,->,cblue!90!black}}
\tikzset{beta/.style={very thick,->,corange!80!white!90!black}}

\tikzset{devcirc/.style={circle,draw,fill=white,inner sep=0,minimum size=4\pgflinewidth}}
\tikzset{dev/.style={postaction={decorate},decoration={
  markings,
  mark=at position .5 with \node [devcirc] {};}}
}
    
\tikzset{medium tree/.style={
    level 1/.style={sibling distance=17mm},
    level 2/.style={sibling distance=9mm},
    level 3/.style={sibling distance=5mm},
    level 4/.style={sibling distance=4mm},
  }}

\tikzset{startAt/.style={inner sep=0mm,r=#1,xshift=-1.2mm,yshift=.8mm}}

\newcommand{\arrowTriangle}[2]{
  \pgfpathmoveto{\pgfpoint{-0.01#1+#2}{.6#1}}
  \pgfpathlineto{\pgfpoint{#1+#2}{0}}
  \pgfpathlineto{\pgfpoint{-0.01#1+#2}{-.6#1}}
  \pgfusepathqfill
}

\newdimen\prearrowsize
\newdimen\arrowsize
\newdimen\temparrowsize
\newcommand{\arrowscale}{5}
\newcommand{\setarrowsize}{
  \arrowsize=0.000000001pt
  \prearrowsize=\arrowscale\pgflinewidth
  \normalizearrowsize
}
\newcommand{\normalizearrowsize}{
  \ifdim\prearrowsize>2mm
    \addtolength{\arrowsize}{2mm}
    \addtolength{\prearrowsize}{-2mm}
    \temparrowsize=0.5\prearrowsize
    \prearrowsize=\temparrowsize
  \else
  \fi

  \addtolength{\arrowsize}{\prearrowsize}
}

\pgfarrowsdeclarealias{share>}{share>}{stealth'}{stealth'}%

\pgfarrowsdeclare{red>}{red>}
{
  \setarrowsize
  \pgfarrowsleftextend{1.5\pgflinewidth} 
  \pgfarrowsrightextend{\arrowsize-0.1\arrowsize}
}{
  \setarrowsize
  \pgfsetdash{}{0pt} % do not dash
  \pgfsetroundjoin % fix join 
  \pgfsetroundcap % fix cap
  \arrowTriangle{\arrowsize}{-0.1\arrowsize}
}

\pgfarrowsdeclare{red>>}{red>>}
{
  \setarrowsize
  \pgfarrowsleftextend{1.5\pgflinewidth}
  \pgfarrowsrightextend{\arrowsize-0.1\arrowsize+.8\arrowsize}
}{
  \setarrowsize
  \pgfsetdash{}{0pt} % do not dash
  \pgfsetroundjoin % fix join 
  \pgfsetroundcap % fix cap
  \arrowTriangle{\arrowsize}{-0.1\arrowsize}
  \arrowTriangle{\arrowsize}{-0.1\arrowsize+.8\arrowsize}
}

\pgfarrowsdeclare{red>>>}{red>>>}
{
  \setarrowsize
  \pgfarrowsleftextend{1.5\pgflinewidth}
  \pgfarrowsrightextend{\arrowsize-0.1\arrowsize+1.6\arrowsize}
}{
  \setarrowsize
  \pgfsetdash{}{0pt} % do not dash
  \pgfsetroundjoin % fix join 
  \pgfsetroundcap % fix cap
  \arrowTriangle{\arrowsize}{-0.1\arrowsize}
  \arrowTriangle{\arrowsize}{-0.1\arrowsize+.8\arrowsize}
  \arrowTriangle{\arrowsize}{-0.1\arrowsize+1.6\arrowsize}
}

\tikzstyle{gyellow}=[draw=black!80,top color=white!50,bottom color=black!20]
\tikzstyle{gblue}=[draw=blue!50,top color=white,bottom color=blue!60]
\tikzstyle{gred}=[draw=red!50,top color=white,bottom color=red!60]
\tikzstyle{ggreen}=[draw=blue!80!green!90!black,top color=white,bottom color=blue!80!green!60]

\tikzstyle{roundNode}=[gyellow,thick,circle,minimum size=4mm,inner sep=0.5mm]

\definecolor{cblue}{rgb}{0,0.4,0.7}
\definecolor{clighterblue}{rgb}{0,0.6,1.0}
\colorlet{cred}{red}
\colorlet{cgreen}{green!80!black}
\colorlet{corange}{orange!70!red}
\colorlet{cpureorange}{orange}
\colorlet{cpurple}{clighterblue!50!cred}

\colorlet{clightblue}{clighterblue!50!cblue!40}%{cblue!70!blue!30}
\colorlet{clightred}{cred!40}
\colorlet{clightgreen}{cgreen!80!cblue!40}
\colorlet{clightyellow}{corange!40!yellow!50}
\colorlet{clightorange}{cred!50!orange!40}
\colorlet{clightpurple}{clighterblue!50!cred!50}

\colorlet{cdarkred}{cred!70!black}
\colorlet{cdarkgreen}{cgreen!60!black}
\colorlet{cdarkblue}{cblue!60!black}

\colorlet{chighlight}{orange!50!yellow!60}

\tikzset{pgnode/.style={smallCircle,fill=white,draw=black,minimum size=1.3mm,outer sep=0.5mm}}
\tikzset{pgnodecolor/.style={pgnode,minimum size=4mm,scale=0.9}}
\tikzset{pgnodebig/.style={roundNode,gyellow,outer sep=1mm}}
\tikzset{pgrelation/.style={ultra thick,cblue!80!black,decorate,decoration={snake,amplitude=.4mm,segment length=4mm}}}

\tikzset{exi/.style={densely dotted}}
\tikzset{decweak/.style={-,green!50!black,very thick,opacity=0.7}}
\tikzset{decstrict/.style={->,orange,very thick,opacity=0.7}}

\tikzset{graphNode/.style={circle,draw=black,inner sep=.5mm,outer sep=1mm}}

\tikzset{sloop/.style={looseness=7}}

\tikzset{interconnect/.style={dotted,cred}}
\tikzset{match/.style={cdarkgreen,very thick}}
\tikzset{jigsaw/.style={circle,draw=black,minimum size=2mm,fill=white,minimum size=3.5mm}}

\colorlet{myblue}{blue!80!black}
\colorlet{mygreen}{cdarkgreen}
\colorlet{myred}{cred}
\colorlet{myorange}{corange}
\colorlet{mypurple}{blue!40!cred}

\tikzset{smalljigsaw/.style={rectangle,rounded corners=3mm,inner sep=1mm}}

\tikzset{epattern/.style={}}
\tikzset{eset/.style={draw=black!50}}
\tikzset{npattern/.style={rectangle,rounded corners=2mm,draw=black,inner sep=0.5mm,outer sep=.5mm,minimum size=4.5mm}}
\tikzset{nset/.style={npattern,draw=black!50,fill=white}}
\tikzset{label/.style={scale=0.85,inner sep=0,outer sep=0.5mm}}

\tikzset{short/.style={node distance=10mm}}

\newcommand{\annotate}[2]{
  \node at (#1.north east) [anchor=south west,xshift=-.75mm,yshift=-.75mm,label,outer sep=0mm] {#2};
}

\newcommand{\graphbox}[8]{
  \begin{scope}[xshift=#2,yshift=#3]
    \draw [rounded corners=2mm] (0,0) rectangle (#4,-#5);
    \node at (0,0mm) [anchor=north west,inner sep=1mm] {#1};
    \begin{scope}[xshift=#4/2+#6,yshift=#7] 
    #8
    \end{scope}
  \end{scope}
}

\newcommand{\transparentgraphbox}[8]{
{
  \transparent{0.5}
  \graphbox{#1}{#2}{#3}{#4}{#5}{#6}{#7}{#8}
  }
}

\newcommand{\vertex}[2]{%
  \begin{tikzpicture}[baseline=-1ex]%
    \node [rectangle,rounded corners=2mm,inner sep=0.5mm,fill=#2] {$#1$};%
  \end{tikzpicture}%
}

\newcommand{\rulescale}{0.9}

\makeatletter

\newcommand*{\saveparinfo}[1]{%
  \expandafter\xdef\csname my@#1@parindent\endcsname{\the\parindent}%
  \expandafter\xdef\csname my@#1@parskip\endcsname{\the\parskip}%
}

\newcommand*{\useparinfo}[1]{%
  \parindent=\csname my@#1@parindent\endcsname \relax
  \parskip=\csname my@#1@parskip\endcsname \relax
}

\makeatother

\newcommand{\qedappendix}{\hfill \textcolor{blue}{$\circledast$}}

\newcommand{\PB}{{\normalfont PB}}
\newcommand{\PO}{{\normalfont PO}}

\newcommand{\Graphleq}{{\GraphLattice{(\labels,\leq)}}}

\newcommand{\pbpostep}[5]{
{#1} \Rightarrow_{\mathrm{PBPO}}^{#3, (#4, #5)} {#2}
}

\newcommand{\pbposteprule}[3]{    {\Rightarrow_{\mathrm{PBPO}}^{#1, (#2, #3)}}
}

\newcommand{\pbpostepminimal}[1]{   {\Rightarrow_{\mathrm{PBPO}}^{#1}}
}

\newcommand{\pbporulematchonly}[2]{
\Rightarrow_{\mathrm{PBPO}}^{#1, #2}
}

\newcommand{\pbpostrongstep}[5]{
{#1} \Rightarrow_{\mathrm{PBPO}^{+}}^{#3, (#4, #5)} {#2}
}

\newcommand{\pbpostrongrulematchonly}[2]{
\Rightarrow_{\mathrm{PBPO}^{+}}^{#1, #2}
}

\newcommand{\pbpostrongsteprule}[3]{
{#1} \Rightarrow_{\mathrm{PBPO}^{+}}^{#3} {#2}
}

\newcommand{\pbpostrongstepminimal}[1]{
{\Rightarrow_{\mathrm{PBPO}^{+}}^{#1}}
}

\newcommand{\DPOstep}[4]{{#1} \Rightarrow_\mathrm{DPO}^{#3, #4} {#2}}
\newcommand{\DPOsteprule}[2]{
  {\Rightarrow_\mathrm{DPO}^{#1, #2}}
}

\newcommand{\materialization}[1]{{\langle #1 \rangle}}

\newcommand{\parclass}[2]{\varphi(#1,#2)}
\newcommand{\parclassid}[1]{\overline{#1}}

\usepackage{eso-pic}% http://ctan.org/pkg/eso-pic

\AddToShipoutPictureBG{% Add picture to background of every page
\ifnum\value{page}=1 % only first page
  \AtPageLowerLeft{%
    \raisebox{3\baselineskip}{\makebox[\paperwidth]{\begin{minipage}{21cm}\centering
      The final authenticated version is available online at \url{https://doi.org/10.1016/j.jlamp.2023.100873}.
    \end{minipage}}}%
  }
  \fi
}
\begin{document}

\title{Graph Rewriting and Relabeling with \pbpostrong:\\ A Unifying Theory for Quasitoposes}

\titlerunning{Graph Rewriting and Relabeling with \pbpostrong (Extended Version)}

\author{Roy Overbeek \and J\"{o}rg Endrullis \and Alo\"{i}s Rosset}
\authorrunning{Roy Overbeek \and J\"{o}rg Endrullis \and Alo\"{i}s Rosset}

\institute{Vrije Universiteit Amsterdam, Amsterdam, The Netherlands \\
  \email{\{r.overbeek, j.endrullis, a.rosset\}@vu.nl}
}
\maketitle

\begin{abstract}%
    We extend the powerful Pullback-Pushout (PBPO) approach for graph rewriting with strong matching.
    Our approach, called \pbpostrong, allows more control over the embedding of the pattern in the host graph, which is important for a large class of rewrite systems. We argue that \pbpostrong{} can be considered a unifying theory in the general setting of quasitoposes, by demonstrating that \pbpostrong{} can define a strict superset of the rewrite relations definable by PBPO, AGREE and DPO. Additionally, we show that \pbpostrong{} is well suited for rewriting labeled graphs and some classes of attributed graphs, by introducing a lattice structure on the label set and requiring graph morphisms to be order-preserving.
\end{abstract}

\saveparinfo{parinfo}%

\section{Introduction}
\label{sec:introduction}

Injectively matching a graph pattern $P$ into a host graph $G$ induces a classification of $G$ into three parts: (i)~a \emph{match graph} $M$, the image of $P$; (ii)~a \emph{context graph} $C$, the largest subgraph disjoint from $M$; and (iii)~a \emph{patch} $J$, the set of edges that are in neither $M$ nor $C$.
For example, if $P$ and $G$ are
respectively
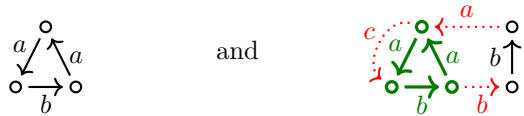
\begin{center}
    \begin{tikzpicture}[default,node distance=8mm,n/.style={graphNode}]
    \begin{scope}
    \node (3)[n] {};
    \node (4)[n] [right of=3] {};
    \node (5)[n] at ($(3)!.5!(4) + (0mm,8mm)$) {};
    \draw [->] (3) to node [below] {$b$} (4);
    \draw [->] (4) to node [right] {$a$} (5);
    \draw [->] (5) to node [above left] {$a$} (3);
    \end{scope}
    \node at (29.5mm,5mm) {and};
    \begin{scope}[xshift=50mm]
        \node (3)[n, match] {};
        \node (4)[n, match] [right of=3] {};
        \node (5)[n, match] at ($(3)!.5!(4) + (0mm,8mm)$) {};
        
        \node(6)[n] [right of=4] {};
        \node(7)[n] [above of=6] {};
        
        \draw [->, match] (3) to node [below] {$b$} (4);
        \draw [->, match] (4) to node [right] {$a$} (5);
        \draw [->, match] (5) to node [above left] {$a$} (3);
        \draw [->, myred, interconnect] (4) to node [below] {$b$} (6);
        \draw [->, myred, interconnect] (7) to node [above] {$a$} (5);
        \draw [->] (6) to node [left] {$b$} (7);
        \draw [->, myred, interconnect] (5) to[bend right=80] node [above left] {$c$} (3);
    \end{scope}
    \end{tikzpicture}
\end{center}
then $M$, $C$ and $J$ are indicated in green (bold), black and red (dotted), respectively.
We call this kind of classification a \emph{patch decomposition}.

Guided by the notion of patch decomposition, we recently introduced the expressive Patch Graph Rewriting (PGR) formalism~\cite{overbeek2020patch}. Like most graph rewriting formalisms, PGR rules specify a replacement of a left-hand side (lhs) pattern $L$ by a right-hand side (rhs) $R$. Unlike most rewriting formalisms, however, PGR rules allow one to (a)~constrain the permitted shapes of patches around a match for $L$, and (b)~specify how the permitted patches should be transformed, where transformations include rearrangement, deletion and duplication of patch edges. 

Whereas PGR is defined set theoretically, in this paper we propose a more elegant categorical approach, called \emph{\pbpostrong}, inspired by the same ideas. The name derives from the fact that the approach is obtained by strengthening the matching mechanism of the Pullback-Pushout (PBPO) approach by Corradini et al.~\cite{corradini2019pbpo}. Categorical approaches have at least three important advantages over set theoretic ones: (i)~the classes of structures the method can be applied to is vastly generalized, (ii)~typical meta-properties of interest (such as parallelism and concurrency) are more easily studied, and~(iii) it makes it easier to compare to existing categorical frameworks.

After discussing the preliminaries in Section~\ref{sec:preliminaries}, we introduce \pbpostrong in Section~\ref{sec:pbpostrong}, and then provide a detailed comparison with PBPO in Section~\ref{sec:relating}. We argue that \pbpostrong is preferable in situations where matching is not controlled, such as when specifying generative grammars or modeling execution.

Next, we study \pbpostrong in the setting of quasitoposes in Section~\ref{sec:quasitopos}. Quasitoposes generalize toposes, which can be described as capturing ``set-like'' categories~\cite[Preface]{wyler1991lecture}. Quasitoposes include not only the categories of sets, directed multigraphs~\cite{lack2004adhesive} and typed graphs~\cite{corradini2015agree} (all of which are toposes), but also a variety of structures such as Heyting algebras (considered as categories)~\cite{wyler1991lecture}, and the categories of simple graphs (equivalently, binary relations)~\cite{johnstone2007quasitoposes}, fuzzy sets~\cite{wyler1991lecture}, algebraic specifications~\cite{johnstone2007quasitoposes} and safely marked Petri Nets~\cite{johnstone2007quasitoposes}. Most importantly, we show that, using regular monic matching, \pbpostrong has enough expressive power in the quasitopos setting to generate any rewrite relation generated by the PBPO, AGREE~\cite{corradini2015agree,corradini2020algebraic}, or DPO~\cite{ehrig1973graph} rewrite formalisms, while the converse statements are not true. Thus, in this setting, \pbpostrong can be viewed as a unifying theory, additionally enabling the definition of new rewrite relations.

In Section~\ref{sec:category:graph:lattice}, we adopt a more applied perspective, and show that \pbpostrong{} easily lends itself for rewriting labeled graphs and certain attributed graphs. To this end, we define a generalization of the usual category of labeled graphs, \GraphLattice{(\labels,\leq)}, in which the set of labels forms a complete lattice $(\labels,\leq)$.
Not only does the combination of \pbpostrong and \GraphLattice{(\labels,\leq)} enable constraining and transforming the patch graph in flexible ways, it also allows naturally modeling notions of relabeling, variables and sorts in rewrite rules.
As we will clarify in the Discussion (Section~\ref{sec:discussion}), such mechanisms have typically been studied in the context of Double Pushout (DPO) rewriting~\cite{ehrig1973graph}, where the requirement to construct a pushout complement leads to technical complications and restrictions.

This paper extends the paper presented at ICGT2021~\cite{overbeek2021pbpo}. In the conference paper, we showed that in the setting of toposes, \pbpostrong can define all rewrite relations definable by PBPO~\cite[Section 4.3]{overbeek2021pbpo}. For the present extension, we have generalized this result to quasitoposes (additionally improving the proof), and integrated it into the wider quasitopos study provided in Section~\ref{sec:quasitopos}, which is completely new material. Additionally, we have improved the proofs and presentation in Section~\ref{sec:pbpostrong}, and we relate category \GraphLattice{(\labels,\leq)} to the category of fuzzy sets in Section~\ref{sec:quasitopos}.

\begin{remark}
    We have recently published a gentle tutorial on \pbpostrong{}~\cite{overbeek2023tutorial}, targeted especially to readers unacquainted with category theory.
\end{remark}

\section{Preliminaries}
\label{sec:preliminaries}

We assume familiarity with various basic categorical notions, notations and results, including morphisms $X \to Y$, pullbacks, pushouts, monomorphisms (monos) $X \mono Y$, epimorphisms (epis) $X \epi Y$, and identities $1_X : X \to X$~\cite{mac1971categories,awodey2006category}.

The following well-known lemma (also known as the pasting law for pullbacks) is used frequently throughout the paper.

\begin{lemma}[Pullback Lemma]\label{lemma:pullback:lemma}
    Suppose the right square of
    \begin{center}
        % https://tikzcd.yichuanshen.de/#N4Igdg9gJgpgziAXAbVABwnAlgFyxMJZABgBpiBdUkANwEMAbAVxiRAEEQBfU9TXfIRQBGclVqMWbAELdeIDNjwEiAJjHV6zVohABhOXyWCiZYeK1TdAEUML+yoclHnNknSACidxQJUp1VwltNgAxbnEYKABzeCJQADMAJwgAWyQAZmocCCQAFh5ElPTEPOzcxABWQpBktKR1EBykavk6ktEmioK24qRO5sRVGvakMi7+kb7EccGMqfrS8oa3EN0AHXWcGAAPHGAABWkuEGoGLDAPSEvTkAALGDooNmvWbLosBheCN5AGOgARjAGAcHCZdBdsLAIlwgA
        \begin{tikzcd}[row sep=3mm, column sep=5mm]
        A \arrow[r] \arrow[d] & B \arrow[d] \arrow[r]                        & C \arrow[d] \\
        D \arrow[r]           & E \arrow[r] \arrow[ru, "\mathrm{PB}", phantom] & F          
        \end{tikzcd}
    \end{center}
    is a pullback and the left square commutes. Then the outer square is a pullback iff the left square is a pullback. \qed
\end{lemma}

\begin{corollary}
\label{lem:pullback:cube}
    If in a commutative cube
    \begin{center}
        \begin{tikzcd}[row sep={20,between origins},column sep={20,between origins}]
            & F  \ar[rr] \ar[dd] \ar[dl] & &  G \ar[dd] \ar[dl] \\
            E \ar[rr, crossing over] \ar[dd] & & H \\
              & B \ar[rr] \ar[dl] & &  C \ar[dl] \\
            A \ar[rr] && D \ar[from=uu,crossing over]
        \end{tikzcd}
    \end{center}
    all vertical faces except the back face ($FGBC$) are known to be pullbacks, then the back face is also a pullback. \qed
\end{corollary}

\newcommand{\src}{\mathit{s}}
\newcommand{\tgt}{\mathit{t}}
\newcommand{\lbl}{\ell}
\newcommand{\lblv}{\lbl^V}
\newcommand{\lble}{\lbl^E}

\begin{definition}[Graph Notions]
    Let a label set $\labels$ be fixed.
    An ($\labels$-labeled) \emph{(multi)graph} $G$ consists of a set of vertices $V$, a set of edges $E$, source and target functions $\src,\tgt : E \to V$, and label functions $\lblv : V \to \labels$ and $\lble : E \to \labels$.
    
    A graph is \emph{unlabeled} if $\labels$ is a singleton.
    
    A \emph{premorphism} between graphs $G$ and $G'$ is a pair of maps
    \[
        \phi = (\phi_V : V_G \to V_{G'}, \phi_E : E_G \to E_{G'})
    \]
    satisfying $(s_{G'}, t_{G'}) \circ \phi_E = \phi_V \circ (s_G, t_G)$.
  
    A \emph{homomorphism} is a label-preserving premorphism $\phi$, i.e., a premorphism satisfying $\lblv_{G'} \circ \phi_V = \lblv_{G}$ and $\lble_{G'} \circ \phi_E = \lble_{G}$.
\end{definition}

\begin{definition}[Category $\Graph$~\cite{ehrig2006}]
    The category $\Graph$ has graphs as objects, parameterized over some global (and usually implicit) label set $\labels$, and homomorphisms as arrows.
\end{definition}

\begin{remark}
\label{remark:graph:as:presheaf}
    $\Graph$ can also be obtained as a slice of a presheaf category (see Corradini et al.~\cite[Section 2.1-2]{corradini2020algebraic} for details), by virtue of which one immediately obtains that it is a topos, and hence a quasitopos.
\end{remark}

The definition below formally fixes some graph terminology (see Section~\ref{sec:introduction}).

\begin{definition}[Patch Decomposition~\cite{overbeek2020patch}]
    Given a premorphism $x : X \to G$, we call the image $M = \mathit{im}(x)$ of $x$ the \emph{match graph} in $G$, $G - M$  the \emph{context graph} $C$ induced by $x$ (i.e., $C$ is the largest subgraph  disjoint from $M$), and the set of edges $E_G - E_M - E_C$ the set of \emph{patch edges} (or simply, \emph{patch}) induced by $x$. 
\end{definition}

\section{\pbpostrong{}}
\label{sec:pbpostrong}

We introduce \pbpostrong, which strengthens the matching mechanism of PBPO~\cite{corradini2019pbpo}. In Section~\ref{sec:relating}, we compare the two approaches in detail.

\begin{definition}[\pbpostrong Rewrite Rule]\label{def:pbpostrong:rewrite:rule}
    A \emph{\pbpostrong rewrite rule} $\rho$ is a diagram
    \begin{center}
      $\rho \ = \ $
      % https://tikzcd.yichuanshen.de/#N4Igdg9gJgpgziAXAbVABwnAlgFyxMJZABgBpiBdUkANwEMAbAVxiRABkQBfU9TXfIRRkAjFVqMWbdgHJuvEBmx4CREeXH1mrRCADS8vssFrSY6lqm69cnkYGqUAJg0XJOkACVu4mFADm8ESgAGYAThAAtkhkIDgQSOoS2mw4APqc1Ax0AEYwDAAK-CpCIFhg2LCGIOFRSADM1PGJbim6DHJZuflFxo5lFVhVdjUR0YgucQmISZYeADrzkXQ4ABZhkcAFAEJcIF15hcUmuuWVrFnlHpBgFyCrMHRQbDd3OHRYDC8ErCO145NmohGskrHE0gYDj1jv0zkNfgp-khAdNYnM2F8oUc+qU4cNEWNkU1pgAWVpgsL7EDZQ69By4wbDChcIA
        \begin{tikzcd}
        L \arrow[d, "t_L" description] & K \arrow[ld, "\mathrm{PB}", phantom] \arrow[d, "t_K" description] \arrow[l, "l" description] \arrow[r, "r" description] & R \\
        L'                             & K' \arrow[l, "l'" description]                                                                                          &  
        \end{tikzcd}.
    \end{center}
    $L$ is the \emph{lhs pattern} of the rule, $L'$ its \emph{(context) type} and $t_L$ the \emph{(context) typing} of $L$. Similarly for the \emph{interface} $K$. $R$ is the \emph{rhs pattern} or \emph{replacement for $L$}.
\end{definition}

We often depict the pushout $K' \xrightarrow{r'} R' \xleftarrow{t_R} R$ for span $K' \xleftarrow{t_K} K \xrightarrow{r} R$, because it shows the schematic effect of applying the rewrite rule. We reduce the opacity of $R'$ to emphasize that it is not part of the rule definition. 

\begin{example}[Rewrite Rule in \Graph]\label{ex:simple:rewrite:rule}
    A simple example of a rule for unlabeled graphs is the following:
    {% scope

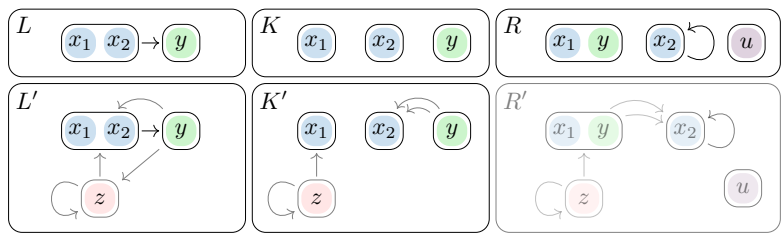
\newcommand{\nodexa}{\vertex{x_1}{cblue!20}}
\newcommand{\nodexb}{\vertex{x_2}{cblue!20}}
\newcommand{\nodey}{\vertex{y}{cgreen!20}}
\newcommand{\nodez}{\vertex{z}{cred!10}}
\newcommand{\nodeu}{\vertex{u}{cpurple!25}}

\begin{center}
  \scalebox{\rulescale}{
  \begin{tikzpicture}[->,node distance=12mm,n/.style={}]
    \graphbox{$L$}{0mm}{0mm}{35mm}{10mm}{-4mm}{-5mm}{
      \node [npattern] (xaxb)
      {\nodexa \ \nodexb};
      \node [npattern] (y) [right of=xaxb] {\nodey};
      \draw [epattern] (xaxb) to node {} (y);
    }
    \graphbox{$K$}{36mm}{0mm}{35mm}{10mm}{-8mm}{-5mm}{
      \node [npattern] (xa) {\nodexa};
      \node [npattern] (xb) [right of=xa, short] {\nodexb};
      \node [npattern] (y) [right of=xb, short] {\nodey};
    }
    \graphbox{$R$}{72mm}{0mm}{42mm}{10mm}{-8mm}{-5mm}{
      \node [npattern] (xay)
      {\nodexa \ \nodey};
      \node [npattern] (xb) [right of=xay] {\nodexb};
      \draw [epattern,loop=0,looseness=3] (xb) to node {} (xb);
      
      \node [npattern] (u) [right of=xb] {\nodeu};
      
    }
    \graphbox{$L'$}{0mm}{-11mm}{35mm}{22mm}{-4mm}{-7mm}{
      \node [npattern] (xaxb)
      {\nodexa \ \nodexb};
      \node [npattern] (y) [right of=xaxb] {\nodey};
      \node [nset] (z) [below of=xaxb, short] {\nodez};
      
      \draw [epattern] (xaxb) to node {} (y);
      
      \draw [eset] (y) to [bend right=50] node {} (xaxb);
      \draw [eset] (z) to node {} (xaxb);
      \draw [eset,loop=180,looseness=3] (z) to node {} (z);
      \draw [eset] (y) to node {} (z);
    }
    \graphbox{$K'$}{36mm}{-11mm}{35mm}{22mm}{-8mm}{-7mm}{
      \node [npattern] (xa) {\nodexa };
      \node [npattern] (xb) [right of=xa, short] {\nodexb};
      \node [npattern] (y) [right of=xb, short] {\nodey};
      \node [nset] (z) [below of=xa,short] {\nodez};
      
      \draw [eset] (y) to [bend right=40] node {} (xb);
      \draw [eset] (y) to [bend right=60] node {} (xb);
      \draw [eset,loop=180,looseness=3] (z) to node {} (z);
      \draw [eset] (z) to node {} (xa);
    }
    \transparentgraphbox{$R'$}{72mm}{-11mm}{42mm}{22mm}{-8mm}{-7mm}{
      \node [npattern] (xay)
      {\nodexa \ \nodey};
      \node [npattern] (xb) [right of=xay, xshift=3mm] {\nodexb};
      \draw [epattern,loop=0,looseness=3] (xb) to node {} (xb);
      \node [nset] (z) [below of=xay, short] {\nodez};
      
      \draw [eset] (xay) to [bend left=20] node {} (xb);
      \draw [eset] (xay) to [bend left=40] node {} (xb);
      \draw [eset,loop=180,looseness=3] (z) to node {} (z);
      
      \draw [eset] (z) to node {} (xay);
      \node [npattern] (u) [below right of=xb] {\nodeu};
    }
  \end{tikzpicture}
  }
\end{center}

}% scope
    \noindent
    In this and subsequent examples
    a vertex is a  non-empty set $\{ x_1,\ldots,x_n \}$ represented by a box
    {
    \hspace{-3.4mm}
    \newcommand{\nodexa}{\vertex{x_1}{cblue!20}}
    \newcommand{\nodexb}{\vertex{x_n}{cblue!20}}
    %\hspace{-3.2mm}
    \begin{tikzcd}[->,node distance=12mm,n/.style={}]
    \node [npattern] (xaxb)
          {\nodexa \ \raisebox{1mm}{$\cdots$} \  \nodexb};
    \end{tikzcd}
    }; and
    each morphism $\phi = (\phi_V, \phi_E) : G \to G'$ is the unique morphism satisfying $S \subseteq \phi(S)$ for all $S \in V_G$.
    For instance, for $\{ x_1 \}, \{ x_2 \} \in V_K$,  $l(\{x_1\}) = l(\{ x_2 \}) = \{ x_1,x_2 \} \in V_L$. For notational convenience, we will usually use examples that ensure uniqueness of 
    each $\phi$ (in particular, we ensure that $\phi_E$ is uniquely determined). Colors are purely supplementary, and elements in $L'$ and $K'$ have reduced opacity if they do not lie in the images of $t_L$ and $t_K$, respectively.
\end{example}

\newcommand{\picturematch}{
  \begin{tikzpicture}[node distance=13mm,l/.style={inner sep=1mm},baseline=-7.5mm,v/.style={node distance=11.5mm}]
  % match
    \node (G) {$G_L$};
    \node (L) [left of=G] {$L$};
      \draw [->] (L) to node [above,l] {$m$} (G);
    \node (L2) [below of=L,v] {$L$};
      \draw [>->] (L) to node [left,l] {$1_L$} (L2);
      \node at ([shift={(4mm,-4mm)}]L) {PB};
    \node (LP) [below of=G,v] {$L'$};
      \draw [->] (L2) to node [above,l] {$t_L$} (LP);
      \draw [->] (G) to node [left,l] {$\alpha$} (LP);
  \end{tikzpicture}
}

\begin{definition}[Strong Match]\label{def:pbpostrong:match}
    A \emph{match morphism} $m : L \to G_L$ and an \emph{adherence morphism} $\alpha : G_L \to L'$ form a \emph{strong match} for a context typing $t_L : L \to L'$ if square $t_L \circ \id{L} = \alpha \circ m$, i.e.,
    \begin{center}
        \begin{tikzcd}[column sep=7mm]
        L \arrow[d, equals] \arrow[r, "m" description] \arrow[rd, "\mathrm{PB}", phantom] & G_L \arrow[d, "\alpha" description] \\
        L \arrow[r, "t_L" description]                                                                 & L'                                 
        \end{tikzcd}%
    \end{center}
    is a pullback.
\end{definition}

\begin{remark}[Preimage Interpretation]
    \label{remark:preimage:interpretation}
    In both \Set and \Graph, and many other categories of structured sets, the strong match diagram of Definition~\ref{def:pbpostrong:match} states that the preimage of $t_L(L)$ under $\alpha : G_L \to L'$ is $L$ itself. So each element of $t_L(L)$ is the $\alpha$-image of exactly one element of $G_L$.
\end{remark}

\begin{proposition}
\label{prop:m:monic:iff:t_L:monic}
    Assume $\CC$ has pullbacks.
    Let a strong match as in Definition~\ref{def:pbpostrong:match} be given. Then $m$ is monic iff $t_L$ is monic.
\end{proposition}
\begin{proof}
    If $t_L$ is monic, monicity of $m$ follows by pullback stability.

    For the other direction, assume $m$ is monic.
    Suppose $t_L \circ x = t_L \circ y$ for a parallel pair of morphisms $x,y : X \to L$. Then
    $\alpha \circ m \circ x = t_L \circ x = t_L \circ y = \alpha \circ m \circ y$. Then in diagram
    \begin{center}
        % https://tikzcd.yichuanshen.de/#N4Igdg9gJgpgziAXAbVABwnAlgFyxMJZAJgBoBGAXVJADcBDAGwFcYkQAZEAX1PU1z5CKAMwVqdJq3YcA5Dz4gM2PASJkADBIYs2iTgv4qhRMVpo7p+gOIB9LryOC1KcqXOTd7ABqGlA1WFkNyoLKT0QX0d-YxdkDXdtcJ8eCRgoAHN4IlAAMwAnCABbJASQHAgkN08rcvsQGkZ6ACMYRgAFAJN9LDBsWD8C4qQycsrEMXL6LEZ2HGnZsK99EsaWts7Y4RBe-rZooZKJmgqqpdqAHQumNAALegaQJtaOrpcdvqwBg8Kj0dPEGVGL0IlAIMxmow2DRbjB6FB2JAwNCpjNEQR9opDiMTuNqpYIlcivQcLd8kVgO0AELcR7PDZvba7L4o4HI9HsmFwhH6JEo+Zo3kYwa-JAAFlxOJqEQAHnT1q8tuxmd8saLEABWSWA86y+UvTbOJmfVV5dUSsZILVPEHsMEQqGPWHwjn8haukXDRAWgFlAnsIkkskU6m0tYGxnKk2s21CzkgZ088DCk7uuOYs1egBs2utbNB4MhKMTHtTguT7J+2e1k1aYB5Ij9yRWAAIrgBjLD5dstgCe+oZSp60c9Rxzlu9jVjIDBOBw6Ue-v0AEIAF4DxVGqN7VLcIA
        \begin{tikzcd}
        X \arrow[rd, equals] \arrow[rrr, "m \circ y" description, bend left] \arrow[r, "!z" description, dotted] & X \arrow[r, "x" description] \arrow[d, equals] \arrow[rd, "\mathrm{PB}", phantom] & L \arrow[r, "m" description, tail] \arrow[d, equals] \arrow[rd, "\mathrm{PB}", phantom] & G_L \arrow[d, "\alpha" description] \\
                                                                                                                              & X \arrow[r, "x" description]                                                                   & L \arrow[r, "t_L" description]                                                                       & L'                                 
        \end{tikzcd}
    \end{center}
    the two pullback squares compose by the pullback lemma. Hence there exist a unique $z$ such that $\id{X} = \id{X} \circ z = z$ and $m \circ y = m \circ x \circ z$. Hence $z$ can be canceled and $m \circ y = m \circ x$. By monicity of $m$, $x = y$. Thus $t_L$ is monic.
\qed
\end{proof}

Because of Proposition~\ref{prop:m:monic:iff:t_L:monic}, if match morphisms $m$ are required to be monic (as often is the case), rules with non-monic $t_L$ will not give rise to strong matches, and so will not give rise to rewrite steps.

\begin{proposition}[Unique First Factor]
\label{prop:unique:first:factor}
    In any category, if diagrams
    \begin{center}
        % https://tikzcd.yichuanshen.de/#N4Igdg9gJgpgziAXAbVABwnAlgFyxMJZABgBoBGAXVJADcBDAGwFcYkQBBEAX1PU1z5CKchWp0mrdgGEefEBmx4CRUcXEMWbRCABCc-kqFEy6mpqk6u3cTCgBzeEVAAzAE4QAtkjIgcEJFEJLXYXEBpGegAjGEYABQFlYRAsMGxYAxB3LyQAJhp-QPNJbRB7cJBImPjE4x1U9LZeVw9vRABmAoDEX0ZU0qgIZijGNhoACxh6KHZIMDG-eixGWYIm+Wy2zr9u-ODLEHGKqtiEoxV6tKwM5qzWpG3CxCCLUoAdN896HHG3T2A4rpuMdoqdahcUlcMhF+qt5hVJtM4QscEsVjo5k1KNwgA
        \begin{tikzcd}
            A \arrow[d, equals] \arrow[r, "h" description] \arrow[rd, "\mathrm{PB}", phantom] & B \arrow[d, "g" description] \\
            A \arrow[r, "f" description]                                                                   & C                           
        \end{tikzcd}
        \quad and \quad
        % https://tikzcd.yichuanshen.de/#N4Igdg9gJgpgziAXAbVABwnAlgFyxMJZABgBoBGAXVJADcBDAGwFcYkQBBEAX1PU1z5CKchWp0mrdgGEefEBmx4CRUcXEMWbRCABCc-kqFEy6mpqk6u3cTCgBzeEVAAzAE4QAtkjIgcEJFEJLXYXEBpGegAjGEYABQFlYRAsMGxYAxB3LyQAJhp-QPNJbRB7cJBImPjE4x1U9LZeVw9vRABmAoDEX0ZU0qgIZijGNhoACxh6KHZIMDG-eixGWYIm+Wy2zr9u-ODLEHGAcgqq2ISjFXq0rAzmrNakbcLEIItSgF5T6PPaq5SbhkIv1VvMKpNpqCFjglisdHMmpRuEA
        \begin{tikzcd}
        A \arrow[d, equals] \arrow[r, "h'" description] \arrow[rd, "=", phantom] & B \arrow[d, "g" description] \\
        A \arrow[r, "f" description]                                                          & C                           
        \end{tikzcd}    
    \end{center}
    hold then $h = h'$.
\end{proposition}
\begin{proof}
    The universal morphism obtained from the pullback square is $\id{A}$. \qed
\end{proof}

Proposition~\ref{prop:unique:first:factor} implies that because of the strong match property, $m$ is uniquely defined by $\alpha$ and $t_L$. However, in practice it is usually more natural to first fix a match $m$, and to subsequently verify whether it can be extended into a suitable adherence morphism. For certain choices of $t_L$, $\alpha$ may moreover be uniquely determined by $m$ (if it exists). See Theorem~\ref{thm:determinism} for an example.

\newcommand{\pictureruletikz}{
  \begin{tikzpicture}[node distance=11mm,l/.style={inner sep=.5mm},baseline=-6mm]
    \node (L) {$L$};
    \node (K) [right of=L] {$K$}; \draw [->] (K) to node [above] {$l$} (L);
    \node (LP) [below of=L] {$L'$};
      \draw [->] (L) to node [left] {$t_L$} (LP);
    \node (KP) [below of=K] {$K'$};
      \draw [->] (K) to node [right] {$t_K$} (KP);
      \draw [->] (KP) to node [below] {$l'$} (LP);
       \node at ([shift={(-4mm,-4mm)}]K) {\PB};
    \node (R) [right of=K] {$R$}; 
      \draw [->] (K) to node [above] {$r$} (R);
  \end{tikzpicture}
}

\newcommand{\picturerule}{
      % https://tikzcd.yichuanshen.de/#N4Igdg9gJgpgziAXAbVABwnAlgFyxMJZABgBpiBdUkANwEMAbAVxiRABkQBfU9TXfIRRkAjFVqMWbdgHJuvEBmx4CREeXH1mrRCADS8vssFrSY6lqm69cnkYGqUAJg0XJOkACVu4mFADm8ESgAGYAThAAtkhkIDgQSOoS2mw4APqc1Ax0AEYwDAAK-CpCIFhg2LCGIOFRSADM1PGJbim6DHJZuflFxo5lFVhVdjUR0YgucQmISZYeADrzkXQ4ABZhkcAFAEJcIF15hcUmuuWVrFnlHpBgFyCrMHRQbDd3OHRYDC8ErCO145NmohGskrHE0gYDj1jv0zkNfgp-khAdNYnM2F8oUc+qU4cNEWNkU1pgAWVpgsL7EDZQ69By4wbDChcIA
        \begin{tikzcd}[ampersand replacement=\&]
        L \arrow[d, "t_L" description] \& K \arrow[ld, "\mathrm{PB}", phantom] \arrow[d, "t_K" description] \arrow[l, "l" description] \arrow[r, "r" description] \& R \\
        L'                             \& K' \arrow[l, "l'" description]                                                                                          \&  
        \end{tikzcd}
}

\newcommand{\picturesteptikz}{
  \begin{tikzpicture}[node distance=18mm,l/.style={inner sep=1mm},baseline=2mm,v/.style={node distance=11mm}]
  % match
    \node (G) {$G_L$};
    \node (L) [left of=G] {$L$};
      \draw [->] (L) to node [above,l] {$m$} (G);
    \node (L2) [below of=L,v] {$L$};
      \draw [>->] (L) to node [left,l] {$1_L$} (L2);
      \node at ([shift={(4mm,-4mm)}]L) {\PB};
    \node (LP) [below of=G,v] {$L'$};
      \draw [->] (L2) to node [below,l] {$t_L$} (LP);
      \draw [->] (G) to node [left,l] {$\alpha$} (LP);
     % duplication
     \node (GKP) [right of=G] {$G_K$};
     \draw [->] (GKP) to node [above,l] {$g_L$} (G);
     \node (KP) [below of=GKP,v] {$K'$};
     \draw [->] (GKP) to node [right,l] {$u'$} (KP);
     \draw [->] (KP) to node [below,l] {$l'$} (LP);
     \node at ([shift={(-4mm,-4mm)}]GKP) {\PB};
     % the PO
     \node (K) [above of=GKP,v] {$K$};
     \draw [->, dotted] (K) to node [left,l] {$!u$} (GKP);
     \node (R) [right of=K] {$R$};
     \draw [->] (K) to node [above,l] {$r$} (R);
     \node (GP) [below of=R,v] {$G_R$};
     \draw [->] (GKP) to node [below,l,pos=0.7] {$g_R$} (GP);
     \draw [->] (R) to node [right,l] {$w$} (GP);
     \node at ([shift={(-4mm,4mm)}]GP) {\PO};
     
     % adding t_K
     \draw [-,draw=white,line width=1.5mm] (K) to[out=-30,in=30,looseness=0.8] (KP);
     \draw [->,draw=black] (K) to[out=-30,in=30,looseness=0.8] node[right,l,pos=0.85] {$t_K$} (KP);
  \end{tikzpicture}
}

\newcommand{\picturestep}{
% https://tikzcd.yichuanshen.de/#N4Igdg9gJgpgziAXAbVABwnAlgFyxMJZABgBoBGAXVJADcBDAGwFcYkQAZEAX1PU1z5CKMgCZqdJq3Zde-bHgJFyFCQxZtEIAOIB9WXxAYFQoqNU110rXoDSPQ8cFKUAZguSN7PQCUH852Fkc2I1KU0QezkjAUUg91DLcPY-aKc4s1JxJK8tWwByfxiTF2QVbM9rTkLuCRgoAHN4IlAAMwAnCABbJHMQHAgkFUqIgB1RrvocAAt2ruAABQAhbhAaRnoAIxhGBdjTLSwwbFg1kEYjiMgwNhppmHoodmvb-vosRmeCNmiO7qQyP1BohhhcbuwoBBmJtGK97o8vuCaDh3p8tC8zhttrt9i5zjBWjgin8eohAQNeustjs9iVhCAjidXlYIj1fp1ScMKYgABxU7G0wLsRlYU45Ko4fTEjmUoFIPnnak4unC46i5nJLTjJhoab0aX-RDuOWIPpYmm4+kisUjdgNKXsw3G7kAFn5FpVhzVNpZdt0qUMJKQAFZkcDjeblUKvUyzr6tABCZiYy4QiA4HD1A2k53AgDs7qjGRj6rjmpAzBqgZliALJoVkcFxYZ3o1uXOVbaNdDJoAbOKIu1MUqmwcW7HHaT+ya3YqBZbVbGB+wAO7Z+VhpAR1Po75neFPXdIt4fRGvRsLks+8vjSYzOaLFbro2bxDTsFXPd3B6H8Bfk9on+x4Xp646lsuWoTFMszzAsADyqyTiGr51tsYCHq4gIgdGYHXu2kpRJQ3BAA
\begin{tikzcd}[ampersand replacement=\&,baseline=0mm, column sep=12mm]
                                                            \&                                                                        \& K \arrow[d, "!u" description, dotted] \arrow[r, "r" description]                    \& R \arrow[d, "w" description] \\
L \arrow[d, equals] \arrow[r, "m" description] \& G_L \arrow[ld, "\mathrm{PB}", phantom] \arrow[d, "\alpha" description] \& G_K \arrow[l, "g_L" description] \arrow[r, "g_R" description] \arrow[d, "u'" description] \arrow[ru, "\mathrm{PO}", phantom] \& G_R                          \\
L \arrow[r, "t_L" description]                              \& L' \arrow[ru, "\mathrm{PB}", phantom]                                  \& K' \arrow[from=uu, crossing over, pos=0.75, "t_K" description, bend left=42] \arrow[l, "l'" description]                                                                                               \&                             
\end{tikzcd}
}

In reading the following definition, it may be helpful to refer to Example~\ref{ex:simple:rewrite:step} alongside it.

\begin{definition}[\pbpostrong Rewrite Step]
\label{def:pbpostrong:rewrite:step}
    A \pbpostrong rewrite rule $\rho$ (left), a match morphism $m : L \to G_L$
    and an adherence morphism $\alpha : G_L \to L'$ induce a rewrite step $\pbpostrongstep{G_L}{G_R}{\rho}{m}{\alpha}$
    on arbitrary $G_L$ and $G_R$ if the properties indicated by the commuting diagram (right)
    \begin{center}
      %\raisebox{9mm}{$\rho \ = \ $\picturerule} \hspace{.7cm}
      $\rho \ = \ $\picturerule\hspace{.7cm}
      \picturestep
    \end{center}
    hold, where $u : K \to G_K$ is the unique morphism satisfying $t_K = u' \circ u$ (Lemma~\ref{lemma:on:u}).
    
    We write $G_L \pbpostrongrulematchonly{\rho}{m} G_R$ if there exists an $\alpha$ such that $\pbpostrongstep{G_L}{G_R}{\rho}{m}{\alpha}$, and
    $\pbpostrongsteprule{G_L}{G_R}{\rho}$ if there exists an $m$ such that $G_L \pbpostrongrulematchonly{\rho}{m} G_R$.~\footnote{No assumptions are needed on the underlying category $\CC$: if not all (co)limits exist, this simply restricts the possible rewrite rules and steps.}
\end{definition}

It can be seen that the rewrite step diagram consists of a match square, a pullback square for extracting (and possibly duplicating) parts of $G_L$, and finally a pushout square for gluing these parts along pattern $R$.

We must prove that there indeed exists a unique $u$ such that $t_K = u' u$ in the rewrite step diagram. The following lemma establishes this and two other facts.

\begin{lemma}[On $u$]
\label{lemma:on:u}
    In any category, let the pullback of the rewrite rule and the pullbacks of the rewrite step be given. Then there exists a unique morphism $u : K \to G_K$ satisfying $t_K = {u' \circ u}$.
    Moreover, the following properties hold:
        \begin{enumerate}
            \item squares
                \begin{center}
                    % https://tikzcd.yichuanshen.de/#N4Igdg9gJgpgziAXAbVABwnAlgFyxMJZABgBpiBdUkANwEMAbAVxiRABkQBfU9TXfIRRkAjFVqMWbAOIB9Tjz7Y8BIiNJjq9Zq0Qg5AaW68QGZYLXlx2qXqOLT-FUOQBmK1sm6Q9k2YGqKO6aEjpsvkoBLgAsGtZe4QDkxpHORLGUnmF6himO5oHIAEweobYgef5pKCUhNt7c4jBQAObwRKAAZgBOEAC2SGQgOBBI6mXeA9QMdABGMAwACk4Welhg2LB5Pf1IJcOjiOP1bC3yINNzC8sFQiDrm6wOOwOI7geDWeUMFyAz80sVoF7hssFtnr1Xu8Rnsvt4mL9-tcgXcHmCniYXkhYh9EABWagACxgdCgbEgYFY03W3igECYswYVOGdCwPz0FIxXUhSAJuIAbHC2DhZEZLgCblE2GjwZieYgcTDEAB2IV6BHi5G3aWg2Xc3Yq6hKwUTNhMZKawHata6rkgLGIAAcRsOAE5LZLqiDHoiaeSCMziaT-ZTfjhWezwAHfic9AAdOM4GAADzwOGAdDAUC4jS4QA
                    \begin{tikzcd}
                    L \arrow[d, "m" description] & K \arrow[l, "l" description] \arrow[d, "u" description] & {} \arrow[d, "\textit{and}", phantom] & K \arrow[d, equals] \arrow[r, "u" description] & G_K \arrow[d, "u'" description] \\
                    G_L                          & G_K \arrow[l, "g_L" description]                        & {}                                    & K \arrow[r, "t_K" description]                              & K'                             
                    \end{tikzcd}
                \end{center}
            are pullbacks; and
            \item if $m$ or $t_K$ is monic, then $u$ is monic.
        \end{enumerate}
\end{lemma}
\begin{proof}
    Consider the commuting diagram
    \begin{center}
        % https://tikzcd.yichuanshen.de/#N4Igdg9gJgpgziAXAbVABwnAlgFyxMJZABgBoBGAXVJADcBDAGwFcYkQBxAfQBkQBfUuky58hFGQBM1Ok1bseAcgFCQGbHgJFypaTQYs2iEAGllg4RrHaKMg-OPcTKy6K0odxO3KOmXakU1xElIvfR8FARkYKABzeCJQADMAJwgAWyQAZhocCCQyWUN2WN4QGkZ6ACMYRgAFQOtjLDBsWH9UjILc-MQdIocQAB0hpjQAC3pykEqa+sb3EBa2tgsQTszEHJA8pElw4uNmZQrq2oarReWsdrWNvZ6kfvtfRhOZs-nL8SXWm9XVPdEAAWR6IfYgGpgKBIYEATgOgxwXGcuXoWEY7Bw6MxpzmFzcP2ut0BaU2hV24MRvhGOBgAA8cMA6gAhfjTRgtXyQMBsGjjGD0GHGHl8nY49iijmfAlBdjEgHJMmwsEAVmp7FxH3xCyJfxJSq6iHVO16hRe7EyaIxWIlePOuvl+sV62VxrB-ShwoAtPCNcZkXx7V9CU6VtNsTaAxK7m7QabshUuewoBAcHSYf6QABCZjSnXfMP-CMS6M2-iUfhAA
        \begin{tikzcd}
        L \arrow[d, "m" description] \arrow[dd, "t_L" description, bend right=49] & K \arrow[dd, "t_K" description, bend left=49] \arrow[l, "l" description] \arrow[d, "!u" description, dotted] \\
        G_L \arrow[d, "\alpha" description] \arrow[rd, "\text{PB}", phantom]                  & G_K \arrow[l, "g_L" description] \arrow[d, "u'" description]                                                             \\
        L'                                                                                    & K' \arrow[l, "l'" description]                                                                                          
        \end{tikzcd}
    \end{center}
    where morphism $u$ satisfying $t_K = {u' \circ u}$ and ${m \circ l} = {g_L \circ u}$ is inferred by the universal property of the bottom square. Because the outer square is the pullback of the rewrite rule, ${m \circ l} = {g_L \circ u}$ is a pullback by the pullback lemma. Moreover, this shows that monicity of $u$ follows from monicity of $t_K$ ($g \circ f$ is monic $\Rightarrow$ $f$ is monic) or monicity of $m$ and pullback stability.
    
    The accumulated squares can be represented as the commutative cube
    \begin{center}
        \begin{tikzcd}[row sep={30,between origins},column sep={30,between origins}]
            & K  \ar[rr, "u" description] \ar[dd, equals] \ar[dl, "l" description] & &  G_K \ar[dd, "u'" description] \ar[dl, "g_L" description] \\
            L \ar[rr, crossing over, "m" description, pos=0.25] \ar[dd, equals] & & G_L \\
              & K \ar[rr, "t_K" description, pos=0.27] \ar[dl, "l" description, hook'] & &  K' \ar[dl, "l'" description] \\
            L \ar[rr, "t_L" description] && L' \ar[from=uu,crossing over, "\alpha" description, pos=0.75]
         \end{tikzcd} .
    \end{center}
    which shows that square $t_K \circ \id{K} = t_K = u \circ u'$ is a pullback square by Corollary~\ref{lem:pullback:cube}. Finally, because it is a pullback square, it follows from Proposition~\ref{prop:unique:first:factor} that for any $v$ with $t_K = u' \circ v$, $v = u$.
    \qed
\end{proof}

\begin{proposition}[Bottom-Right Pushout]
\label{prop:bottomright:pushout}
    Let cospan $K' \xrightarrow{r'} R' \xleftarrow{t_R} R$ be a pushout for span $R \xleftarrow{r} K \xrightarrow{t_K} K'$ of rule $\rho$ in Definition~\ref{def:pbpostrong:rewrite:step}. Then in the rewrite step diagram, there exists a morphism $w' : G_R \to R'$ such that
    $t_R = w' \circ w$, and
    $K' \xrightarrow{r'} R' \xleftarrow{w'} G_R$ is a pushout for $K' \xleftarrow{u'} G_K \xrightarrow{g_R} G_R$.
\end{proposition}
\begin{proof}
    The argument is similar to the initial part of the proof of Lemma~\ref{lemma:on:u}, but now uses the dual statement of the pullback lemma. \qed
\end{proof}

Lemma~\ref{lemma:on:u} and Proposition~\ref{prop:bottomright:pushout} show that a \pbpostrong step defines a commuting diagram similar to the PBPO definition (Definition~\ref{def:pbpo:rewrite:step}):
\begin{center}
    % https://tikzcd.yichuanshen.de/#N4Igdg9gJgpgziAXAbVABwnAlgFyxMJZARgBoAGAXVJADcBDAGwFcYkQAZEAX1PU1z5CKMsWp0mrdgHEA+l179seAkTIAmcQxZtEnAOQ8+IDMqFF1pMTW1S9cgNJGlg1Skuabk3SAeHFJgIqwsiWVF467E4Bpq4hAMwUWt7sAErOgWZuyInWEpH2sukxQeYoiZ75diCp-saxwUTkVskFnBkNZcjNlbY+CuIwUADm8ESgAGYAThAAtkjNIDgQSGRVPvM0jPQARjCMAAqlbiBYYNiwGdNzqzTLSJbr7AA6z0xoABb0IFu7+0dZYSnc5YS4Ba7zRCJJYrRAAFgi1WYhl+e0OxyBZwubHBM0hAFY7rDoX12MwfiBtmiAXF2FjQTjjBCkAA2IlIADsiJ8AHcKVT-hi6SCwUy8Zz2YgABzc9g8lGUv7owHC7FXcVQyVrUl6YbyflKmmNPT00WTDXQ+6ILlPXVFA3UoUmkWM803eGSx46ykKgXK2nOtW490ImFIGW2kBTX2Gp3AoNi92EsOINmRqYOwUqwMM9VJyWLb2MTP+43x3PByEATi1sr0m0Vjuz5bNIGZiGIiytXpSehw+tRWYDLddbY1NZTna2Zx8UAgzB2jDYNA+MHoUHYkDAy6W9Cwxb0W9H7Yn3brIFes3oOA+U1mwAOACFuCWjWUR-yZ5uCDvV+vv9uFI4HuB7gD+eYErWkaXtet73k+L6DqW76mjujBfoe4ErmuG6YYBdwgQBx4Wp654wTed4Ps+r5xqhn6AXhv44URQGEYxEFIMmVo2t65FwQ+ADyiGNkOZZ0dODFgfhIB-rhUk7sB+4sZWSCWrCEa8c8V4UfBQk0c24mUhh8kUrJLEEUp7EqamkoRnsYC4fEha9ks9pIW+Jx0dZXawo8foeZiLoUt6-ZcDQ9mOeQ1lcbCoYRapzltP2TjubRQXcJQ3BAA
    \begin{tikzcd}[column sep=80]
                                                                                                   & \color{gray}{L} \arrow[d, "m" description, gray]  & K \arrow[d, "u" description] \arrow[r, "r" description] \arrow[l, "l" description, gray] \arrow[ld, "\mathrm{PB}", phantom, gray] \arrow[rd, "\mathrm{PO}", phantom] & R \arrow[d, "w" description] \arrow[dd, "t_R" description, bend left=55, pos=0.7, gray] \\
    L \arrow[r, "m" description] \arrow[d, equals] \arrow[rd, "\mathrm{PB}", phantom] & G_L \arrow[d, "\alpha" description]                                   & G_K \arrow[d, "u'" description] \arrow[l, "g_L" description] \arrow[r, "g_R" description] \arrow[ld, "\mathrm{PB}", phantom] \arrow[rd, "\mathrm{PO}", phantom, gray]                                   & G_R \arrow[d, "w'" description, gray]                                       \\
    L \arrow[r, "t_L" description]                                                                 & L'  \arrow[from=uu, "t_L" description, bend left=55, crossing over, pos=0.7, gray]                                                                  & K' \arrow[l, "l'" description] \arrow[r, "r'" description, gray] \arrow[from=uu, "t_K" description, bend left=55, crossing over, pos=0.7]                                                                                                                               & \color{gray}{R'}                                                            
    \end{tikzcd}
\end{center}
which will guide our depictions of rewrite steps. We will omit the match diagram in such depictions.

\begin{example}[Rewrite Step]
\label{ex:simple:rewrite:step}
    Applying the rule given in Example~\ref{ex:simple:rewrite:rule} to $G_L$ (as depicted below) has the following effect:
    {% scope

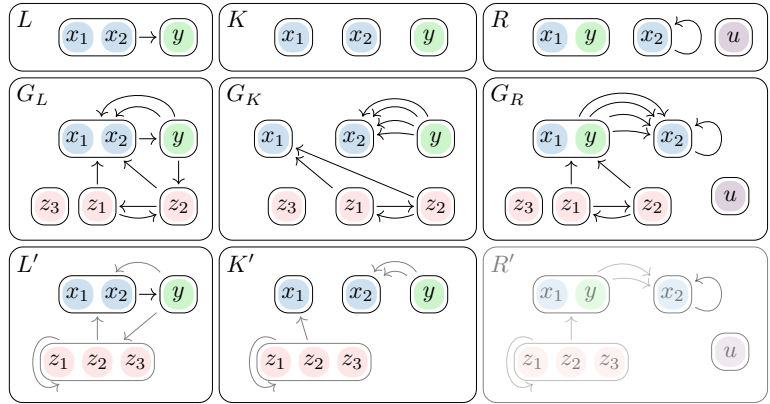
\newcommand{\nodexa}{\vertex{x_1}{cblue!20}}%
\newcommand{\nodexb}{\vertex{x_2}{cblue!20}}%
\newcommand{\nodey}{\vertex{y}{cgreen!20}}%
\newcommand{\nodez}{\vertex{z}{cred!10}}%
\newcommand{\nodeza}{\vertex{z_1}{cred!10}}%
\newcommand{\nodezb}{\vertex{z_2}{cred!10}}%
\newcommand{\nodezc}{\vertex{z_3}{cred!10}}%
\newcommand{\nodeu}{\vertex{u}{cpurple!25}}%
\begin{center}\vspace{-.25ex}
  \scalebox{\rulescale}{
  \begin{tikzpicture}[->,node distance=12mm,n/.style={}]
      \graphbox{$L$}{0mm}{0mm}{30mm}{10mm}{-2mm}{-5mm}{
      \node [npattern] (xaxb)
      {\nodexa \ \nodexb};
      \node [npattern] (y) [right of=xaxb] {\nodey};
      \draw [epattern] (xaxb) to node {} (y);
    }
    \graphbox{$G_L$}{0mm}{-11mm}{30mm}{24mm}{-2mm}{-9mm}{
      \node [npattern] (xaxb)
      {\nodexa \ \nodexb};
      \node [npattern] (y) [right of=xaxb] {\nodey};
      \draw [epattern] (xaxb) to node {} (y);
      
      % extension to the pattern
      \node [npattern] (za) [below of=xaxb, short] {\nodeza};
      \node [npattern] (zb) [right of=za] {\nodezb};
      \node [npattern] (zc) [left of=za,node distance=7mm] {\nodezc};
      
      \draw [epattern] (za) to node {} (xaxb);
      \draw [epattern] (zb) to node {} (za);
      \draw [epattern] (zb) to node {} (xaxb);
      \draw [epattern] (za) to [bend right=20] node {} (zb);
      
      \draw [epattern] (y) to [bend right=50] node {} (xaxb);
      \draw [epattern] (y) to [bend right=80] node {} (xaxb);
      \draw [epattern] (y) to (zb);
    }
    \graphbox{$G_K$}{31mm}{-11mm}{38mm}{24mm}{-2mm}{-9mm}{
      \node [npattern, xshift=3mm] (xaxb)
      {\nodexb};
      \node [npattern] (xa) [left of=xaxb]
      {\nodexa};
      \node [npattern] (y) [right of=xaxb] {\nodey};
      
      % extension to the pattern
      \node [npattern] (za) [below of=xaxb, short] {\nodeza};
      \node [npattern] (zb) [right of=za] {\nodezb};
      \node [npattern] (zc) [left of=za,node distance=10mm] {\nodezc};
      
      \draw [epattern] (zb) to [bend left=20] node {} (za);
      \draw [epattern] (za) to node {} (zb);
      
      \draw [epattern] (y) to [bend right=50] node {} (xaxb);
      \draw [epattern] (y) to [bend right=70] node {} (xaxb);
      \draw [epattern] (y) to [bend right=30] node {} (xaxb);
      \draw [epattern] (y) to [bend right=10] node {} (xaxb);
      \draw [epattern] (za) to node {} (xa);
      \draw [epattern] (zb) to node {} (xa);
      
    }
    \graphbox{$K$}{31mm}{0mm}{38mm}{10mm}{-8mm}{-5mm}{
      \node [npattern] (xa) {\nodexa};
      \node [npattern] (xb) [right of=xa, short] {\nodexb};
      \node [npattern] (y) [right of=xb, short] {\nodey};
    }
    
    \graphbox{$R$}{70mm}{0mm}{42mm}{10mm}{-8mm}{-5mm}{
      \node [npattern] (xay)
      {\nodexa \ \nodey};
      \node [npattern] (xb) [right of=xay] {\nodexb};
      \draw [epattern,loop=0,looseness=3] (xb) to node {} (xb);
      
      \node [npattern] (u) [right of=xb] {\nodeu};
    }
    
    \graphbox{$G_R$}{70mm}{-11mm}{42mm}{24mm}{-8mm}{-9mm}{
      \node [npattern] (xaxb)
      {\nodexa \ \nodey};
      \node [npattern] (y) [right of=xaxb, xshift=3mm] {\nodexb};
      
      % extension to the pattern
      \node [npattern] (za) [below of=xaxb, short] {\nodeza};
      \node [npattern] (zb) [right of=za] {\nodezb};
      \node [npattern] (zc) [left of=za,node distance=7mm] {\nodezc};

      \node [npattern] (u) [below right of=y] {\nodeu};
      
      \draw [epattern] (zb) to [bend left=20] node {} (za);
      \draw [epattern] (za) to node {} (zb);
      
      \draw [epattern] (xaxb) to [bend left=10] node {} (y);
      \draw [epattern] (xaxb) to [bend left=30] node {} (y);
      \draw [epattern] (xaxb) to [bend left=50] node {} (y);
      \draw [epattern] (xaxb) to [bend left=70] node {} (y);
      \draw [epattern,loop=0,looseness=3] (y) to node {} (y);
      
      \draw [epattern] (za) to node {} (xaxb);
      \draw [epattern] (zb) to node {} (xaxb);
    }
    \graphbox{$L'$}{0mm}{-36mm}{30mm}{23mm}{-2mm}{-7mm}{
      \node [npattern] (xaxb)
      {\nodexa \ \nodexb};
      \node [npattern] (y) [right of=xaxb] {\nodey};
      \node [nset] (z) [below of=xaxb, short]
      {\nodeza \ \nodezb \ \nodezc};
      
      \draw [epattern] (xaxb) to node {} (y);
      
      \draw [eset] (y) to [bend right=50] node {} (xaxb);
      \draw [eset] (z) to node {} (xaxb);
      \draw [eset,loop=180,looseness=3] (z) to node {} (z);
      
      \draw[eset] (y) to node {} (z);
    }
    \graphbox{$K'$}{31mm}{-36mm}{38mm}{23mm}{-8mm}{-7mm}{
      \node [npattern] (xa) {\nodexa };
      \node [npattern] (xb) [right of=xa, short] {\nodexb};
      \node [npattern] (y) [right of=xb, short] {\nodey};
      \node [nset] (z) [below of=xa,short, xshift=3mm] {\nodeza \ \nodezb \ \nodezc};
      
      \draw [eset] (y) to [bend right=40] node {} (xb);
      \draw [eset] (y) to [bend right=60] node {} (xb);
      \draw [eset,loop=180,looseness=3] (z) to node {} (z);
      \draw [eset] (z) to node {} (xa);
    }
    \transparentgraphbox{$R'$}{70mm}{-36mm}{42mm}{23mm}{-8mm}{-7mm}{
      \node [npattern] (xay)
      {\nodexa \ \nodey};
      \node [npattern] (xb) [right of=xay, xshift=3mm] {\nodexb};
      \draw [epattern,loop=0,looseness=3] (xb) to node {} (xb);
      \node [nset] (z) [below of=xay, short] {\nodeza \ \nodezb \ \nodezc};
      
      \draw [eset] (xay) to [bend left=20] node {} (xb);
      \draw [eset] (xay) to [bend left=40] node {} (xb);
      \draw [eset,loop=180,looseness=3] (z) to node {} (z);
      
      \draw [eset] (z) to node {} (xay);
      \node [npattern] (u) [below right of=xb] {\nodeu};
    }
  \end{tikzpicture}
  }\vspace{-1ex}
\end{center}

}% scope
    This example illustrates (i)~how permitted patches can be constrained (e.g., $L'$ forbids patch edges \emph{targeting} $y$), (ii)~how patch edge endpoints that lie in the image of $t_L$ can be redefined, (iii)~how patch edges can be deleted, and (iv)~how patch edges can be duplicated.
\end{example}

In the examples of this section, we have restricted our attention to unlabeled graphs. In Section~\ref{sec:category:graph:lattice}, we introduce a new category \GraphLattice{(\labels,\leq)}, and show that is more suitable than \Graph for rewriting labeled graphs using \pbpostrong{}. Section~\ref{sec:category:graph:lattice} can largely be read independently of Sections~\ref{sec:relating} and~\ref{sec:quasitopos}, which situate \pbpostrong and are more foundational in character.

\section{Relating \pbpostrong and PBPO}
\label{sec:relating}

In Section~\ref{sec:relating:rule:match:step}, we recall and compare the PBPO definitions for rule, match and step, clarifying why \pbpostrong is shorthand for \emph{PBPO with strong matching}. We then argue why strong matching is usually desirable in Section~\ref{sec:relating:case:for:strong:matching}.

\subsection{PBPO: Rule, Match \& Step}
\label{sec:relating:rule:match:step}

\smallskip

\noindent
\begin{minipage}{0.68\textwidth}
    \begin{definition}[PBPO Rule~\cite{corradini2019pbpo}]
        \label{def:pbpo:rule}
        A \emph{PBPO rule} $\rho$ is a commutative diagram as shown on the right.
        The bottom span can be regarded as a typing for the top span. The rule is in \emph{canonical form} if the left square is a pullback and the right square is a pushout.
    \end{definition}
\end{minipage}\hfill%
\begin{minipage}{0.25\textwidth}
    % https://tikzcd.yichuanshen.de/#N4Igdg9gJgpgziAXAbVABwnAlgFyxMJZABgBpiBdUkANwEMAbAVxiRABkQBfU9TXfIRRkAjFVqMWbdgHJuvEBmx4CREaTHV6zVohABpOTz7LBa8uO1S9++SYGqUAJgtbJukACU7i-iqHILpoSOmyeRuIwUADm8ESgAGYAThAAtkhkIDgQSOoh1lkA+pzUDHQARjAMAAp+ZnpYYNiwPslpSC5ZOYh5Vh4McqUVVbWmjiCNzazGIG3piADM1NkZbqF6DCBDlTV145NYLTNzSEtdSAAsawVJWyBlO6MOQhNNh9MKJ4vL3Z19bDhCrZtiM9i8DkdPil5p0VogAKzXDxJQb3Ya7Mbgt6QxLQy4-JCI-IeQHeEEY55sCEfXHtb7nHpItgAXjuD1BmKp2NYpUaHkgYB5IAAFjA6FA2AKhTg6FhNnopa08fS4UT-npWeSnv4uVM2XzJQQhaLxYbBXcZXKzdMKFwgA
    \hspace*{-6.2mm}
    \begin{tikzcd}
    L \arrow[d, "t_L" description] & K \arrow[l, "l" description] \arrow[r, "r" description] \arrow[d, "t_K" description] \arrow[ld, "=", phantom] \arrow[rd, "=", phantom] & R \arrow[d, "t_R" description] \\
    L'                             & K' \arrow[l, "l'" description] \arrow[r, "r'" description]                                                                             & R'                            
    \end{tikzcd}
\end{minipage}
\medskip

Every PBPO rule is equivalent to a rule in canonical form~\cite{corradini2019pbpo}, and in \pbpostrong, rules are limited to those in canonical form. 

\begin{definition}[PBPO Match~\cite{corradini2019pbpo}]
    A \emph{PBPO match} for a typing $t_L : L \to L'$ is a pair of morphisms $(m : L \to G,\alpha : G \to L')$ such that $t_L = \alpha \circ m$.
\end{definition}

The pullback construction used to establish a match in \pbpostrong (Definition~\ref{def:pbpostrong:match}) implies $t_L = \alpha \circ m$. Thus PBPO matches are more general than the strong match used in \pbpostrong (Definition~\ref{def:pbpostrong:match}). More specifically for \Graph, PBPO allows mapping elements of the host graph $G_L$ not in the image of $m : L \to G_L$ onto the image of $t_L$, whereas \pbpostrong forbids this. In the next subsection, we will argue why it is often desirable to forbid such mappings.

\begin{definition}[PBPO Rewrite Step~\cite{corradini2019pbpo}]\label{def:pbpo:rewrite:step}
    A PBPO rule $\rho$ (as in Definition~\ref{def:pbpo:rule}) induces a \emph{PBPO step} $\pbpostep{G_L}{G_R}{\rho}{m}{\alpha}$ if there exists a diagram
    \begin{center}
        \begin{tikzcd}[column sep=80]
        L \arrow[d, "m" description]  & K \arrow[d, "u" description] \arrow[r, "r" description] \arrow[l, "l" description] \arrow[ld, "\mathrm{=}", phantom] \arrow[rd, "\mathrm{PO}", phantom] & R \arrow[d, "w" description] \arrow[dd, "t_R" description, bend left=55, pos=0.7] \\
         G_L \arrow[d, "\alpha" description]                                   & G_K \arrow[d, "u'" description] \arrow[l, "g_L" description] \arrow[r, "g_R" description] \arrow[ld, "\mathrm{PB}", phantom] \arrow[rd, "=", phantom]                                   & G_R \arrow[d, "w'" description]                                       \\
         L'  \arrow[from=uu, "t_L" description, bend left=55, crossing over, pos=0.7]                                                                  & K' \arrow[l, "l'" description] \arrow[r, "r'" description] \arrow[from=uu, "t_K" description, bend left=55, crossing over, pos=0.7]                                                                                                                               & R'                                                            
        \end{tikzcd}
    \end{center}
    where (i)~$u : K \to G_K$ is uniquely determined by the universal property of pullbacks and makes the top-left square commuting, and (ii)~$w' : G_R \to R'$ is uniquely determined by the universal property of pushouts and makes the bottom-right square commuting, and (iii)~$t_L = \alpha \circ m$.
\end{definition}

The strong match square of \pbpostrong{} allows simplifying the characterization of $u$, as shown in the proof to Lemma~\ref{lemma:on:u}. This simplification is not possible for PBPO (see Remark~\ref{remark:pbpo:u:not:unique}). The bottom-right square is omitted in the definition of a \pbpostrong{} rewrite step, but can be reconstructed through a pushout (modulo isomorphism). So this difference is not essential.
\\[-1.5ex] % this is a hack; for some reason there needs to be a line break or we get some mysterious error in Section 6

\begin{remark}
\label{remark:pbpo:u:not:unique}
    In a PBPO rewrite step, not every morphism $u : K \to G_K$ satisfying $u' \circ u = t_K$ corresponds to the arrow uniquely determined by the top-left pullback. This can be seen in the example of a (canonical) PBPO rewrite rule and step depicted in:
    \begin{center}
        {% scope

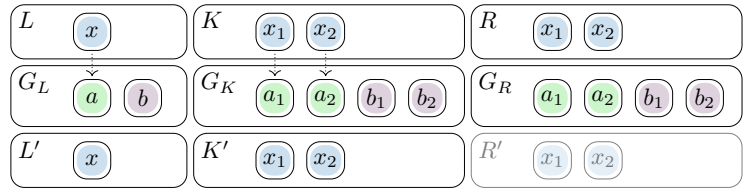
\newcommand{\nodexa}{\vertex{x_1}{cblue!20}}
\newcommand{\nodexb}{\vertex{x_2}{cblue!20}}
\newcommand{\nodexc}{\vertex{x_3}{cblue!20}}
\newcommand{\nodexd}{\vertex{x_4}{cblue!20}}

\newcommand{\nodea}{\vertex{a}{cgreen!20}}
\newcommand{\nodeb}{\vertex{b}{cpurple!25}}

\newcommand{\nodeaa}{\vertex{a_1}{cgreen!20}}
\newcommand{\nodeab}{\vertex{a_2}{cgreen!20}}
\newcommand{\nodeba}{\vertex{b_1}{cpurple!25}}
\newcommand{\nodebb}{\vertex{b_2}{cpurple!25}}

\newcommand{\nodex}{\vertex{x}{cblue!20}}

%\begin{center}
  \scalebox{\rulescale}{
  \begin{tikzpicture}[->,node distance=12mm,n/.style={}]
    \graphbox{$R$}{68mm}{0mm}{40mm}{8mm}{-8mm}{-4mm}{
      \node [npattern] (xa)
      {\nodexa};
      \node [npattern] (xb) [right of=xa, short,xshift=-2.5mm]
      {\nodexb};
    }
    \graphbox{$G_L$}{0mm}{-9mm}{26mm}{9mm}{-1mm}{-4.8mm}{
      \node [npattern] (a)
      {\nodea};
      \node [npattern] (b) [right of=a, short, xshift=-2.5mm] {\nodeb};
    }
      \graphbox{$L$}{0mm}{0mm}{26mm}{8mm}{-1mm}{-4mm}{
      \node [npattern] (x)
      {\nodex};
       \draw [epattern, densely dotted] (x) to node {} (a);
    }
    \graphbox{$G_K$}{27mm}{-9mm}{40mm}{9mm}{-8mm}{-4.8mm}{
      \node [npattern] (aa)
      {\nodeaa};
      \node [npattern] (ab) [right of=aa, short, xshift=-2.5mm] {\nodeab};
      \node [npattern] (ba) [right of=ab, short, xshift=-2.5mm]
      {\nodeba};
      \node [npattern] (bb) [right of=ba, short, xshift=-2.5mm] {\nodebb};
    }
    \graphbox{$K$}{27mm}{0mm}{40mm}{8mm}{-8mm}{-4mm}{
      \node [npattern] (xa)
      {\nodexa};
      \node [npattern] (xb) [right of=xa, short, xshift=-2.5mm]
      {\nodexb};
      \draw [epattern, densely dotted] (xa) to node {} (aa);
      \draw [epattern, densely dotted] (xb) to node {} (ab);
    }
    \graphbox{$G_R$}{68mm}{-9mm}{40mm}{9mm}{-8mm}{-4.8mm}{
      \node [npattern] (aa)
      {\nodeaa};
      \node [npattern] (ab) [right of=aa, short, xshift=-2.5mm] {\nodeab};
      \node [npattern] (ba) [right of=ab, short, xshift=-2.5mm]
      {\nodeba};
      \node [npattern] (bb) [right of=ba, short, xshift=-2.5mm] {\nodebb};
    }
    \graphbox{$L'$}{0mm}{-19mm}{26mm}{8mm}{-1mm}{-4mm}{
      \node [npattern] (x)
      {\nodex};
    }
    \graphbox{$K'$}{27mm}{-19mm}{40mm}{8mm}{-8mm}{-4mm}{
      \node [npattern] (xa)
      {\nodexa};
      \node [npattern] (xb) [right of=xa, short, xshift=-2.5mm]
      {\nodexb};
    }
    \transparentgraphbox{$R'$}{68mm}{-19mm}{40mm}{8mm}{-8mm}{-4mm}{
      \node [npattern] (xa)
      {\nodexa};
      \node [npattern] (xb) [right of=xa, short, xshift=-2.5mm]
      {\nodexb};
    }
  \end{tikzpicture}
  }
%\end{center}

}% scope
    \end{center}
    Because our previous notational convention breaks for this example, we indicate two morphisms by dotted arrows. The others can be inferred.
    
    Morphism $u : K \to G_K$ (as determined by the top-left pullback) is indicated. However, it can be seen that three other morphisms $v : K \to G_K$ satisfy $u' \circ v = t_K$, because every $x \in V_{K'}$ has two elements in its preimage in $G_{K}$.
\end{remark}

\subsection{The Case for Strong Matching}
\label{sec:relating:case:for:strong:matching}

The two following examples serve to illustrate why we find it necessary to strengthen the matching criterion when matching is not controlled. 

\begin{example}
    \label{ex:pbpo:remove:loop}
    In \pbpostrong, an application of the rule
    \begin{center}
        % diagram for Example 22

{% scope
\newcommand{\nodexa}{\vertex{x_1}{cblue!20}}%
\newcommand{\nodexb}{\vertex{x_2}{cblue!20}}%
\newcommand{\nodey}{\vertex{y}{cgreen!20}}%
\newcommand{\nodez}{\vertex{z}{cred!10}}%
\newcommand{\nodeu}{\vertex{u}{cpurple!25}}%
\newcommand{\nodex}{\vertex{x}{cblue!20}}%
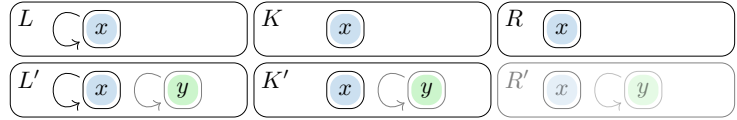
\begin{center}\vspace{0ex}
  \scalebox{\rulescale}{
  \begin{tikzpicture}[->,node distance=12mm,n/.style={}]
    \graphbox{$L$}{0mm}{0mm}{35mm}{8mm}{-4mm}{-4mm}{
      \node [npattern] (x)
      {\nodex};
      \draw [epattern,loop=180,looseness=3] (x) to node {} (x);
    }
    \graphbox{$K$}{36mm}{0mm}{35mm}{8mm}{-4mm}{-4mm}{
      \node [npattern] (x)
      {\nodex};
    }
    \graphbox{$R$}{72mm}{0mm}{35mm}{8mm}{-8mm}{-4mm}{
      \node [npattern] (x)
      {\nodex};
    }
    \graphbox{$L'$}{0mm}{-9mm}{35mm}{8mm}{-4mm}{-4mm}{
      \node [npattern] (x)
      {\nodex};
      \draw [epattern,loop=180,looseness=3] (x) to node {} (x);
      \node [nset] (y) [right of=x] {\nodey};
      \draw [eset,loop=180,looseness=3] (y) to node {} (y);
    }
    \graphbox{$K'$}{36mm}{-9mm}{35mm}{8mm}{-4mm}{-4mm}{
      \node [npattern] (x)
      {\nodex};
      \node [nset] (y) [right of=x] {\nodey};
      \draw [eset,loop=180,looseness=3] (y) to node {} (y);
    }
    \transparentgraphbox{$R'$}{72mm}{-9mm}{35mm}{8mm}{-8mm}{-4mm}{
      \node [npattern] (x)
      {\nodex};
      \node [nset] (y) [right of=x] {\nodey};
      \draw [eset,loop=180,looseness=3] (y) to node {} (y);
    }
  \end{tikzpicture}
  }\vspace{0ex}
\end{center}%
}% scope
    \end{center}
    in an unlabeled graph $G_L$ removes a loop from an isolated vertex that has a single loop, and preserves everything else. In PBPO, a match is allowed to map all of $G_L$ into the component determined by vertex $\{ x \}$, so that the rule deletes all of $G_L$'s edges at once.
    (Before studying the next example, the reader is invited to consider what the effect of the PBPO rule is if $R$ and $R'$ are replaced by $L$ and $L'$, respectively.)
\end{example}

\begin{example}
    \label{example:pbpo:spiral}
    Consider the following PBPO rule application
    \begin{center}
        {% scope

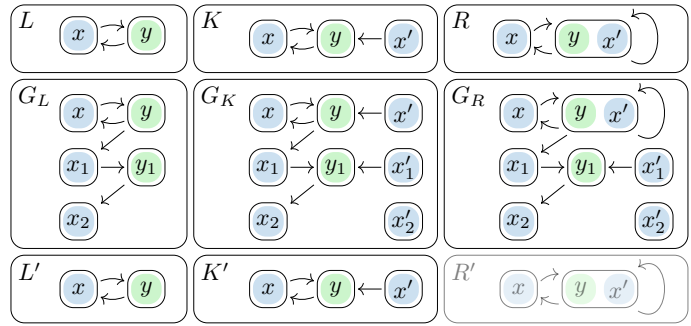
\newcommand{\nodex}{\vertex{x}{cblue!20}}
\newcommand{\nodexa}{\vertex{x_1}{cblue!20}}
\newcommand{\nodexb}{\vertex{x_2}{cblue!20}}
\newcommand{\nodexp}{\vertex{x'}{cblue!20}}
\newcommand{\nodexap}{\vertex{x_1'}{cblue!20}}
\newcommand{\nodexbp}{\vertex{x_2'}{cblue!20}}
\newcommand{\nodey}{\vertex{y}{cgreen!20}}
\newcommand{\nodeya}{\vertex{y_1}{cgreen!20}}
\newcommand{\nodez}{\vertex{z}{cred!10}}
\newcommand{\nodeza}{\vertex{z_1}{cred!10}}
\newcommand{\nodezb}{\vertex{z_2}{cred!10}}
\newcommand{\nodezc}{\vertex{z_3}{cred!10}}
\newcommand{\nodeu}{\vertex{u}{cpurple!25}}

%\begin{center}
  \scalebox{\rulescale}{
  \begin{tikzpicture}[->,node distance=12mm,n/.style={}]
      \graphbox{$L$}{0mm}{0mm}{26mm}{10mm}{-3mm}{-4.5mm}{
      \node [npattern] (x)
      {\nodex};
      \node [npattern] (y) [right of=x, short] {\nodey};
      \draw [epattern] (x) to[bend left=20] node {} (y);
      \draw [epattern] (y) to[bend left=20] node {} (x);
    }
    \graphbox{$G_L$}{0mm}{-11mm}{26mm}{25mm}{-3mm}{-5mm}{
      \node [npattern] (x) {\nodex};
      \node [npattern] (y) [right of=x, short] {\nodey};
      \draw [epattern] (x) to[bend left=20] node {} (y);
      \draw [epattern] (y) to[bend left=20] node {} (x);
      
      % extensions to pattern
      
      \node [npattern] (xa) [below of=x, short, yshift=2mm] {\nodexa};
      \node [npattern] (ya) [right of=xa, short] {\nodeya};
      \node [npattern] (xb) [below of=xa, short, yshift=2mm] {\nodexb};
      
      \draw [epattern] (y) to (xa);
      \draw [epattern] (xa) to (ya);
      \draw [epattern] (ya) to (xb);
    }
    \graphbox{$G_K$}{27mm}{-11mm}{36mm}{25mm}{-7mm}{-5mm}{
      \node [npattern] (x)
      {\nodex};
      \node [npattern] (y) [right of=x, short] {\nodey};
      \draw [epattern] (x) to[bend left=20] node {} (y);
      \draw [epattern] (y) to[bend left=20] node {} (x);
      
      % extensions to pattern
      
      \node [npattern] (xa) [below of=x, short, yshift=2mm] {\nodexa};
      \node [npattern] (ya) [right of=xa, short] {\nodeya};
      \node [npattern] (xb) [below of=xa, short, yshift=2mm] {\nodexb};
      
      \draw [epattern] (y) to (xa);
      \draw [epattern] (xa) to (ya);
      \draw [epattern] (ya) to (xb);
      
      % extension to G_L
      
      \node [npattern] (xp) [right of=y, short] {\nodexp};
      \node [npattern] (xpa) [right of=ya, short] {\nodexap};
      \node [npattern] (xpb) [below of=xpa, short, yshift=2mm] {\nodexbp};
    
      \draw [epattern] (xp) to (y);
      \draw [epattern] (xpa) to (ya);
    }
    \graphbox{$K$}{27mm}{0mm}{36mm}{10mm}{-7mm}{-5mm}{
      \node [npattern] (x)
      {\nodex};
      \node [npattern] (y) [right of=x, short] {\nodey};
      \draw [epattern] (x) to[bend left=20] node {} (y);
      \draw [epattern] (y) to[bend left=20] node {} (x);
      
      % extra
      \node [npattern] (xp) [right of=y, short] {\nodexp};
      \draw [epattern] (xp) to (y);
    }
    
    \graphbox{$R$}{64mm}{0mm}{36mm}{10mm}{-8mm}{-5mm}{
      \node [npattern] (x) {\nodex};
      \node [npattern] (y) [right of=x] {\nodey \ \nodexp};
      \draw [epattern] (x) to[bend left=20] node {} (y);
      \draw [epattern] (y) to[bend left=20] node {} (x);
      
      \draw [epattern,loop=0,looseness=3] (y) to node {} (y);
    }
    
    \graphbox{$G_R$}{64mm}{-11mm}{36mm}{25mm}{-7mm}{-5mm}{
            \node [npattern] (x)
      {\nodex};
      \node [npattern] (y) [right of=x] {\nodey \ \nodexp};
      \draw [epattern] (x) to[bend left=20] node {} (y);
      \draw [epattern] (y) to[bend left=20] node {} (x);
      
      % extensions to pattern
      
      \node [npattern] (xa) [below of=x, short, yshift=2mm] {\nodexa};
      \node [npattern] (ya) [right of=xa, short] {\nodeya};
      \node [npattern] (xb) [below of=xa, short, yshift=2mm] {\nodexb};
      
      \draw [epattern] (y) to (xa);
      \draw [epattern] (xa) to (ya);
      \draw [epattern] (ya) to (xb);
      
      % extension to G_L
      
      \node [npattern] (xpa) [right of=ya, short] {\nodexap};
      \node [npattern] (xpb) [below of=xpa, short, yshift=2mm] {\nodexbp};
    
      \draw [epattern] (xpa) to (ya);
      
      % loop
      
      \draw [epattern,loop=0,looseness=3] (y) to node {} (y);
      
    }
    \graphbox{$L'$}{0mm}{-37mm}{26mm}{10mm}{-3mm}{-5mm}{
      \node [npattern] (x)
      {\nodex};
      \node [npattern] (y) [right of=x, short] {\nodey};
      \draw [epattern] (x) to[bend left=20] node {} (y);
      \draw [epattern] (y) to[bend left=20] node {} (x);
    }
    \graphbox{$K'$}{27mm}{-37mm}{36mm}{10mm}{-7mm}{-5mm}{
      \node [npattern] (x)
      {\nodex};
      \node [npattern] (y) [right of=x, short] {\nodey};
      \draw [epattern] (x) to[bend left=20] node {} (y);
      \draw [epattern] (y) to[bend left=20] node {} (x);
      
      % extra
      \node [npattern] (xp) [right of=y, short] {\nodexp};
      \draw [epattern] (xp) to (y);
    }
    \transparentgraphbox{$R'$}{64mm}{-37mm}{36mm}{10mm}{-7mm}{-5mm}{
      \node [npattern] (x)
      {\nodex};
      \node [npattern] (y) [right of=x] {\nodey \ \nodexp};
      \draw [epattern] (x) to[bend left=20] node {} (y);
      \draw [epattern] (y) to[bend left=20] node {} (x);
      
      \draw [epattern,loop=0,looseness=3] (y) to node {} (y);
    }
  \end{tikzpicture}
  }
%\end{center}

}% scope
    \end{center}
    to a host graph $G_L$ (the morphisms are defined in the obvious way). Intuitively, host graph $G_L$ is spiralled over the pattern of $L'$. The pullback then duplicates all elements mapped onto $x \in V_{L'}$ and any incident edges directed at a node mapped into $y \in V_{L'}$. The pushout, by contrast, affects only the image of $u : K \to G_K$.
\end{example}

The two examples show how locality of transformations cannot be enforced using PBPO. They also illustrate how it can be difficult to characterize the class of host graphs $G_L$ and adherences $\alpha$ that establish a match, even for trivial left-hand sides. Finally, Example \ref{example:pbpo:spiral} in particular highlights an asymmetry that we find unintuitive: if one duplicates and then merges/extends pattern elements of $L'$, the duplication affects all elements in the $\alpha$-preimage of $t_L(L)$ (which could even consist of multiple components), whereas the pushout affects only $u(K) \subseteq G_K$. In \pbpostrong, by contrast, transformations of the pattern affect the pattern only, and the overall applicability of a rule is easy to understand if the context graph is relatively simple (e.g., as in Example~\ref{ex:simple:rewrite:step}).

\begin{remark}[$\Gamma$-preservation]
    A locality notion has been defined for PBPO called \emph{$\Gamma$-preservation}~\cite{corradini2019pbpo}. $\Gamma$ is some subobject of $L'$, and a rewrite step
    $\pbpostep{G_L}{G_R}{\rho}{m}{\alpha}$ is said to be \emph{$\Gamma$-preserving} if the $\alpha : G_L \to L'$ preimage of $\Gamma \subseteq L'$ is preserved from $G_L$ to $G_R$ (roughly meaning that this preimage is neither modified nor duplicated). Similarly, a rule is $\Gamma$-preserving if the rewrite steps it gives rise to are $\Gamma$-preserving. If one chooses $\Gamma$ to be the context graph (the right component) of $L'$ in Example~\ref{ex:pbpo:remove:loop}, then the rule, interpreted as a PBPO rule, is $\Gamma$-preserving. Nonetheless,
    PBPO does not prevent the mapping of arbitrarily large parts of the context graph of $G_L$ onto the image of $t_L$ (in Example~\ref{ex:pbpo:remove:loop}, the left component of $L'$) which usually \emph{is} modified. In this sense, the rule can still give rise to nonlocal effects.
\end{remark}

\section{\pbpostrong as a Unifying Theory for Quasitoposes}
\label{sec:quasitopos}

\newcommand{\subsumes}{\subseteq}
\renewcommand{\models}{\prec}
\newcommand{\simulates}{<}

\newcommand{\framework}[1]{\mathcal{#1}}
\newcommand{\FF}{\framework{F}}
\newcommand{\GG}{\framework{G}}

\newcommand{\MMmono}{\hookrightarrow}
\newcommand{\leftMMmono}{\hookleftarrow}

For this section, we need the following vocabulary.

\begin{definition}
    Let $\FF$, $\GG$ be rewriting formalisms. We write $\FF \prec_\CC \GG$ (or $\FF \prec \GG$ if $\CC$ is clear from context) to denote that in category $\CC$, for any $\FF$ rule $\rho$, there exists a $\GG$ rule $\tau$ such that ${\Rightarrow^\rho_\FF} = {\Rightarrow^\tau_\GG}$. If $\FF \prec_\CC \GG$ holds, then we say that \emph{$\GG$ models $\FF$ in $\CC$}. Similarly, we say that an $\FF$ rule $\rho$ can be \emph{modeled by a class} of $\GG$-rules $S$ if equation
    \[
        {\Rightarrow_\FF^\rho} \; = \; {\bigcup_{\tau \in S} \, {\Rightarrow_\GG^\tau}} 
    \]
    holds.
\end{definition}

Some of the most well-known graph rewriting formalisms include DPO~\cite{ehrig1973graph}, SPO~\cite{lowe1993algebraic}, SqPO~\cite{corradini2006sesqui} and AGREE~\cite{corradini2015agree, corradini2020algebraic}. In the conference version of this paper~\cite{overbeek2021pbpo}, we conjectured that in $\Graph$ and with monic matching, which implies conflict-freeness of SPO matches,
\begin{equation}
    \label{eq:modeling:claims}
    \begin{tikzpicture}[default,nodes={rectangle,inner sep=3mm},baseline=(l.base)]
        \node (t) {\pbpostrong};
        \node (l) at (t.west) [anchor=east,yshift=-3.5mm]  
          {$\text{SPO} \models \text{SqPO} \models \text{AGREE}$};
        \node (r) at (t.east) [anchor=west,xshift=-2mm,yshift=-3.5mm]  
          {$\text{DPO}$};
        \node (b) at ($(t)+(0,-7mm)$) {PBPO};
        \node at ($(l.north east)!0.5!(t.south west)$) [rotate=35] {$\models$};
        \node at ($(r.north west)!0.5!(t.south east) + (-.5mm,0mm)$) [rotate=180-35] {$\models$};
        \node at ($(t)!0.5!(b)$) [rotate=90] {$\models$};
    \end{tikzpicture}
\end{equation}
holds, and that the other comparisons do not hold. In this diagram, the claims AGREE $\models$ \pbpostrong and DPO $\models$ \pbpostrong were the two open ones: the claim PBPO $\models$ \pbpostrong was established in the conference paper for toposes~\cite[Lemma 32]{overbeek2021pbpo} (and thus in particular for \Graph), and SPO $\models$ SqPO~\cite[Proposition 13]{corradini2006sesqui} and SqPO $\models$ AGREE~\cite[Theorem 2]{corradini2015agree} were known for categories more general than \Graph in the literature.

The main contribution of this section is that we establish Diagram~\eqref{eq:modeling:claims} (but without SPO: see Remark~\ref{remark:modeling:SPO} below) much more generally: namely, for $\CC$ any quasitopos, and assuming regular monic matching (Theorem~\ref{thm:pbpostrong:models:others}). Of these claims, only SqPO $\prec$ AGREE follows as a corollary of the previously cited theorem~\cite[Theorem 2]{corradini2015agree}.
Quasitoposes include not only the categories of sets, directed multigraphs~\cite{lack2004adhesive} and typed graphs~\cite{corradini2015agree} (all of which are toposes), but also a variety of structures such as Heyting algebras (considered as categories)~\cite{wyler1991lecture}, and the categories of simple graphs (equivalently, binary relations)~\cite{johnstone2007quasitoposes}, fuzzy sets~\cite{wyler1991lecture}, algebraic specifications~\cite{johnstone2007quasitoposes} and safely marked Petri Nets~\cite{johnstone2007quasitoposes}. Behr et al.~\cite{behr2021concurrency} recently proposed quasitoposes as a natural setting for non-linear rewriting.

This section is structured as follows. In Section~\ref{sec:quasitopos:preliminaries} we provide all the required definitions and results pertaining to quasitoposes. We then first prove a useful sufficient condition for determinism of \pbpostrong rules in Section~\ref{sec:determinism}. This result is independent of the main result of this section. Next, we work towards the main result, establishing PBPO $\prec$ \pbpostrong, AGREE $\prec$ \pbpostrong and DPO $\prec$ \pbpostrong for quasitoposes in Sections \ref{sec:pbpostrong:models:pbpo}, \ref{sec:pbpostrong:models:agree} and \ref{sec:pbpostrong:models:dpo}, respectively.

\begin{remark}[Modeling SPO]
\label{remark:modeling:SPO}
    Deciding SPO $\models$ \pbpostrong for quasitoposes (with regular monic matching) appears relatively involved due to the fact that SPO is defined for partial morphisms, rather than total ones. For this reason, we consider the question to be beyond the scope of the present paper. However, we conjecture that for  any SPO rule, which is a partial morphism $\rho : L \rightharpoonup R$, the \pbpostrong rule
    \begin{center}
        % https://tikzcd.yichuanshen.de/#N4Igdg9gJgpgziAXAbVABwnAlgFyxMJZABgBpiBdUkANwEMAbAVxiRABkQBfU9TXfIRQBGclVqMWbANLdeIDNjwEiAJjHV6zVohAAlOXyWCio4eK1TdAFQAU0gJSGF-ZUJKlzmyTpB32TlziMFAA5vBEoABmAE4QALZIoiA4EEgAzNQ4dFgMbAAWEBAA1iDe2mwAOpUw2QD6stQMdABGMAwACq4mulhg2LDOsQlJWWmIZCk5ebqFJYw4ZRIVunlNre1dxiq9-ViDPNFxiYjJqUjqy1YgMUvNbZ3dOyB9A6yHIMMnk+eIACzla7VWp0OqcdYPLYCZ6vfasLLTApFUofL6jFLjAFXXzVeJ0HD5GLxYAdABCXDuG0e2yELz2gyafV8kDA8JA+RgdCgbBZbOyuR5BHe8jRiEyGKQWMsOMqaDoMQAxs04HBgMD6tIuMAGBSglwgA
        \begin{tikzcd}[column sep=20mm]
        \pmb{L} \arrow[d, "\eta_L" description, hook] & \pmb{K} \arrow[d, "\eta_K" description, hook] \arrow[l, "\pmb{l}" description, hook', thick] \arrow[r, "\pmb{r}" description, thick] \arrow[ld, "\mathrm{PB}", phantom] & \pmb{R} \\
        T(L)                                    & T(K) \arrow[l, "\parclass{\eta_K}{l}" description]                                                                               &  
        \end{tikzcd}
    \end{center}
    models it, where the top span is the span representation of $\rho$ (Definition~\ref{def:m:partial:map}), and the hooked arrows $\MMmono$ are regular monomorphisms. For the meaning of the given pullback square, see the $\MM$-partial map classifier definition (Definition~\ref{def:partial:map:classifier}).
\end{remark}

\begin{remark}[Double Pullback Rewriting]
    There also exists the double pullback rewriting (DPU) approach by Bauderon~\cite{bauderon1995uniform} and Bauderon and Jacquet~\cite{bauderon2001pullback}. As a first approximation, a double pullback rule $\rho$ is of the form $L \stackrel{l}{\to} A \stackrel{r}{\leftto} R$; a match is a morphism $m : G \to L$; and a step from $G$ to $H$ is given by a diagram
    \begin{center}
        \begin{tikzcd}
        G \arrow[d, "m" description] \arrow[r, "q" description] \arrow[rd, "\mathrm{PB}", phantom] & D \arrow[d, "u" description] & H \arrow[d] \arrow[l] \arrow[ld, "\mathrm{PB}", phantom] \\
        L \arrow[r, "l" description]                                              & A           & R \arrow[l, "r" description]                            
        \end{tikzcd}
    \end{center}
    where $(q, u)$ is a \emph{pullback complement}.  On top of this, constraints on rules, occurrences and steps are added in order to make the method well-behaved. The constraints are rather technical and vary slightly between~\cite{bauderon1995uniform} and~\cite{bauderon2001pullback}.
    
    We have provided a detailed comparison between \pbpostrong{} and DPU in a recent \pbpostrong{} tutorial~\cite{overbeek2023tutorial}. The most relevant conclusion is that although DPU uses pullbacks (which are generic), the details are defined for slice categories of simple graphs~\cite[Definition 3]{bauderon2001pullback}, and it is not clear to what extent these can be generalized (let alone be generalized to (quasi)toposes). So it is not evident how DPU can be fit into the general picture.
\end{remark}

\subsection{Quasitoposes}
\label{sec:quasitopos:preliminaries}

\newcommand{\interpretas}[2]{{\llbracket {#1} \rrbracket}_{#2}}

The following definitions on $\MM$-partial map classification derive from work by Cockett and Lack on restriction categories \cite{cockett2002restrictionI, cockett2003restrictionII}, but we follow the presentations of Corradini et al.~\cite{corradini2015agree} and Behr et al.~\cite{behr2021concurrency}.

\begin{definition}[Stable System of Monics~\cite{cockett2002restrictionI}]
    A \emph{stable system of monics} $\MM$ in $\CC$ is a class of monos in $\CC$ that includes all isomorphisms, is closed under composition (${m,m' \in \MM} \implies {m\circ m' \in \MM}$) and is stable under pullback, i.e.,
    if
    \begin{center}
        % https://tikzcd.yichuanshen.de/#N4Igdg9gJgpgziAXAbVABwnAlgFyxMJZABgBpiBdUkANwEMAbAVxiRAEEQBfU9TXfIRRkAjFVqMWbAELdeIDNjwEiI0mOr1mrRCAAicvksGry4rVN0BhbuJhQA5vCKgAZgCcIAWyRkQOCCQAZmoGOgAjGAYABX5lIRAsMGxYEE1JHRAvAHJDEA9vJDV-QMQAJnTtNh9QiKjY4xVdJJTWHjdPH0Q-AKLKy3zc2siYuJNm5KxU9vzO4Ope8v7M1zSQMJGGgSbEyemKLiA
        \begin{tikzcd}
        A \arrow[r, "m'" description] \arrow[d, "f'" description] & C \arrow[d, "f" description] \\
        B \arrow[r, "m" description]                              & D                           
        \end{tikzcd}
    \end{center}
    is a pullback and $m \in \MM$, then $m' \in \MM$.
    
    We use $\MMmono$ to denote $\MM$-monos.
\end{definition}

Examples of stable systems of monics in any category include the class of all monos and the class of all isomorphisms.

Consider how in set theory, any partial function $g: A \rightharpoonup B$ with domain $A' \subseteq A$ can be represented by a total injective function $m : X \mono A$ (typically an inclusion) and total function $f : X \to B$ such that
$m(X) = A'$ and $f(x) = g(m(x))$ for all $x \in X$. The following definition is the categorical generalization of this idea.

\begin{definition}[$\MM$-Partial Map]
\label{def:m:partial:map}
    Let $\MM$ be a stable system of monics. An \emph{$\MM$-partial map} is a span $A \stackrel{m}{\leftMMmono} X \stackrel{f}{\to} B$ where $m \in \MM$.
\end{definition}

Alternatively in set theory, one can extend set $B$ to $B_\star = B \uplus \{ \star \}$, where $\star$ represents undefined elements. Then any partial map $g : A \rightharpoonup B$ with domain $A' \subseteq A$ is represented by the total function $g_\star : A \to B_\star$ satisfying $g_\star(x) = g(x)$ for $x \in A'$, and $g_\star(x) = \star$ for $x \notin A'$. This extension is minimal in the sense that $g_\star : A \to B_\star$ is uniquely defined by $g_\star(x) = g(x)$ for $x \in A'$, and $g_\star(x) \notin B$ for $x \notin A'$. The following definition is the categorical generalization of this idea.

\begin{definition}[$\MM$-Partial Map Classifier~\cite{corradini2015agree}]
\label{def:partial:map:classifier}
    Let $\MM$ be a stable system of monics in $\CC$. An
    \emph{$\MM$-partial map classifier} $(T, \eta)$ in $\CC$ consists of a functor $T : \CC \to \CC$ and a natural transformation $\eta : \id{\CC} \to T$ such that 
    \begin{itemize}
        \item for all objects $X \in \obj{\CC}$, $\eta_X : X \to T(X)$ is in $\MM$; and
        \item for each \emph{$\MM$-partial map}, i.e., span $A \stackrel{m}{\leftMMmono} X \stackrel{f}{\to} B$ with $m \in \MM$, there exists a unique morphism $\parclass{m}{f} : A \to T(B)$ making
        \begin{center}
            % https://tikzcd.yichuanshen.de/#N4Igdg9gJgpgziAXAbVABwnAlgFyxMJZABgBpiBdUkANwEMAbAVxiRAA0QBfU9TXfIRQBGclVqMWbAELdeIDNjwEiZYePrNWiEAEE5fJYKKj11TVJ0AVABTSAlN3EwoAc3hFQAMwBOEALZIoiA4EEgAzNQ4dFgMbAAWEBAA1iDmktogADpZMNEA+rLUDHQARjAMAAr8ykIgWGDYsAYgvgFIZCFhiABMUTFxOokpaRJabIHFZRXVRio6DU2sPN5+gYidoUHp4zpeoyXlVTXGC41YzSuta0h9XRE7ltlZaHQ+AMYlcHDA-lzAXi4B2mxzmdUWF2WFC4QA
            \begin{tikzcd}[column sep=15mm]
            X \arrow[d, "m" description, hook] \arrow[r, "f" description] & B \arrow[d, "\eta_B" description, hook] \\
            A \arrow[r, "\parclass{m}{f}" description]                    & T(B)                                   
            \end{tikzcd}
        \end{center}
    a pullback square.
\end{itemize}

\noindent
For $\MM$-morphisms $m : X \MMmono A$, we let $\parclassid{m}$ denote $\parclass{m}{\id{X}} : A \to T(X)$.
\end{definition}

\begin{example}
\label{example:partial:map:classifier}
    As explained in the running text above, in $\Set$ there exists a mono-partial map classifier $(T, \eta)$ with $T(B) = B \uplus \{ \star \}$ and $\eta_B$ the inclusion.
    
    In the category of unlabeled directed graphs, there exists a mono-partial map classifier $(T, \eta)$, for graphs $G = (V,E,s,t)$ so that $T(G) = (V_\star, E_\star, s_\star, t_\star)$ where $V_\star = V \uplus \{ \star \}$;  $E_\star = E \uplus (V_\star \times V_\star)$; $s_\star(e) = s(e)$ if $e \in E$ and $\pi_1(e)$ otherwise; and $t_\star(e) = t(e)$ if $e \in E$ and $\pi_2(e)$ otherwise ($\pi_i$ the $i$'th projection). An example is given by
    \begin{center}
        $G \; = $
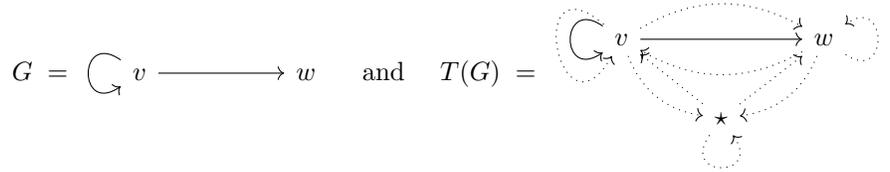
        \begin{tikzcd}
        v \arrow[loop, distance=2em, in=215, out=145] \arrow[rr] &  & w
        \end{tikzcd}
        \quad and \quad 
        $T(G) \; = \hspace{-3mm} $ % https://tikzcd.yichuanshen.de/#N4Igdg9gJgpgziAXAbVABwnAlgFyxMJZABgBpiBdUkANwEMAbAVxiVpAF9T1Nd9CUAJnJVajFmwDunbiAzY8BIgEZSy0fWatEIADq6ARhBydRMKAHN4RUADMAThAC2SMiBwRX1BhAhoiggDsZLaMcDCiDHQGMAwACryKAiD2WBYAFiZcdo4uiG4eSMrZIA7OXu6eiMIgDFhg2iBQxjjmINQxYFBIALQAzMQlZXk1hYiqtfWNzTit3UO5SKNVNT5+RCFhEd7RsQkK-GypGSbeU2wzczI55ePUYxNr-igAnJsM4ZG78YmHOseZdqTBoXFptBa3ApVNxPALBUihD7bWrffZ8JT-NKAs4gnSXcGyYYVB4dGBdJADHHTMHzQmLO6VCqdbqISnA6mzAk3Eb3aFU0Gc2ncoq8pak8msmHnPE00wcIA
        \begin{tikzcd}
        v \arrow[loop, distance=2em, in=215, out=145] \arrow[rr] \arrow[rd, dotted, bend right] \arrow[dotted, loop, distance=4em, in=240, out=140, looseness=3] \arrow[rr, dotted, bend left] &                                                                                                & w \arrow[dotted, loop, distance=2em, in=35, out=325] \arrow[ll, dotted, bend left] \arrow[ld, dotted, bend left] \\
                                                                                                                                                                                  & \star \arrow[ru, dotted] \arrow[dotted, loop, distance=2em, in=305, out=235] \arrow[lu, dotted] &                                                                                                                 
        \end{tikzcd},
    \end{center}
    where the dotted edges represent the edges $e \in V_\star \times V_\star$.
    It can be seen that for any partial homomorphism $\psi: H \to G$ defined on subgraph $H' \subseteq H$, there exists exactly one homomorphism $\psi_\star : H \to T(G)$ such that $\psi_\star(x) = \psi(x)$ for $x \in V_{H'} \cup E_{H'}$ and $\psi_\star(x) \notin V_{G} \cup E_{G}$ for $x \notin V_{H'} \cup E_{H'}$. Equivalently, $\psi_\star$ is the unique morphism $\parclass{m}{f} : H \to T(G)$ making
        \begin{center}
            % https://tikzcd.yichuanshen.de/#N4Igdg9gJgpgziAXAbVABwnAlgFyxMJZABgBpiBdUkANwEMAbAVxiRAA0QBfU9TXfIRQBGclVqMWbAELdeIDNjwEiZYePrNWiEAEE5fJYKKj11TVJ0AVABTSAlN3EwoAc3hFQAMwBOEALZIoiA4EEgAzNQ4dFgMbAAWEBAA1iDmktogADpZMNEA+rLUDHQARjAMAAr8ykIgWGDYsAYgvgFIZCFhiABMUTFxOokpaRJabIHFZRXVRio6DU2sPN5+gYidoUHp4zpeoyXlVTXGC41YzSuta0h9XRE7ltlZaHQ+AMYlcHDA-lzAXi4B2mxzmdUWF2WFC4QA
            \begin{tikzcd}[column sep=15mm]
            H' \arrow[d, "m" description, hook] \arrow[r, "f" description] & G \arrow[d, "\eta_G" description, hook] \\
            H \arrow[r, "\parclass{m}{f}" description]                    & T(G)                                   
            \end{tikzcd}
        \end{center}
    a pullback square, where $H \stackrel{m}{\hookleftarrow} H' \stackrel{f}{\to} G$ a partial map span representation of $\psi$, and $\eta_G$ and $m$ are inclusions.
    
    The generalization to labeled graphs is straightforward: between any two nodes $u,v \in V_\star$ and $l \in \labels$, there is one $l$-labeled edge representing an undefined $l$-edge between $u$ and $v$.
\end{example}

The following two definitions are standard in algebraic graph rewriting.

\begin{definition}[{$\MM$-Van Kampen Square~\cite[Definition 2.3.2]{ehrig2010categorical}}]
    Let $\MM$ be a class of monomorphisms in a category $\CC$. A pushout along $m \in \MM$ is an \emph{$\MM$-Van Kampen} (\emph{VK}) \emph{square} if, whenever it lies at the bottom of a commutative cube
    \begin{center}
        \begin{tikzcd}[row sep={30,between origins},column sep={30,between origins}]
            & F  \ar[rr] \ar[dd, hook] \ar[dl] & &  G \ar[dd, hook] \ar[dl] \\
            E \ar[rr, crossing over] \ar[dd, hook] & & H \\
              & B \ar[rr] \ar[dl, hook', "m" description] & &  C \ar[dl] \\
            A \ar[rr] && D \ar[from=uu,crossing over, hook]
        \end{tikzcd}
    \end{center}
    where the vertical back faces are pullbacks and all vertical arrows are in $\MM$, the top face is a pushout iff the vertical front faces are pullbacks.
\end{definition}

\begin{definition}[{$\MM$-Adhesive Category~\cite[Definition 2.4]{ehrig2010categorical}}]
\label{def:m:adhesive:category}
    Let $\MM$ be a class of monomorphisms in a category $\CC$. $\CC$ is \emph{$\MM$-adhesive} (also known as \emph{vertical weak adhesive HLR}) if
    \begin{itemize}
        \item $\CC$ has pushouts along $\MM$-morphisms;
        \item $\CC$ has pullbacks along $\MM$-morphisms;
        \item pushouts along $\MM$-morphisms are $\MM$-VK squares; and
        \item $\MM$ contains all isomorphisms and is closed under composition, pullback and pushout.
    \end{itemize}
\end{definition}

\begin{lemma}[{\cite[Lemma 13]{lack2004adhesive}}]
\label{lemma:m:adhesive:pushouts:are:pullbacks}
    If $\CC$ is $\MM$-adhesive for a class of monomorphisms $\MM$, then pushouts along $\MM$-morphisms are pullbacks. \qed
\end{lemma}

\newcommand{\Reg}[1]{\mathrm{rm}(#1)}
\newcommand{\RegC}{\Reg{\CC}}

Recall that a monomorphism is \emph{regular} if it is an equalizer for a parallel pair of morphisms. We write $\RegC$ to denote the class of regular monomorphisms of a category $\CC$.

\begin{definition}[{Quasitopos~\cite{wyler1991lecture, adamek2009joy,johnstone2002sketches}}]
    A category $\CC$ is a \emph{quasitopos} if it has all finite limits and colimits, it is locally cartesian closed, and it has a regular-subobject classifier.
\end{definition}

In this paper we rely on various results about quasitoposes: the notions of local cartesian closure and regular-subobject classifier will not be used directly, and therefore need not be understood. We could have equivalently defined quasitoposes in terms of regular-partial map classifiers and additional properties~\cite{adamek2009joy,wyler1991lecture}, but because these definitions appear less standardized, we decided against it.

The following results for quasitoposes will be used throughout the section. We cite original sources, but stress our indebtedness to the summary provided by Behr et al.~\cite[Corollary 1]{behr2021concurrency}.

\begin{proposition}[Properties of Quasitoposes]
\label{prop:quasitopos:properties}
    A quasitopos \CC
    \begin{enumerate}
        \item has a stable system of monics $\MM = \RegC$~\cite[Corollary 28.6(3) and Proposition 28.3(2)]{adamek2009joy};
        \item has an $\MM$-partial map classifier~\cite[Definition 19.3]{wyler1991lecture};
        \item\label{prop:decomposition:regular} satisfies ${gf \in \RegC} \implies {f \in \RegC}$ for all morphisms $f$ and $g$~\cite[Proposition 7.62(2) and Corollary 28.6(2)]{adamek2009joy};
        \item is partial map adhesive~\cite[Lemma 13]{heindel2010hereditary} and thus $\MM$-adhesive~\cite[Theorems 3.4 and 3.6]{ehrig2010categorical};
        \item is rm-quasiadhesive~\cite[Definition 1.1]{garner2012axioms}, meaning that pushouts along regular monos exist, and that these pushouts along regular monos are pullbacks and are stable under pullback; 
        \item has unique epi-$\MM$ factorizations for every morphism~\cite[Proposition 28.10]{adamek2009joy}, i.e., every $f : A \to B$ factors uniquely as $A \stackrel{e}{\epi} C \stackrel{m}{\MMmono} B$ (up to isomorphism in $C$), where $e$ is epic and $m$ regular monic; and
        \item is stable under slicing, i.e., every slice category $\CC/X$ (for $X \in \obj{\CC}$) is a quasitopos~\cite[Theorem 19.4]{wyler1991lecture}.
    \end{enumerate}
    \qed
\end{proposition}

We will also need the following basic result, valid in any category.

\begin{proposition}[{\cite[Corollary 7.63 and Proposition 7.66]{adamek2009joy})}]
\label{prop:epi:regular:mono:is:iso}
    An epimorphism that is a regular monomorphism is an isomorphism. \qed
\end{proposition}

Many categories of interest are toposes. We close this subsection by giving some relevant examples of non-topos quasitoposes.

\begin{example}[Simple Graphs]
    The category of simple graphs \SimpleGraph has pairs of sets $G = (V,E)$ with $E \subseteq V \times V$ as objects, and the usual graph homomorphisms $\psi$ as arrows.  \SimpleGraph is not a topos, but it is a quasitopos (see, e.g., \cite[Section 2.1]{behr2021concurrency}). The regular monos of \SimpleGraph are injective homomorphisms that reflect edges (i.e., for regular monos $\psi : G \MMmono H$, $(\psi(v), \psi(w)) \in E_H$ implies $(v,w) \in E_G$). A modeling example using simple graphs can be found in~\cite[Section 3]{corradini2006sesqui}. The regular-mono partial map classifier $(T, \eta)$ for $\SimpleGraph$ sends simple graphs $G = (V,E)$ to $T(G) = (V \uplus \{ \star \}, 
    E \cup (V \times \{ \star \}) \cup (\{ \star \} \times V) \cup \{ (\star, \star) \})$, with $\eta_G : G \MMmono T(G)$ the inclusion~\cite[Exercise 28.2(3)]{adamek2009joy}.
\end{example}

\begin{definition}[Complete Lattice]
\label{definition:complete:lattice}
    A \emph{complete lattice} $(\labels, \leq)$ is a poset such that all subsets $S$ of $\labels$ have a supremum (join) $\bigvee S$ and an infimum (meet) $\bigwedge S$.
    Any complete lattice has a global maximum $\top$ and global minimum $\bot$, respectively.
    
    A complete lattice $\labels$ is \emph{infinitely distributive} if
    \[
        x \wedge (\bigvee_{y \in S} y) \quad = \quad \bigvee_{y \in S} (x \wedge y)
    \]
    holds for all $x \in \labels$ and $S \subseteq \labels$.
\end{definition}

Observe that pullbacks and pushouts in a complete lattice, considered as a category, correspond to meets and joins, respectively.

\begin{example}[Infinitely Distributive Complete Lattice]
\label{example:inf:distr:complete:lattice}
    A complete lattice $\labels$ considered as a category is a quasitopos iff $\labels$ is infinitely distributive~\cite[Exercise, 28D.(b)]{adamek2009joy}. Such complete lattices are exactly the complete Heyting algebras.
\end{example}

The regular monos in a complete Heyting algebra $\labels$ are exactly the identities. So the regular mono-partial map classifier is not very interesting. However, some toposes can be endowed with a complete Heyting algebra as in the following definition, which gives rise to more interesting classifiers.

\begin{example}[{$\labels$-Fuzzy Set~\cite[Chapter 8]{wyler1991lecture},\cite{goguen1967fuzzy}}]
\label{example:L:fuzzy:set}
     An \emph{$\labels$-fuzzy set} $(A, \alpha)$ consists of a set $A$ and a membership function $\alpha : A \to \labels$, where $\labels$ is a complete lattice. An $\labels$-fuzzy set morphism from $(A, \alpha)$ to $(B, \beta)$ is a function $\phi : A \to B$ such that $\alpha(x) \leq \beta \circ \phi(x)$ for all $x \in A$. $\labels$-fuzzy sets have been well studied, in particular in the context of fuzzy logics.
     
     The regular monomorphisms in the category of $\labels$-fuzzy sets are the $\labels$-fuzzy set morphisms $\phi$ that preserve membership (i.e., $\alpha(x) = \beta \circ \phi(x)$).
     The category of $\labels$-fuzzy sets is known to be a quasitopos if $\labels$ is a complete Heyting algebra (see Stout~\cite[Corollary 8]{stout1992thelogic} and Goguen~\cite[Proposition 4]{goguen1967fuzzy}).
     The regular mono-partial map classifier $(T, \eta)$ then sends sets $S = (A, \alpha)$ to $T(S) = (A \uplus \{ \star \}, \alpha \cup \{(\star \mapsto \top) \})$, with $\eta_S : S \MMmono T(S)$ the inclusion. An example of how this can be used to redefine the membership of an element using \pbpostrong{} is given by   
     \begin{center}
     % https://tikzcd.yichuanshen.de/#N4Igdg9gJgpgziAXAbVABwnAlgFyxMJZABgBpiBdUkANwEMAbAVxiRAB13g6A9ADwAEnAL4hhpdJlz5CKMgEYqtRizadu-UgIBGPAJ5aAxjwBeQ9qPGTseAkTIAmJfWatEHLrz5btxzjgg0ETEJEAwbGSJ5UidqF1V3dV5ObQgcHz92ALRzS1Dw6TsUaMU4lTcPYAFk9lT0nX0jU1yQ60LZZGjKMtc1Tx4UtJarMKlbDodyZ3K+qt4mYfyxyJRJ0uVexK5qniYfRoFjE2DhJRgoAHN4IlAAMwAnCABbJDIQAKRo97osBjYACwgEAA1iAegkQC9qAw6NoYAwAArLIogLBgbCwVogB7PT7UD6ISYgGFwxHI2So9FYTHgiqcRhof50LE4l6IN4EolwsBQJAAWgAHG94hUcAB9AAyAgAvOYYDg6GLgOo+CJRNDYfCkREUWiMax8T8-u5ASCWY82QBmfEQJBEkla8lsPXUg0bCEMADk5txiAALDa8e6KhdJWDiZqyTqKS7MSNWUgA+9bYhrd9fmxIGA3SK2ExvRrSdr2s6qXHQgnEABWQPs2lsP6Fx3R0v6n1smvJxOGjMmoGg+vuJjhh1RkvuWOseMWpCdgkANkHIHuI8jxfGrdd7e7XcQAHYl6GAEqrotOidlqcVmeIRe7g-p40gU0D4NsADup+b48pbenvrnFMvlzLYnjoHB-nuJ5gARAAhdUIzPFsLzbaE0QqLM3X+GA6F5dxMPDBVe3AAgrzuG8k05JdODAiCoJg+CvzHDcUK3NDs0zUjw2w3DOI4nsnwI-8O1rB8QI8WjIOghEAHkENHdcVl-NjiXQvisJwvCSP4x91O3VNazTBggTQexSFuRg4BgJQFPPZcsAuf4cDEChhCAA
        \begin{tikzcd}[column sep=20mm]
        \{a^x \} \arrow[d, "m" description, hook] \arrow[dd, "t_L = \eta_{\{x\}}" description, hook, out=180,in=180,looseness=1.5] & \{a^\bot \} \arrow[l, "l" description] \arrow[d, "u" description, hook] \arrow[r, "r" description] \arrow[ld, "\mathrm{PB}", phantom] \arrow[rd, "\mathrm{PO}", phantom] & \{ a^u \} \arrow[d, "w" description, hook] \\
        {\{a^x, b^y, c^z \}} \arrow[d, "\alpha" description]                                                        & {\{ a^\bot, b^y, c^z \}} \arrow[l, "g_L" description] \arrow[d, "u'" description] \arrow[r, "g_R" description] \arrow[ld, "\mathrm{PB}", phantom]                        & {\{ a^u, b^y, c^z\}}                       \\
        {\{a^x, \star^\top\}}                                                                                          & {\{a^\bot, \star^\top \}} \arrow[l, "l'" description]                                                                  &                                           
        \end{tikzcd}
     \end{center}
     where $a,b,c,\star$ are set elements, and superscripts denote membership values in $\labels$. Morphisms $\alpha$ and $u'$ map $b$ and $c$ onto $\star$; all other elements are mapped by way of inclusion. The rule can be seen to change the fuzziness of an element with fuzziness $x$ to $u$, in any context. The context is not changed.
\end{example}

Rewriting fuzzy sets is only slightly more interesting than rewriting sets. But in Section~\ref{sec:category:graph:lattice}, we define the category \Graphleq, which is essentially the category of ``fuzzy graphs'' (Example~\ref{example:L:fuzzy:set}), and argue that it is a very useful graph category for modeling and relabeling. We have proven very recently that this category is a quasitopos if $\labels$ is a complete Heyting algebra~\cite{rosset2023fuzzy}.

\subsection{Determinism}
\label{sec:determinism}

For this subsection, we need the following definition.

\begin{definition}
    A \pbpostrong rule $\rho$ is \emph{deterministic in a category $\CC$} if
    \[
        {\pbpostrongstep{G_L}{G_R}{\rho}{m}{\alpha} \; \text{  and  } \; \pbpostrongstep{G_L}{G_R'}{\rho}{m}{\beta}} \qquad \Longrightarrow \qquad {G_R \iso G_R'}
    \]
    for any $G_L, G_R, G_R' \in \obj{\CC}$, match morphism $m : L \to G_L$ and adherence morphisms $\alpha, \beta : G_L \to L'$.
\end{definition}

By uniqueness of (co)limits, it is easy to see that in any category $\CC$, \pbpostrong rewriting is deterministic, if, for any match morphism $m : L \to G_L$, there exists at most one adherence morphism $\alpha : G_L \to L'$ that establishes a strong match square. In the setting where $\CC$ is a quasitopos, we define the following concept, which we can use to prove a more useful sufficient condition on the type morphism $t_L$.

\begin{definition}[Restricted Classifier]
    Let $\CC$ be a quasitopos.
    A regular mono $t_L : L \MMmono L'$ is a \emph{restricted classifier} if 
    \begin{center}
        % https://tikzcd.yichuanshen.de/#N4Igdg9gJgpgziAXAbVABwnAlgFyxMJZABgBpiBdUkANwEMAbAVxiRABkQBfU9TXfIRQAmclVqMWbACoAKdgEpuvEBmx4CRAIxjq9Zq0QcA5N3EwoAc3hFQAMwBOEALZIyIHBCQ6QAIxhgUEgAzO76UkYAOpEwOHQA+pzUDHT+DAAK-BpCIFhg2LAg1HFYDGwAFhAQANbK9k6uiO6eSKISBmw4iUUgKWmZ6oJseQWsxXSlFVW1PPUurcVeiD7hhiAIyakwGVlDRiNYheOTRiVlXBRcQA
        \begin{tikzcd}
        L \arrow[rr, "\eta_L" description, hook, bend left] \arrow[r, "t_L" description, hook] & L' \arrow[r, "s" description, tail] & T(L)
        \end{tikzcd}
    \end{center}
    commutes for some mono $s : L' \mono T(L)$.
\end{definition}

\begin{example}
    Recall $G$ and $T(G)$ from Example~\ref{example:partial:map:classifier}.  Let $G'$ be a graph resulting from deleting any number of dotted edges (and possibly $\star$) from $T(G)$, e.g.,
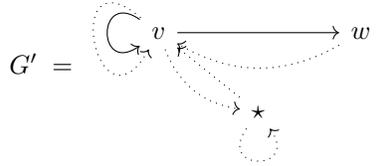
    \begin{center}
        $G' \; = \hspace{-3mm} $ % https://tikzcd.yichuanshen.de/#N4Igdg9gJgpgziAXAbVABwnAlgFyxMJZABgBpiBdUkANwEMAbAVxiVpAF9T1Nd9CUAJnJVajFmwDunbiAzY8BIgEZSy0fWatEIADq6ARhBydRMKAHN4RUADMAThAC2SMiBwRX1BhAhoiggDsZLaMcDCiDHQGMAwACryKAiD2WBYAFiZcdo4uiG4eSMrZIA7OXu6eiMIgDFhg2iBQxjjmINQxYFBIALQAzMQlZXk1hYiqtfWNzTit3UO5SKNVNT5+RCFhEd7RsQkK-GypGSbeU2wzczI55ePUYxNr-igAnJsM4ZG78YmHOseZdqTBoXFptBa3ApVNxPALBUihD7bWrffZ8JT-NKAs4gnSXcGyYYVB4dGBdJADHHTMHzQmLO6VCqdbqISnA6mzAk3Eb3aFU0Gc2ncoq8pak8msmHnPE00wcIA
        \begin{tikzcd}
        v \arrow[loop, distance=2em, in=215, out=145] \arrow[rr] \arrow[rd, dotted, bend right] \arrow[dotted, loop, distance=4em, in=240, out=140, looseness=3]  &                                                                                                & w \arrow[ll, dotted, bend left] \\
                                                                                                                                                                                  & \star \arrow[dotted, loop, distance=2em, in=305, out=235] \arrow[lu, dotted] &                                                                                                                 
        \end{tikzcd}.
    \end{center}
    The inclusion $G \hookrightarrow G'$ is a restricted classifier. The $t_L$ of Example~\ref{ex:simple:rewrite:rule} is also a restricted classifier.
    
    Recall the regular mono-partial map classifier $(T, \eta)$ for $\labels$-fuzzy sets, with $\labels$ a complete Heyting algebra (Example~\ref{example:L:fuzzy:set}), which sends fuzzy sets $S = (A, \alpha)$ to $T(S) = (A \uplus \{ \star \}, \alpha \cup \{(\star \mapsto \top) \})$, with $\eta_S : S \MMmono T(S)$ the inclusion. The inclusion $(A, \alpha) \MMmono (A \uplus \{ \star \}, \alpha \cup \{ \star \mapsto x \})$ for $x < \top$ is a restricted classifier, and effectively sets the upper bound for context elements to $x$ instead of $\top$. Upper bounds for context elements can similarly be changed for fuzzy graphs (Section~\ref{sec:category:graph:lattice}).
\end{example}

\begin{proposition}[{\cite[Lemma 3.2]{corradini2020algebraic}}]
\label{prop:comm:square:mono:is:pb}
    Any commuting square of the form
    \begin{center}
        % https://tikzcd.yichuanshen.de/#N4Igdg9gJgpgziAXAbVABwnAlgFyxMJZABgBpiBdUkANwEMAbAVxiRAEEQBfU9TXfIRQBGclVqMWbTjz7Y8BImWHj6zVohAAhbrxAZ5goqJXU1UzQGFu4mFADm8IqABmAJwgBbJGRA4ISABMZpIaIPYg1Ax0AEYwDAAK-ApCIFhg2LC6rh7eiKJ+AYgAzCHqbAAWkSDRcYnJRprpmayyIO5eQdT+SKV+dFgMbDgDQ2UW7dW18UmGik0ZWFltHXm+PflR6WFQEEwxDKzUFTB0UGyQYEf9gxcErRRcQA
        \begin{tikzcd}
        A \arrow[d, "g" description] \arrow[r, equals] & A \arrow[d, "h" description] \\
        B \arrow[r, "f" description, tail]                          & C                           
        \end{tikzcd}
    \end{center}
    is a pullback. \qed
\end{proposition}

\begin{proposition}
    Let $\CC$ be a quasitopos. For any regular mono $t_L$, $t_L$ is a restricted classifier iff $\parclassid{t_L}$ (Definition~\ref{def:partial:map:classifier}) is monic.
\end{proposition}
\begin{proof}
    For direction $\Longrightarrow$, use Proposition~\ref{prop:comm:square:mono:is:pb} and the uniqueness property of partial map classifiers to conclude that the obtained mono $s$ is $\parclassid{t_L}$. Direction $\Longleftarrow$ is immediate. \qed
\end{proof}

\begin{lemma}
\label{lemma:unique:alpha}
    Let $\CC$ be a quasitopos, and let $L \stackrel{t_L}{\MMmono} L'$ and $L \stackrel{m}{\MMmono} G_L$ be given. If $t_L$ is a restricted classifier, then any $\alpha : G_L \to L'$ making
    \begin{center}
        % https://tikzcd.yichuanshen.de/#N4Igdg9gJgpgziAXAbVABwnAlgFyxMJZABgBpiBdUkANwEMAbAVxiRABkQBfU9TXfIRQBGclVqMWbTjz7Y8BIqOHj6zVog4BybrxAZ5gomRXU1UzQHEA+jPEwoAc3hFQAMwBOEALZIyIHAgkAGZqHDosBjYACwgIAGsQM0kNEF9qBjoAIxgGAAV+BSEQLDBsWF13L19EUQCgxAAmMIiozViEpIl1NhxbLsyc-MKjTVLy1lkQTx8-MIa6hlLUqAgmLIZWamiYOig2SDAtgNaDgkm9GZrQ+qRm7osQAB0nxjRougHs3ILDRTGylgKlwKFwgA
        \begin{tikzcd}
        L \arrow[d, "m" description, hook] \arrow[r, equals] & L \arrow[d, "t_L" description, hook] \\
        G_L \arrow[r, "\alpha" description]                               & L'                                  
        \end{tikzcd}
    \end{center}
    a pullback square is unique.
\end{lemma}
\begin{proof}
    Let $\alpha$ and $\beta$ each make the left square of
    \begin{center}
        % https://tikzcd.yichuanshen.de/#N4Igdg9gJgpgziAXAbVABwnAlgFyxMJZABgBpiBdUkANwEMAbAVxiRABkQBfU9TXfIRQBGclVqMWbTjz7Y8BIqOHj6zVog4BybrxAZ5gomRXU1UzQHEA+jL0GBilACZSpierYAVABTsAlLpyjkLIrpRmkhoc3OIwUADm8ESgAGYAThAAtkhkIDgQSADM1Dh0WAxsABYQEADWIJGemjnUDHQARjAMAAr8CkIgWGDYsEEgGdlIovmFiK755ZWaNfWNHhb5tuvtXb39RprDo6yyE5k5iHkF023D0VAQTB0MrNRVMHRQbJBgb4sVH4EU56SaXEqzJALczRAA6sMYaCqdB2nW6fUMThA6SwCSqOHWcCqWFSBMQwjOYOKpTmCyJJLJAFoZjC2PCumVxlTyTSkABWO5-NiPZ6vdYfL5AoWlJZSkFpC783mIAAsMsBK1qDSam3hMDK2zaaP2mMGxywY0pivmyrVGzhsLQdHSAGN2nBRsAcLYuKi9hiQmxzWNBdE4BAGBb1mUNQDKlapjzIaqdQ6snQcFV0llgD0AEK+o3+g5Y4P-SNCzS-f4S75V4HR2X1oVcChcIA
        \begin{tikzcd}[column sep=35]
        L \arrow[d, "m" description, hook] \arrow[r, equals]    & L \arrow[d, "t_L" description, hook] \arrow[r, equals] \arrow[rd, "\mathrm{PB}", phantom] & L \arrow[d, "\eta_L" description, hook] \\
        G_L \arrow[r, "\alpha"', shift right] \arrow[r, "\beta", shift left] & L' \arrow[r, "\parclassid{t_L}" description, tail]                                                     & T(L)                                   
        \end{tikzcd}
    \end{center}
    a pullback square. Because the right square is a pullback square, both $\parclassid{t_L}\alpha$ and $\parclassid{t_L}\beta$ make the outer square a pullback square, using the pullback lemma. By the uniqueness property of the partial map classifier, $\parclassid{t_L}\alpha = \parclassid{t_L}\beta$. By monicity of $\parclassid{t_L}$, $\alpha = \beta$. \qed
\end{proof}

\begin{definition}[Classifying \pbpostrong Rule]
\label{def:classifying:rule}
    A \pbpostrong rule $\rho$ is said to be \emph{classifying} if $L \stackrel{t_L}{\MMmono} L'$ is a restricted classifier.
\end{definition}

We can now state the following sufficient condition.

\begin{theorem}[Determinism]
\label{thm:determinism}
    Let $\CC$ be a quasitopos, and let $\rho$ be a classifying \pbpostrong rule.
    If both
    ${\pbpostrongstep{G_L}{G_R}{\rho}{m}{\alpha}}$ and ${\pbpostrongstep{G_L}{G_R'}{\rho}{m}{\beta}}$, then $\alpha = \beta$ and $G_R \iso G_R'$.
\end{theorem}
\begin{proof}
    From Lemma~\ref{lemma:unique:alpha} and general pullback and pushout properties.~\qed
\end{proof}

Note that in categories with strict initial objects (such as \Set and \Graph), another sufficient condition for determinism is having $K'$ initial. Then the result of any rewrite step is $R$. This observation implies that necessary conditions for determinism cannot be phrased merely in terms of type graphs and matching.

\begin{remark}
    For some categories, we conjecture that a sufficient condition for determinism of a rule $\rho$ (with type morphism $t_L : L \mono L'$) is the case where there exists a commuting diagram
    \begin{center}
        % https://tikzcd.yichuanshen.de/#N4Igdg9gJgpgziAXAbVABwnAlgFyxMJZABgBpiBdUkANwEMAbAVxiRABkQBfU9TXfIRQBGclVqMWbdgHIZ3XiAzY8BIgCYx1es1aIQAFQAU7AJQK+KwUVHDxOqftndxMKAHN4RUADMAThAAtkiiIDgQSJphdFgMbDgxcdQAFjB0UEhgTAwM1Ax0AEYwDAAK-KpCIFhg2LAg2pJ6IAA6zWh0fgDG+XC1wAAeXBYg-kFIZGERiKEJsWzJEBAA1vUg+UWl5db61bWsDbps-cOjwYgT4ZHURWAZiADMF4nziysHji3NMAkA+px5hWKZSsah2NSwdR4vgCZwuU3u1FmcX0SNWDiaOD+q3WQK2oKq4MhilOIUR8JSaTuYQA7hBUukEIjnijngCNsCBPjdhD9hJDvofC4uEA
        \begin{tikzcd}
        L \arrow[r, "x" description, hook] \arrow[rr, "\eta_L" description, hook, bend left] \arrow[rd, "t_L" description, tail] & L'' \arrow[r, "\parclassid{x}" description, tail] \arrow[d, "f" description, two heads, tail] & T(L) \\
                                                                                                                                 & L'                                                                                            &     
        \end{tikzcd}
    \end{center}
    for some restricted classifier $x : L \MMmono L''$, and monic and epic morphism $f : L'' \to L'$. 
    For instance, in category {\Graphleq} (Section~\ref{sec:category:graph:lattice}), this property corresponds to  relaxing the labels of a restricted classifier $x : L \MMmono L''$.
    We are interested in knowing in which classes of categories this property holds.
\end{remark}

\subsection{\pbpostrong Models PBPO}
\label{sec:pbpostrong:models:pbpo}

In the conference version of the present paper, we proved that in any topos, any PBPO rule can be modeled by a class of \pbpostrong rules~{\cite[Corollary 1]{overbeek2021pbpo}}. In this subsection, we generalize this result to the setting of quasitoposes (any topos is a quasitopos). Additionally, we streamline the original proof significantly.

\newcommand{\frgt}[1]{U_{#1}}

In order to prove our result, we first need to establish a technical result about slice categories $\CC/X$. We let $(A, f: A \to X)$ denote the objects of $\CC/X$; and $\frgt{X}$ the forgetful functor $\frgt{X} : \CC/X \to \CC$.

\begin{proposition}[{\cite[(17.3, Remarks)]{wyler1991lecture}}]
\label{prop:wyler:forgetful:functor}
    Let $f$ be a morphism in $\CC/X$. We have
    \begin{enumerate}
        \item\label{prop:wyler:forgetful:functor:mono} $f$ is monic $\iff$ $\frgt{X} f$ is monic;
        \item\label{prop:wyler:forgetful:functor:epi} if $\CC$ has finite products, then: $f$ is epic $\iff$  $\frgt{X} f$ is epic; and
        \item if $f$ is a regular mono, then $\frgt{X} f$ is a regular mono.
    \end{enumerate}
    \qed
\end{proposition}

Wyler additionally remarks the following for categories $\CC$ with finite products: if a morphism $\frgt{X} f : A \to X$ is a regular mono in $\CC$, then the morphism $f : (A, f) \to (X, \id{X})$ is a regular mono in $\CC/X$ \cite[(17.3, Remarks)]{wyler1991lecture}.
We prove that $\frgt{X}$ in fact \emph{reflects} all regular monomorphisms, i.e., that if $\frgt{X} f$ is a regular mono (with any codomain) in $\CC$, then so is $f$ in $\CC/X$ (Corollary~\ref{corr:forgetful:reflect:regular:monos} below).

Let $A \stackrel{f}{\leftarrow} C \stackrel{g}{\to} B$ be a span in a category with finite products.
Recall that the \emph{product map} $\langle f,g \rangle : C \to A \times B$ is the unique morphism induced by the universal property of product $A \times B$.

\begin{proposition}[{\cite[(Exercises 2--3, Chapter 3)]{goldblatt1979topoi}}]
\label{prop:product:map:equality}
    Product maps satisfy
    \begin{enumerate}
        \item\label{prop:product:map:equality:decompose} ${\langle f, g \rangle = \langle k, h \rangle} \Longleftrightarrow {f = k \text{ and } g = h}$; and
        \item\label{prop:product:map:equality:distribute} $\langle fh, gh \rangle = \langle f, g \rangle h$.
    \end{enumerate} \qed
\end{proposition}

\begin{proposition}
\label{prop:equalizer:product:maps}
    Assume $\CC$ has finite products. An equalizer $f : A \to B$ for parallel morphisms $g,h : B \to C$ is also an equalizer for the product maps $\langle 1_B, g \rangle, \langle 1_B, h \rangle : B \to B \times C$.
\end{proposition}
\begin{proof}
    First, $gf = hf \implies \langle 1_Bf, gf \rangle = \langle 1_Bf, hf \rangle \implies \langle 1_B, g \rangle f = \langle 1_B, h \rangle f$ using Proposition~\ref{prop:product:map:equality}. For universality, suppose $f' : A' \to B$ also has this equalizing property. Then $\langle 1_B, g \rangle f' = \langle 1_B, h \rangle f' \implies \langle 1_Bf', gf' \rangle = \langle 1_Bf', hf' \rangle \implies gf' = hf'$, again using Proposition~\ref{prop:product:map:equality}. We then obtain the unique $u : A' \to A$ such that $fu = f'$ from the fact that $f$ is an equalizer for $g$ and $h$. In a diagram:
    \begin{center}
        % https://tikzcd.yichuanshen.de/#N4Igdg9gJgpgziAXAbVABwnAlgFyxMJZABgBoBGAXVJADcBDAGwFcYkQBBEAX1PU1z5CKchWp0mrdgCEefEBmx4CRMsXEMWbRJwDkc-kqFEATKXU1NUnQGEDCgcuHIzJjZO0hpAAgA6vvABbeG87bnEYKABzeCJQADMAJwhApDIQHAgkUQktdniQGkZ6ACMYRgAFR2MdLDBsWHsklKQzDKzEdKtPAEJmQpBissrqlVr6rEaius8oCBwcSKbk1MQ2zOzLD3z9ItLyqqMxkDqGtl4Elc32pABmGjKwKCQAFgBOLbydKIGhg9HhCcJo0LiBmqschtEC9PtYQP5imAoowYN5yAB9aSkbxRPy+RL0JEo377EZHQGnSZsB4wJ5IAC071B4OuUJhuThCMJyNRGKx3gAFv4CUTqYNSYdBMdKSD5CzEJCOvcOZ4BSThpKnOwZedKNwgA
        \begin{tikzcd}
        A' \arrow[d, "!u" description, dotted] \arrow[rd, "f'" description] &                                                                                                                                                                                             & C          \\
        A \arrow[r, "f" description]                                        & B \arrow[ru, "g" description, bend left=49] \arrow[rd, "{\langle 1_B, g \rangle}" description, bend right=49] \arrow[rd, "{\langle 1_B, h\rangle}" description] \arrow[ru, "h" description] &            \\
                                                                            &                                                                                                                                                                                             & B \times C
        \end{tikzcd}
    \end{center}
    \qed
\end{proof}

\begin{proposition}
\label{prop:preserving:equalizers}
    Let $f : (A, a) \to (B, b)$ and $g,h : (B, b) \to (C, c)$ be morphisms in $\CC/X$.
    If $\frgt{X} f$ is an equalizer for $\frgt{X} g, \frgt{X} h$ in $\CC$, then $f$ is an equalizer for $g,h$ in $\CC/X$.
\end{proposition}
\begin{proof}
    First, $gf = hf$ in $\CC/X$ because $\frgt{X} f$ is an equalizer for $\frgt{X} g, \frgt{X} h$ in $\CC$, and composition is lifted from $\CC$. For universality, suppose $gf' = hf'$ for some $f' : (A', a') \to (B, b)$ in $\CC/X$. By definition of the slice category and its forgetful functor, $a' = b \circ \frgt{X} f'$ in $\CC$. In addition, because composition is lifted from $\CC$, $\frgt{X} g \circ \frgt{X} f' = \frgt{X} h \circ \frgt{X} f'$, and so one obtains a unique arrow $u : A' \to A$ with $\frgt{X} f \circ u = \frgt{X} f'$ using that $\frgt{X} f$ is an equalizer in $\CC$. Then $a' = b \circ \frgt{X} f' = b \circ \frgt{X} f \circ u$ which implies $a' = b \circ f \circ u = a \circ u$. Hence $u: (A', a') \to (A, a)$ is an arrow in $\CC/X$. Uniqueness in $\CC/X$ follows because any other arrow would violate the uniqueness property on the level of $\CC$.
    
    In a diagram, where solid and dashed arrows represent respectively objects and morphisms in $\CC/X$:
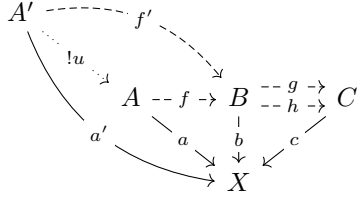
    \begin{center}
        % https://tikzcd.yichuanshen.de/#N4Igdg9gJgpgziAXAbVABwnAlgFyxMJZARgBpiBdUkANwEMAbAVxiRAEEQBfU9TXfIRQAmclVqMWbAELdeIDNjwEiAZjHV6zVohABhOXyWCio4eK1TdADUML+yocgAMpZxck6OAcm7iYUADm8ESgAGYAThAAtkiuIDgQSGQS2mxhINQMdABGMAwACg4mulhg2LCZIAxlXlB0cAAWAVU4dFgMbJBgrDzhUbGIKYlIoiBNWGE4SAC0Y5ZegVXZeYXFKqXlWJVZtWz1TS19IJExydQjiGMTU+epViCNy7n5RcYbIGUVrLs9+w3NKB2U6DMaXdT3LwAY2eqzeAg+X22vXkILu4M0njYOVhr3WQk+W0qxzRiHiGMhbDouLW7wJSOJqIGSAALBckogIXkwEDEDNVPEFlTfFkXrSEfSiSj+mdEGyEhyUjU-roDoCqkLdGERdUxfDHGwGT8QNzeQKScy5ey4pi0roAIRMGn6kqE77LPaqiA4HBHChcIA
        \begin{tikzcd}
        A' \arrow[rrdd, "a'" description, bend right] \arrow[rrd, "f'" description, dashed, bend left] \arrow[rd, "!u" description, dotted] &                                                                  &                                                                                                                               &                               \\
                                                                                                                                            & A \arrow[r, "f" description, dashed] \arrow[rd, "a" description] & B \arrow[r, "g" description, dashed, shift left=2] \arrow[r, "h" description, dashed, shift right] \arrow[d, "b" description] & C \arrow[ld, "c" description] \\
                                                                                                                                            &                                                                  & X                                                                                                                             &                              
        \end{tikzcd}
    \end{center}
    \qed
\end{proof}

\begin{proposition}
\label{prop:reflect:equalizer}
    If $\CC$ has finite products and $\frgt{X} f : A \to B$ is an equalizer for $g,h : B \to C$ in $\CC$, then $f$ is an equalizer for $\langle 1_B, g \rangle, \langle 1_B, h \rangle : (B, b) \to (B \times C, b \circ \pi_1)$ in $\CC/X$.
\end{proposition}
\begin{proof}
    Observe that morphisms $g$ and $h$ may not give rise to morphisms in $\CC/X$. However, by hypothesis, we can construct the product $B \times C$ in $\CC$. From this we can infer the morphism $b \circ \pi_1 : B \times C \to X$ and the product maps $\langle 1_B, g \rangle, \langle 1_B, h \rangle : B \to B \times C$. Crucially, $b \circ \pi_1$ is an object in $\CC/X$ and the product maps are morphisms in $\CC/X$, because $b \circ \pi_1 \circ \langle 1_B, g \rangle  = b = b \circ \pi_1 \circ \langle 1_B, h \rangle $.
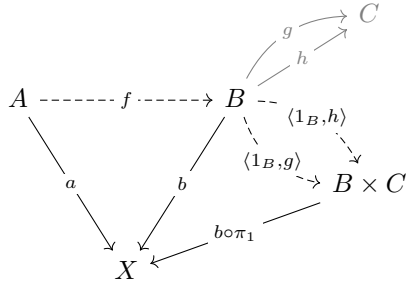
    \begin{center}
        % https://tikzcd.yichuanshen.de/#N4Igdg9gJgpgziAXAbVABwnAlgFyxMJZABgBoBGAXVJADcBDAGwFcYkQBBEAX1PU1z5CKAEwVqdJq3YAhHnxAZseAkXKkAzBIYs2iEAA15-ZUKIbSI7VL0gZAAgA6jvAFt49gMLHFAlcOQLYmtddm9uCRgoAHN4IlAAMwAnCFckMhAcCCR1EEYsMFsoejgACyiQGh1pfQTKvPoAIxhGAAU-M30C7FgfZNT0miykMUlQ-Xp6-ML2OAh8qCmmlvbTVS6wHrZeRJS0xFzhxFHq20al5raO9ZBurF6dkH79i0zs46qbdkbnAGMsJK-JyONBYAD65AuK2uwlum3u2wUzxyQ3er2aYEWiAAtBoMqd2M5GPQwNFGDB7OQwTJSPZos4kiSyWwaMTLqtBDc7r1WQUiiVyotHsiDqikK8CfoiUzyZTqbTSgyZSyGuyYexuSqMVi8byZvpimUKsK9ii3kgACyfcYgaJQq5rWGa+ra8XEE0DUXmxBWkCuxD4r76Ur2jn+DXwh6UbhAA
        \begin{tikzcd}
                                                                           &   &                                                                                                                                                                                                                                              & \textcolor{gray}{C}                                                  \\
        A \arrow[rr, "f" description, dashed] \arrow[rdd, "a" description] &   & B \arrow[ldd, "b" description] \arrow[rd, "{\langle 1_B, g\rangle}" description, dashed, bend right] \arrow[rd, "{\langle 1_B, h\rangle}" description, dashed, bend left] \arrow[ru, "g" description, bend left, color=gray] \arrow[ru, "h" description, color=gray] &                                                    \\
                                                                           &   &                                                                                                                                                                                                                                              & B \times C \arrow[lld, "b\circ \pi_1" description] \\
                                                                           & X &                                                                                                                                                                                                                                              &                                                   
        \end{tikzcd}
    \end{center}
    By Proposition~\ref{prop:equalizer:product:maps}, $\frgt{X} f$ is an equalizer for the product maps in $\CC$. By Proposition~\ref{prop:preserving:equalizers}, $f$ is an equalizer for the product maps in $\CC/X$.
    \qed
\end{proof}

\begin{corollary}
\label{corr:forgetful:reflect:regular:monos}
    If $\CC$ has finite products and $X \in \obj{C}$, then $\frgt{X}$ reflects regular monomorphisms.
\end{corollary}
\begin{proof}
    Let $\frgt{X} f : A \to B$ be a regular monomorphism in $\CC$. By definition this means $\frgt{X} f$ is an equalizer for some $g,h : B \to C$ in $\CC$. By Proposition~\ref{prop:reflect:equalizer}, $f$ is an equalizer for $\langle 1_B, g \rangle, \langle 1_B, h \rangle : (B, b) \to (B \times C, b \circ \pi_1)$ in $\CC/X$. So by definition $f$ is a regular monomorphism.
    \qed
\end{proof}

The definition below generalizes the concept of materialization by Corradini et al.~\cite[Definition 7]{corradini2019rewriting}, by replacing the class of all monics by a stable system of monics $\MM$. (For an example of a materialization and how to construct it, see~\cite[Corollary 9 and Example 10]{corradini2019rewriting}.)

\begin{definition}[$\MM$-Materialization]
    Let $\MM$ be a stable system of monics. The \emph{$\MM$-materialization} of a morphism $f :  A \to B$ is a terminal factorization of the form $A \stackrel{f'}{\MMmono} \materialization{f} \stackrel{f''}{\to} B$ for some object $\materialization{f}$. That is, for any other factorization of the form $A \stackrel{m}{\MMmono} C \stackrel{\alpha}{\to} B$, there exists a unique morphism $\beta : C \to \materialization{f}$ that makes the square of 
    \begin{center}
        % https://tikzcd.yichuanshen.de/#N4Igdg9gJgpgziAXAbVABwnAlgFyxMJZABgBpiBdUkANwEMAbAVxiRAEEQBfU9TXfIRQAmclVqMWbAELdeIDNjwEiZAIzj6zVog5y+SwUTWkN1LVN0AdKwzpgA5gxgACAGYubAJ3tPWPAwEVFBNKc0kdEABhbnEYKAd4IlA3LwgAWyQyEBwIJBMQACMYMCgkAGZsi0i3EGo7YoYABX5lIRAsMGxYfRBUjKRRHLzEcuocOiwGNgALCAgAazqJbTY3AHJlhphm1qNdTu7-eX7M0fGRguq19c36ukaWw2COrqwegL60s+zcpAAWajFUpZcaTaa6OaLZbXXSZe6PPYvQ7vY4pb4Ai75cKraxWRhoGZ0LYPHZPILtFEfE4YxC-EZDBidSJQCBMQrOZYzGB0Mq6SBgVhgqZsAVor4DRCA4YVBFkpGUt49erMtisnA4eIwiJsACENmKE1iXCAA
        \begin{tikzcd}
        A \arrow[rr, "f" description, bend left] \arrow[r, "m" description, hook] \arrow[d, equals] & C \arrow[r, "\alpha" description] \arrow[d, "!\beta" description, dotted] & B \\
        A \arrow[r, "f'" description, hook] \arrow[ru, phantom, "\mathrm{PB}" description]                                                                    & \materialization{f} \arrow[ru, "f''" description]                           &  
        \end{tikzcd}
    \end{center}
    a pullback square, and moreover makes the triangle commute.
\end{definition}

For $\MM$ the class of all monics, Corradini et al.\ have shown that all morphisms have mono-materializations if all slice categories have mono-partial map classifiers~\cite[Proposition 8]{corradini2019rewriting}. This condition holds in any topos. The following proposition and corollary generalize their result.

\begin{proposition}
\label{prop:pmc:strong:amendability}
    If all slice categories $\CC/X$ of a category $\CC$ have $\MM$-partial map classifiers and the forgetful functor $\frgt{X} : \CC/X \to C$ reflects and preserves $\MM$-morphisms, then $\CC$ has all $\MM$-materializations.
\end{proposition}
\begin{proof}
    The proof is analogous to the proof of Corradini et al.\ for $\MM$ the class of all monos, available on arXiv~\cite{corradini2019rewritingarXiv}. The difference is that we require reflection and preservation of $\MM$-morphisms, rather than relying implicitly on these properties for monos. \qed
\end{proof}

\begin{corollary}
    Quasitoposes $\CC$ admit rm-materializations.
\end{corollary}
\begin{proof}
    Using Corollary~\ref{corr:forgetful:reflect:regular:monos} and Proposition~\ref{prop:pmc:strong:amendability}, the fact that quasitoposes have rm-partial map classifiers, and the fact that the quasitopos property is stable under slicing.
    \qed
\end{proof}

\newcommand{\compact}[1]{\mathrm{compact}(#1)}
\newcommand{\compactiso}[1]{\mathrm{compact}^{\cong}(#1)}

\begin{definition}[Compacted Rule]
\label{def:pbpo:compacted:rule}
    Let $\CC$ be a quasitopos.
    For any canonical PBPO rule $\rho$ and any factorization $t_L = f \mathop{\circ} e$ where $e$ is epic (note that $f$ is uniquely determined by $e$, because $e$ is right-cancellative), the \emph{compacted \pbpostrong rule} $\rho_e$ is defined as the bold subdiagram of
    \begin{center}
        % https://tikzcd.yichuanshen.de/#N4Igdg9gJgpgziAXAbVABwnAlgFyxMJZABgBpiBdUkANwEMAbAVxiRABkQBfU9TXfIRRkAzFVqMWbdgHJuvEBmx4CRMgEZx9Zq0QcA+qx59lgtaQBMWybpAAdOwzpgA5gxgACAGYeHAJ2c3IwUlAVUUdXJrHTYAaXkTMKFkCyjqbSk9ACUExX4VZMjNdJs4w1zQgqJU4okY7PLjPNNw5EixEvqQWLkmyrMIy2jM+0dA929fOwDXd16Q-IGU0g66kaz5xKqUVKtOkYcnWc8ff3HPGV7xGCgXeCJQLz8IAFskMhAcCCRIkAAjGBgKBIAC0ABYAJz7Ww4fScahOAEMAAKi3CICwYGwsFyT1e72oXyQqRAAAsYHRgXocAB3CDkykIaFsVgIuhI1EtIQYrFYHFNPFvRAkomIX4AoGgkQfDK2LwgNkctHczHY4KPZ5CkXfRAiQl0LAMNikiAQADWCrWcrkipgKOVbFVfPVIEFSD1nx1v1lbC8l0tiLtnKSjt5-IUbsQYMJOplpT0Rtt9q5obVuM1SGjnqQAFZmXo-AH2UGHXoneGNfjEAA2GPEgUZmt1xAAdgbVYAHM3vfGQAx-Ung9seWn20KodnEB2x0gu5P1HGun4B33i8mQ2Wwy7I3nJ22I42W931DOxb9RQvT1nRdXT7XJxOcAajXoTeb01XHzqRKevz8TweVb3qKHo+noDgvHQOCkn4LzAMiABCXBFkqKabmmCKYrYkBgKyZIUlS4AEHhT6GmwOHbo2146iSYGjJB0GwfBSEBlh5HEZaDKERRlqkS+RG4ShJZoSOzofuO3b5vRUEwXBiHIYOpaiTimGCXoPHUFx7GCfqZHqRxd7HlJEEyUxyIAPIKauqEbspeEMGx+k6fhlLaSRz5ueJmbNkeVpsCZjFyZZQnrsO5b2Y5Al4VpTnuXpUVeYgf5iouBx2AxsnwcFikieFrFqQlmkEZ5un8TxXAUFwQA
        \begin{tikzcd}[row sep=25]
        L \arrow[ddd, "t_L" description, bend right=49] \arrow[d, "e" description, two heads] & K \arrow[l, "l" description] \arrow[r, "r" description] \arrow[d] \arrow[ld, "\mathrm{PB}", phantom] \arrow[rd, "\mathrm{PO}", phantom] & R \arrow[d]                    \\
        \pmb{L_e} \arrow[dd, "f" description, bend right] \arrow[d, "\pmb{f'}" description, hook, thick]         & \pmb{K_e} \arrow[l, thick] \arrow[r, thick] \arrow[d, hook, thick] \arrow[ld, "\pmb{\mathrm{PB}}", phantom] \arrow[rd, "\mathrm{PO}", phantom]                           & \pmb{R_e} \arrow[d]                  \\
        \pmb{\materialization{f}} \arrow[d, "f''" description]                                        & \pmb{\materialization{f}'} \arrow[d] \arrow[l, thick] \arrow[r] \arrow[ld, "\mathrm{PB}", phantom] \arrow[rd, "\mathrm{PO}", phantom]                  & \materialization{f}'' \arrow[d] \\
        L'                                                                                    & K' \arrow[l, "l''" description] \arrow[r, "r''" description]                                                                            & R'                            
        \end{tikzcd}
    \end{center}
    where the outer diagram is rule $\rho$, and $L_e \stackrel{f'}{\MMmono} \materialization{f} \stackrel{f''}{\to} L'$ the rm-materialization of $f$, and objects $\materialization{f}'$ and $\materialization{f}''$ are obtained by respectively pullback and pushout.
    
    We additionally define the class of rules
    \[
        \compact{\rho} = \{ \rho_e \mid \exists f e .\  t_L = f \circ e \text{ and $e$ is epic} \} \text{.}
    \]
    and its subclass
    \[
        \compactiso{\rho} = \{ \rho_e \mid \exists f e .\  t_L = f \circ e \text{ and $e$ is iso} \} \text{.}
    \]
\end{definition}

\begin{theorem}
\label{thm:modeling:pbpo:with:pbpostrong}
    Let $\CC$ be a quasitopos.
    For any canonical PPBO rule $\rho$ with type morphism $t_L : L \to L'$, the class of \pbpostrong rules $\compact{\rho}$ models it, that is,
    \[
        {\Rightarrow_\mathrm{PBPO}^\rho} \qquad =  \; {\bigcup_{\tau \, \in \, \compact{\rho}} \, {\Rightarrow_\mathrm{PBPO^{+}}^\tau}}.
    \]
\end{theorem}
\begin{proof}
    $\subseteq$: Assume a PBPO step $\pbpostep{G_L}{G_R}{\rho}{m}{\alpha}$ with $t_L = \alpha m$ for some $m$ and $\alpha$. Because we are in a quasitopos, $m$ admits a unique (epi, regular mono)-factorization $L \stackrel{e}{\epi} L_e \stackrel{m'}{\MMmono} G_L$. Define $f = \alpha m'$ and let $L_e \stackrel{f'}{\MMmono} \materialization{f} \stackrel{f''}{\to} L'$ be its materialization. Then by the materialization property, there exists a unique $\beta : G_L \to \materialization{f}$ such that $f' = \beta m'$ is a pullback and $\alpha = f''\beta$ commutes. The middle two rows of diagram
    \begin{center}
        % https://tikzcd.yichuanshen.de/#N4Igdg9gJgpgziAXAbVABwnAlgFyxMJZABgBpiBdUkANwEMAbAVxiRABkQBfU9TXfIRQAmclVqMWbANLdeIDNjwEiAFjHV6zVohAAlOXyWCiZAIzitU3ewD6rHkYEqRpC5sk6Q0+4YX9lIWR1dwltNj1fR39jFxJSYUtPNgBxW05oxWcg0USPcN002UyAkxR1PLDrEDSDEtigsgBmJIKQAB12hjowAHMGGAACADNBzoAnHv6HeSzA01JVVur2AHI-ObLkUSX86ul1+uyiURa9r07uvoGRsfbJ66HD2dK49TOqi66pm9GJn6ezyc83Ki2WXj0h3EMCgvXgRFAw3GEAAtkgyCAcBAkGYMVYvDh0iBqN0AEYwBgABVeQhAWDA2FgxJA5LAUCQAFoAGzEaJI1Ho6hYpBc85sNEkujkqk0tj0xmsais9mIDmqACcfORaMQGOFiCaYt0ipAZIp1IacoZWCZ1AAFjA6CrMQB3CAOp0ILUCxBmIXY3VG03Ms0yy26eU2maI7U4-1IUSfNjjENS82yiPWpnenV+zEB9RJxBgJgMBiS6UW466AbDHB+fk6xP6gCsRpLZYr6fDppgdYbsYN8cQoqLIBR6y7YerdKzJpwdCwDDYdogEAA1ox6zmkIX9Yaix3y6a09OQb3+zvEHuA23D6Xj6Gq+fa9v5I3d8OAOzth9T59lBeb4xj6d76gAHL+nYnpWGZAQOPqjvquJBp0jBoHadCprBPaRraLIwGynIaleSEBuqQYAISdOSC7Yd2M54SaDD0l4UAQDgOAwghOo-vmIpQY+p4AS48FXnxEGCf+cGvjxSAUfxvp4skujDKsk4wQx55MXJvp5shyltAwGlPnBOlXmYzbkeJw6WRZVk4mY9m2U0V6QYpZiue+g5mPpAZmIW+LJiZwlmXOumebZqgWcOd5BbonQonQOB2uMKLAJSADyXAhqxbCQGAJoes6BXzouy66KV9FnoB5neT6N5IO58UdO0SUpWlGXZblhX5QQRWOiV-XMguS59b10m4eFNkeQeLWJclqXpVlOUknllXDfag3jWVY0bRNmk1aJdUgbmDm+oFKmte1S1dath0ibSOlrb1+0DU6O0jeVn0xYpc1XQtHXLQAQvdLGveAm0gMVn1Ct9b3VY9VoKrpjUjqhbWLZ1lKgz1XhVVtH0I3De2QwdplTSjM36gp82Y0DGW4y9+NQzDxOYvDZPMaFlNRhF50oWOgO3TjYPrVzzJsxLJMVdLD1hVT9U6geNNKoRKpNIZ1RqYjCt8zLK5rpuDDASAH5Dh5GLKrutNXcMuu89mFBcEAA
        \begin{tikzcd}[column sep=40]
        L \arrow[dddd, "t_L" description, bend right=60] \arrow[dd, "m" description, bend right=49] \arrow[d, "e" description, two heads] &  & K \arrow[ll, "l" description] \arrow[rr, "r" description] \arrow[d] \arrow[rrd, "\mathrm{PO}", phantom] \arrow[lld, "\mathrm{PB}", phantom] &  & R \arrow[d]                    \\
        L_e \arrow[d, "m'" description, hook']   &  & K_e \arrow[ll] \arrow[rr] \arrow[d] \arrow[rrd, "\mathrm{PO}", phantom] \arrow[lld, "\mathrm{PB}", phantom]                                 &  & R_e \arrow[d]                  \\
        G_L \arrow[dd, "\alpha" description, bend right=49] \arrow[d, "!\beta" description, dotted]                                       &  & G_K \arrow[ll] \arrow[rr] \arrow[d] \arrow[rrd, "\mathrm{PO}", phantom] \arrow[lld, "\mathrm{PB}", phantom]                                 &  & G_R \arrow[d]                  \\
        \materialization{f} \arrow[from=uu, "f'" description, hook', bend left=50, pos=0.75, crossing over] \arrow[d, "f''" description]                                                                                    &  & \materialization{f}' \arrow[ll] \arrow[d] \arrow[rr] \arrow[rrd, "\mathrm{PO}", phantom] \arrow[lld, "\mathrm{PB}", phantom]                 &  & \materialization{f} '' \arrow[d] \\
        L' \arrow[from=uuu, "f" description, bend left=68, crossing over]                                                                                                                            &  & K' \arrow[ll, "l'" description] \arrow[rr, "r'" description]                                                                                &  & R'                            
        \end{tikzcd}
    \end{center}
    show that the pullback $f' = \beta m'$ defines a strong match and therefore a step
    \[
    \pbpostrongstep{G_L}{G_R}{\rho_e}{m'}{\beta}
    \]
    for rule $\rho_e \in \compact{\rho}$.
    
    $\supseteq$: We start with the middle two rows of the diagram of direction $\subseteq$, and with $f' = \beta m'$ a pullback.
    Observe that $m'$ is necessarily regular by regular monicity of $f'$ and pullback stability of regular monos. Then from the definition of $\rho_e$ it is immediate that the outer diagram defines a PBPO step using rule $\rho$, match $m$ and adherence morphism $\alpha$.
    \qed
\end{proof}

\newcommand{\pbpomonic}{PBPO$^{\MMmono}$}
\newcommand{\pbpomonicstep}[5]{{#1} \Rightarrow_{\mathrm{PBPO}^{\MMmono}}^{#3, (#4, #5)} {#2}}
\newcommand{\pbpomonicstepminimal}[1]{ \Rightarrow_{\mathrm{PBPO}^{\MMmono}}^{#1}}

\begin{definition}
    We let \pbpomonic{} denote the rewriting framework obtained by modifying PBPO to allow regular monic matches only. That is,
    \[
        {\pbpomonicstep{G_L}{G_R}{\rho}{m}{\alpha}} \quad \Longleftrightarrow \quad
        {\pbpostep{G_L}{G_R}{\rho}{m}{\alpha}} \text{ and $m$ is a regular mono.}
    \]
\end{definition}

\begin{proposition}
    Let $\CC$ be a quasitopos and $\rho$ a \pbpomonic{} rule. The class of \pbpostrong rules $\compactiso{\rho}$ models $\rho$.
\end{proposition}
\begin{proof}
    For any (epi, regular mono)-factorization $m = m'e$ of a regular mono $m$, $e$ is a regular mono (Proposition~\ref{prop:quasitopos:properties}), and hence an isomorphism (Proposition~\ref{prop:epi:regular:mono:is:iso}). Thus one can specialize the claim and proof of Theorem~\ref{thm:modeling:pbpo:with:pbpostrong} to $\compactiso{\rho}$ rather than $\compact{\rho}$. \qed
\end{proof}

\begin{corollary}[\pbpostrong Models PBPO]
\label{corr:pbpostrong:models:pbpo}
    Assume $\CC$ is a quasitopos and let $\rho$ be a PBPO rule.
    \begin{enumerate}
        \item There exists a single \pbpostrong rule $\tau$ such that ${\pbpomonicstepminimal{\rho}} = {\pbpostrongstepminimal{\tau}}$.
        \item If $t_L$ of $\rho$ is a regular mono, then there exists a single \pbpostrong rule $\tau$ such that ${\pbpostepminimal{\rho}} = {\pbpostrongstepminimal{\tau}}$.
    \end{enumerate} 
    \end{corollary}
\begin{proof}
    For the first claim, we are left with only one rule after identifying all isomorphic objects in the category.
    For the second claim, if $t_L$ of $\rho$ is a regular mono, then for all strong matches $t_L = \alpha m$, $m$ is a regular mono by pullback stability. Thus
    $\smash{ {\pbpostepminimal{\rho}}
    =
    {\pbpomonicstepminimal{\rho}}
    }$, and the first claim can be applied. \qed
\end{proof}

\subsection{\pbpostrong Models AGREE}
\label{sec:pbpostrong:models:agree}

AGREE is short for ``Algebraic Graph REwriting with controlled Embedding'' and is a rewriting framework introduced by Corradini et al.~\cite{corradini2015agree,corradini2020algebraic}. In categories where both SqPO and AGREE are applicable, AGREE can roughly be thought of as adding a filtering mechanism on top of the cloning operations that were originally introduced by SqPO. Moreover, an interesting technical aspect of AGREE is that it was the first formalism to utilize partial map classifiers for the definition of rewrite steps.

\newcommand{\AGREEstep}[4]{
  {#1} \Rightarrow_\mathrm{AGREE}^{#3, #4} {#2}
}

\newcommand{\AGREEsteprule}[2]{
  {\Rightarrow_\mathrm{AGREE}^{#1,#2}}
}

\newcommand{\AGREEstepminimal}[1]{
  {\Rightarrow_\mathrm{AGREE}^{#1}}
}

\begin{definition}[AGREE Rewriting~\cite{corradini2015agree, corradini2020algebraic}]
    Assume $\CC$ has $\MM$-partial map classifiers for a stable system of monics $\MM$.
    
    An \emph{AGREE rewrite rule} is of the form
    \begin{center}
        % https://tikzcd.yichuanshen.de/#N4Igdg9gJgpgziAXAbVABwnAlgFyxMJZABgBpiBdUkANwEMAbAVxiRABkQBfU9TXfIRQBGclVqMWbANLdeIDNjwEiAJjHV6zVohAAlOXyWCio4eK1Td0gOTdxMKAHN4RUADMAThAC2SUSA4EEgAzNQ4dFgMbAAWEBAA1iCakjqBAPqy1Ax0AEYwDAAK-MpCIFhg2LCGIF6+-uHBiGQS2mzR2XkFxcYquhVVrDwe3n6IAUFI6iA5+UUlJv2VWNUpbbqe9lxAA
        \begin{tikzcd}
        L & K \arrow[d, "t_K" description, hook] \arrow[l, "l" description] \arrow[r, "r" description] & R \\
          & K'                                                                                         &  
        \end{tikzcd}
    \end{center}
    where $t_K$ is an $\MM$-morphism.
    
    A diagram of the form
    \begin{center}
        % https://tikzcd.yichuanshen.de/#N4Igdg9gJgpgziAXAbVABwnAlgFyxMJZABgBpiBdUkANwEMAbAVxiRABkQBfU9TXfIRQBGclVqMWbANLdeIDNjwEiAJjHV6zVohAAlOXyWCiZYeK1TdAcQD6nHkYEqRpVRck6Q0gOSGF-MpCyKLmmp5sdrKOAcYuyOphEtqRtgYxis7BZO7hKboAKgAU7ACU3OIwUADm8ESgAGYAThAAtkiiIDgQSGTJViAMINQMdABGMAwACoEmulhg2LD+zW0d1N1I6oPjkzNxQiALS6x5A00rLe2IfZuIAMxnXu0ju9OzLkeLWMsbdFhDXQACwgEAA1pc1g8Nj1EAB2J5sAA6SLQdCaAGNRnAlsBWlxhjsJu8Dmxjj9WDFVtcACwwpAI-peFFozHYuDAHC2aRcYAMAmvYn7LJk77LKlXJAAVnpiBpEqhnTuMq6-0BIBB4Mh1xVdwAbArrtt9YberLGRMwFAkABaGkATkRuhRMBwdHshNGQo+h3Jv1VALYmohpsQuth91Djy6sM6DAWXkgYFOGpgdGtuiTKbdgczBBTXr2PtFJ0JlmZSNadBwQKarWAUwAQgTQ4zlU6QCiqzW6w3m563sKgiWKZ6E2ws4SgWmM+B84Sc+rJ6GlbC9R2u9Xa-WpgB5AVEouk+ZigvjvPJqczifzv65ueXleyukgS0Z+63NVB0EQjtc2SCkeIonqWXAUFwQA
        \begin{tikzcd}[column sep=45]
        L \arrow[d, "m" description, hook] \arrow[dd, "\eta_L" description, hook, bend right=49] & K \arrow[l, "l" description] \arrow[r, "r" description] \arrow[d, hook] \arrow[rd, "\mathrm{PO}", phantom]  & R \arrow[d] \\
        G_L \arrow[d, "\parclassid{m}" description] \arrow[ru, "\mathrm{PB}", phantom]           & G_K \arrow[d] \arrow[r] \arrow[l]                                                                                                                         & G_R         \\
        T(L) \arrow[ru, "\mathrm{PB}", phantom]                                                  & K' \arrow[l, "\parclass{t_K}{l}" description]     \arrow[from=uu, "t_K" description, hook, bend left=50, pos=0.8, crossing over]                                                                                                        &            
        \end{tikzcd}
    \end{center}
    defines an \emph{AGREE rewrite step} $\AGREEstep{G_L}{G_R}{\rho}{m}$, i.e., it is a step from object $G_L$ to object $G_R$ induced by rule $\rho$ and match morphism $m : L \MMmono G_L$.
\end{definition}

The proposition below is stated for $\MM$ the class of all monos in \cite{corradini2019pbpo}.

\begin{proposition}[{Relating AGREE and PBPO~\cite[Proposition 3.1]{corradini2019pbpo}}]
\label{prop:relating:AGREE:PBPO}
    Assume $\CC$ has $\MM$-partial map classifiers for a stable system of monics $\MM$. In the diagram
    \begin{center}
        % https://tikzcd.yichuanshen.de/#N4Igdg9gJgpgziAXAbVABwnAlgFyxMJZABgBpiBdUkANwEMAbAVxiRABkQBfU9TXfIRQBGclVqMWbANLdeIDNjwEiAJjHV6zVohAAlOXyWCio4eK1Td0gOSGF-ZUJKlzmyTpAAVABTsAlPaKAioo6m4S2mx6dlziMFAA5vBEoABmAE4QALZIoiA4EEgAzO5RujgA+rLUDHQARjAMAAqOJrpYYNiwINQ4dFgMbAAWEBAA1vaZOXl9RYhkkVYgQ7UNTa3GoSCd3aw86Vm5iPmFSOpLnhm9K+stbdu7WD0HINPHi2eIACx9A0O6UYTG51Rr3LZCHZdZ6sMrLAA68Jg-UqnFe73OcyQAFY4Z4qgY1mDNiFIU8XvIMYhSgV5r9LmxEWhhlgfFVpKQGIEiRsHmToRTDjNqVjELiGboMnYeeDSWxyftKUckPSvvlLJ5Edk6DhhhlssBmgAhLggu4kpzygWwladTyQMA24YwOhQNgOm39QbugiKoXHU7zcUaxnw7W6-WGgDypplFvaUL2ILtPsdN2drtTnv+We4FC4QA
        \begin{tikzcd}[column sep=45, row sep=30]
        \pmb{L} \arrow[d, "\eta_L" description, hook] & \pmb{K} \arrow[d, "\pmb{t_K}" description, hook, thick] \arrow[l, "\pmb{l}" description, thick] \arrow[r, "\pmb{r}" description, thick] \arrow[rd, "\mathrm{PO}", phantom] & \pmb{R} \arrow[d, "t_R" description] \\
        T(L) \arrow[ru, "\mathrm{PB}", phantom] & \pmb{K'} \arrow[l, "{\parclass{t_K}{l}}" description] \arrow[r, "r'" description]                                                         & R'                            
        \end{tikzcd}
    \end{center}
    let the bold subdiagram depict an AGREE rule $\rho$, and the entire diagram a PBPO rule $\interpretas{\rho}{\mathrm{PBPO}}$, where the right square is a pushout. Then for any match $m \in \MM$, we have
    $\AGREEsteprule{\rho}{m} = 
    \pbposteprule{\interpretas{\rho}{\mathrm{PBPO}}}{m}{\parclassid{m}}$. \qed
\end{proposition}

Observe that Proposition~\ref{prop:relating:AGREE:PBPO} relies on a particular choice of adherence morphism $\alpha = \parclassid{m}$.
Thus it does not establish that PBPO models AGREE. In fact, we have the following.

\begin{proposition}
    In $\Graph$ with $\MM$ the class of all monos, AGREE $\not\models$ PBPO.
\end{proposition}
\begin{proof}
    The AGREE rule $\rho$ given by
    \begin{center}
        {% scope
\newcommand{\nodexa}{\vertex{x_1}{cblue!20}}%
\newcommand{\nodexb}{\vertex{x_2}{cblue!20}}%
\newcommand{\nodey}{\vertex{y}{cgreen!20}}%
\newcommand{\nodez}{\vertex{z}{cred!10}}%
\newcommand{\nodeu}{\vertex{u}{cpurple!25}}%
\newcommand{\nodex}{\vertex{x}{cblue!20}}%
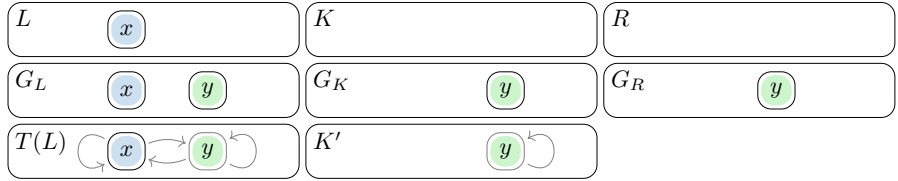
\begin{center}\vspace{0ex}
  \scalebox{\rulescale}{
  \begin{tikzpicture}[->,node distance=12mm,n/.style={}]
    \graphbox{$L$}{0mm}{0mm}{43mm}{8mm}{-4mm}{-4mm}{
      \node [npattern] (x)
      {\nodex};
    }
    \graphbox{$K$}{44mm}{0mm}{43mm}{8mm}{-4mm}{-4mm}{
    }
    \graphbox{$R$}{88mm}{0mm}{43mm}{8mm}{-8mm}{-4mm}{
    }
    \graphbox{$G_L$}{0mm}{-9mm}{43mm}{8mm}{-4mm}{-4mm}{
      \node [npattern] (x)
      {\nodex};
      \node [npattern] (y) [right of=x] {\nodey};
    }
    \graphbox{$G_K$}{44mm}{-9mm}{43mm}{8mm}{-4mm}{-4mm}{
      \node (x) {};
      \node [npattern] (y) [right of=x] {\nodey};
    }
    \graphbox{$G_R$}{88mm}{-9mm}{43mm}{8mm}{-8mm}{-4mm}{
      \node (x) {};
      \node [npattern] (y) [right of=x] {\nodey};
    }
    \graphbox{$T(L)$}{0mm}{-18mm}{43mm}{8mm}{-4mm}{-4mm}{
      \node [npattern] (x)
      {\nodex};
      \draw [eset,loop=180,looseness=3] (x) to node {} (x);
      \node [nset] (y) [right of=x] {\nodey};
      \draw [eset, bend left=20] (x) to node {} (y);
      \draw [eset, bend left=20] (y) to node {} (x);
      \draw [eset,loop=0,looseness=3] (y) to node {} (y);
    }
    \graphbox{$K'$}{44mm}{-18mm}{43mm}{8mm}{-4mm}{-4mm}{
        \node (x) {};
      \node [nset] (y) [right of=x] {\nodey};
      \draw [eset,loop=0,looseness=3] (y) to node {} (y);
    }
  \end{tikzpicture}
  }\vspace{0ex}
\end{center}%
}% scope
    \end{center}
    matches and deletes a single node $x$ and any of its incident edges, in any context. The context itself is preserved. Given the depicted $G_L$ and match $m$, the depicted $\alpha = \parclassid{m} : G_L \to T(L)$ is by definition the only possible adherence morphism. By contrast, PBPO allows an adherence morphism that maps $y$ onto $x$, so that $G_K$ and $G_R$ are both empty. So $\pbporulematchonly{\interpretas{\rho}{\mathrm{PBPO}}}{m} \not\subseteq  \AGREEsteprule{\rho}{m}$. Moreover, it is easy to see that the problem cannot be avoided by redefining the interpretation, because any redefinition will necessarily have $x$ in $L'$, with $x$ deleted in $K'$. \qed
\end{proof}

Similar arguments can be constructed for other categories satisfying the conditions of Proposition~\ref{prop:relating:AGREE:PBPO}, including $\Set$. For \pbpostrong, however, we have the following result.

\begin{proposition}[Relating AGREE and \pbpostrong{}]
    Assume $\CC$ has $\MM$-partial map classifiers for a stable system of monics $\MM$.
    In the diagram
    \begin{center}
        \begin{tikzcd}[column sep=45, row sep=30]
        \pmb{L} \arrow[d, "\eta_L" description, hook] & \pmb{K} \arrow[d, "\pmb{t_K}" description, hook, thick] \arrow[l, "\pmb{l}" description, thick] \arrow[r, "\pmb{r}" description, thick] & \pmb{R} \\
        T(L) \arrow[ru, "\mathrm{PB}", phantom] & \pmb{K'} \arrow[l, "{\parclass{t_K}{l}}" description]                           
        \end{tikzcd}
    \end{center}
    let the bold subdiagram depict an AGREE rule $\rho$, and the entire diagram a \pbpostrong rule $\interpretas{\rho}{\mathrm{PBPO}^{+}}$. Then for any match $m \in \MM$ and adherence $\alpha$ establishing a \pbpostrong step,
    $\AGREEsteprule{\rho}{m} = 
    {\pbpostrongrulematchonly{\interpretas{\rho}{\mathrm{PBPO}^{+}}}{m}}$.
\end{proposition}
\begin{proof}
    By virtue of the partial map classifier, $\parclassid{m}$ is the only adherence morphism establishing a strong match for \pbpostrong. \qed
\end{proof}

\begin{corollary}[\pbpostrong Models AGREE]
\label{corr:pbpostrong:models:agree}
    AGREE $\models$ \pbpostrong in any quasitopos. \qed
\end{corollary}

\subsection{\pbpostrong Models DPO}
\label{sec:pbpostrong:models:dpo}

The Double Pushout (DPO) approach to graph rewriting by Ehrig et al.~\cite{ehrig1973graph} is one of the earliest and most well studied algebraic graph rewriting methods.

\begin{definition}[DPO Rewriting~\cite{ehrig1973graph}]
    A \emph{DPO rewrite rule} $\rho$ is a span $L \stackrel{l}{\hookleftarrow} K \stackrel{r}{\to} R$. A diagram
    \begin{center}
        % https://tikzcd.yichuanshen.de/#N4Igdg9gJgpgziAXAbVABwnAlgFyxMJZABgBpiBdUkANwEMAbAVxiRABkQBfU9TXfIRQBGclVqMWbANLdeIDNjwEiAJjHV6zVohAAlOXyWCiZYeK1TdAcQD6nHkYEqRpc5sk6Qd2Y4X9lIWR1dwltNjsDLnEYKABzeCJQADMAJwgAWyQyEBwIJABmD3DdLOocOiwGNkgwVmoGOgAjGAYABQCTXSwwbFhDEDTMpFFc-MQciqq2AAsICABrRhwQYqsQaobm1o7jFxAevtY-IazEUbykdTD11NWN7fbO-cOsfpP0s4vxgBYP4cQP3K4wK-zOQLGSAArGCrsDobDzvDEEUbl4ADrojJ0HAzVIZYBtADyXHujRaTz2QgOvTe9Q2PS8tXpMxgdCgNQI9Km1V0zIGpxGyKhawxWJxeIJxNJWwpu2c1Ne-QajM5dXurPZau5lV54C53AoXCAA
        \begin{tikzcd}
        L \arrow[d, "m" description] & K \arrow[l, "l" description, hook'] \arrow[r, "r" description] \arrow[d] \arrow[ld, "\mathrm{PO}", phantom] \arrow[rd, "\mathrm{PO}", phantom] & R \arrow[d] \\
        G_L                          & G_K \arrow[l] \arrow[r]                                                                                                                        & G_R        
        \end{tikzcd}
    \end{center}
    defines a \emph{DPO rewrite step} $\DPOstep{G_L}{G_R}{\rho}{m}$, i.e., a step from $G_L$ to $G_R$ using rule $\rho$ and match morphism $m : L \to G_L$.
\end{definition}

\begin{definition}[DPO as \pbpostrong]
    Let $\CC$ be a category with $\MM$-partial map classifiers, in which pushouts along $\MM$-morphisms exist, are pullbacks, and where $\MM$-morphisms are stable under pushout. 
    
    In the diagram
    \begin{center}
        % https://tikzcd.yichuanshen.de/#N4Igdg9gJgpgziAXAbVABwnAlgFyxMJZABgBpiBdUkANwEMAbAVxiRABkQBfU9TXfIRQBGclVqMWbANLdeIDNjwEiAJjHV6zVohAAlOXyWCiZYeK1Td7AOSGF-ZUOSjzmyTpAAVABTSAlNziMFAA5vBEoABmAE4QALZIZCA4EEgAzO7abDgA+pzUOHRYDGwAFhAQANYg1Ax0AEYwDAAKjia6WGDYsPaxCUiiKWmIyUUl5ZVVjDi1Etm6pXWNzW3GKp3dWL080XGJiEOpSOrzViAxc-VNre0bIF09rLsg-QdHIwAshcWluhXVOaWTwAHRBMCKuVkyxuawE90e22e8jeSG+wwyWXODDsMNWdyEDy2vR+E3+UxmfX2g0KI0yZ1BIPidBwZRi8WALQA8lwritbutCYiSSAGF1PJAwKxqGUYHQoGxJdKUr9FQRnhQuEA
        \begin{tikzcd}
        \pmb{L} \arrow[d, "t_L" description, hook] & \pmb{K} \arrow[l, "\pmb{l}" description, hook', thick] \arrow[r, "\pmb{r}" description, thick] \arrow[d, "\eta_K" description, hook] \arrow[ld, "\mathrm{PO}", phantom] & \pmb{R} \\
        L'                             & T(K) \arrow[l, "l'" description, hook']                                                                                                                            &  
        \end{tikzcd}
    \end{center}
    let the top span depict a DPO rule $\rho$, and the left square a pushout (which is a pullback). Then the entire diagram defines a \pbpostrong rule $ \interpretas{\rho}{\mathrm{PBPO}^{+}}$.
\end{definition}

\begin{theorem}[\pbpostrong Models DPO]
    \label{thm:pbpostrong:models:dpo}
    Let $\CC$ be a quasitopos. Then for any DPO rule $\rho$, $\interpretas{\rho}{\mathrm{PBPO}^{+}}$ is well-defined, and for any $\MM$-morphism $m : L \MMmono G_L$, we have
    $\DPOsteprule{\rho}{m} = 
    {\pbpostrongrulematchonly{\interpretas{\rho}{\mathrm{PBPO}^{+}}}{m}}$.
\end{theorem}
\begin{proof}
    Well-definedness of $\interpretas{\rho}{\text{\pbpostrong}}$ follows from Definition~\ref{def:m:adhesive:category}, Lemma~\ref{lemma:m:adhesive:pushouts:are:pullbacks}, and the fact that any quasitopos is $\MM$-adhesive for $\MM = \RegC$.
    
    Direction $\subseteq$:
    Suppose that
    \begin{center}
        % https://tikzcd.yichuanshen.de/#N4Igdg9gJgpgziAXAbVABwnAlgFyxMJZABgBpiBdUkANwEMAbAVxiRABkQBfU9TXfIRQBGclVqMWbANLdeIDNjwEiAJjHV6zVohAAlOXyWCiZYeK1TdAcQD6nHkYEqRpc5sk6Qd2Y4X9lIWR1dwltNjsDP0VnINFVC082ABUACmkASkN-YxcSUgSPcN12AHJucRgoAHN4IlAAMwAnCABbJFEQHAgkMi66LAY2AAsICABrRhwQIqsQIeoGOgAjGAYABQCTXSwwbFhs5rbe6m6kAGZTgaHdUYmZsLn2xZW1zdyhEF391j8j9sQnTOiAALFdBiMxuMHpYvEwHktVhsti4vnssAc-i0AUCeoh1I8vE0EbsvHAIAwMQ9hjA6FAkGAmAwGOCbuACKwXkj3rE2E0sNVhtMscd8ac8QBWWZeADuDxw1zYkDAnJANLpDKZLPmr2RHzY3wxv3k-yQYK6kulbGqtgMXLeKM+hsxJuxZvFF1ZkImUxhSV0Ns49r1vJ26JdjTdiD6wIA7NRVmB6YgALQggCcVt0OHsCN1PMCBvDqoVENuUMOUYAbB7EPHCWwGOUvbplarEQ79WGfpXRebgTWG7oADrDtB0JoAYyWcH2wCYXDz3MdRZ7IpxtcuQ5Ao9adBwwyarWA6wA8ovFqSlRzqbTk+yVfLFa2b8GC9s0WvXaLcUgpdvd33Q9jzPC95ivF9H2odV7zbJ8ywfdt8xXbsjV7DcLSQQdE2Tc4+lhNhRxgBVbFkN8UM-NCWzVCsuAoLggA
        \begin{tikzcd}[column sep=40]
        L \arrow[d, "m" description, hook] & K \arrow[l, "l" description, hook'] \arrow[d, "u" description] \arrow[r, "r" description] \arrow[ld, "\mathrm{PO}", phantom] \arrow[rd, "\mathrm{PO}", phantom] & R \arrow[d, "w" description] \\
        G_L                                & G_K \arrow[r, "g_R" description] \arrow[l, "g_L" description]                                                                                                   & G_R                         
        \end{tikzcd}
    \end{center}
    is a DPO step induced by $\rho$. 
    Then by Lemma~\ref{lemma:m:adhesive:pushouts:are:pullbacks}, the left square is a pullback, and by stability of $\MM$-morphisms under pushout and pullback, $u, g_L \in \MM$.
    
    Because $u \in \MM$, we obtain the classifying arrow $\parclassid{u} : G_K \to T(K)$ that makes the square $\parclassid{u} \circ u = \eta_K  \circ \id{K}$ a pullback. Additionally, from the pushout property of the top left square, we obtain a unique morphism $\alpha : G_L \to L'$ satisfying $t_L = \alpha \circ m$ and $\alpha \circ g_L = l' \circ \parclassid{u}$:
    \begin{equation}
    \label{eq:dpo:to:pbpostrong:diagram}
        % https://tikzcd.yichuanshen.de/#N4Igdg9gJgpgziAXAbVABwnAlgFyxMJZABgBpiBdUkANwEMAbAVxiRABkQBfU9TXfIRQBGclVqMWbANLdeIDNjwEiAJjHV6zVohAAlOXyWCiZYeK1TdAcQD6nHkYEqRpc5sk6Qd2Y4X9lIWR1dwltNjsDP0VnINFVC082ABUACmkASkN-YxcSUgSPcN12AHJucRgoAHN4IlAAMwAnCABbJFEQHAgkMi66LAY2AAsICABrRhwQIqsQIeoGOgAjGAYABQCTXSwwbFhs5rbe6m6kAGZTgaHdUYmZsLn2xZW1zdyhEF391j8j9sQnTOiAALFdBiMxuMHpYvEwHktVhsti4vnssAc-i0AUCeoh1I8vE0EbsvHAIAwMQ9hjA6FAkGAmAwGOCbuACKwXkj3rE2N8Mb95P8kATgQBWWZeADuDxw1zYkDAnJANLpDKZLPmr2RHz56MxQuxSDBXTxEsJbGqtgMXLeKM+-INjSNoNOeMu-QhtyhUxhSV0Vs4tp1vJ2+sFzuOiD6wIA7NRVmB6YgALQggCckrYOHsCO1PMCep+svl3omvqxUYAbG6kPGLboGOVWZDywxpsGC9s0cXKwCTcCaw2QAAdEdoOhNADGSzg+2ATC4ee59qLAsOLtxFyzujHrToOGGTVawHWAHkl4tSQqOdTacn2UqS17H8rEXbdWHe4ao1vEObYTYPcDyPE9z0veZr10RVlVVB8YOfNkEM7Vcv3XPsOlrRAh0TZNzj6QDdxHGA5VsWQUM-Ht1xbMtoQwxAPTjK8n10KAIBwHAqj9YpRxHRg0GGOhlw-UMqMxCguCAA
        \begin{tikzcd}[column sep=40]
        L \arrow[d, "m" description, hook] \arrow[dd, "t_L" description, hook', bend right=49] & K \arrow[l, "l" description, hook'] \arrow[d, "u" description, hook] \arrow[r, "r" description] \arrow[ld, "\mathrm{PO}", phantom] \arrow[rd, "\mathrm{PO}", phantom] & R \arrow[d, "w" description] \\
        G_L     \arrow[d, "\alpha" description, dotted]                                                                               & G_K \arrow[r, "g_R" description] \arrow[l, "g_L" description, hook'] \arrow[d, "\parclassid{u}" description]                                                                                 & G_R                          \\
        L'                                                                                    & T(K) \arrow[l, "l'" description, hook'] \arrow[from=uu, "\eta_K" description, hook, bend left=49, pos=0.75, crossing over]                                                                                                                                                &                             
        \end{tikzcd}
    \end{equation}
    By the dual of the pullback lemma, it follows that the bottom left commuting square is a pushout. Because it is a pushout along an $\MM$-morphism, it is also a pullback.
    
    It remains to show that square $t_L \circ \id{L} = \alpha \circ m$ is a pullback. For this, consider the cubical arrangement
    \begin{center}
        \begin{tikzcd}[row sep={32,between origins},column sep={32,between origins}]
            & K  \ar[rr, equals] \ar[dd, hook, "u" description, pos=0.75] \ar[dl, "l" description, hook'] & &  K \ar[dd, hook, "\eta_K" description] \ar[dl, "l" description, hook'] \\
            L \ar[rr, crossing over, equals] \ar[dd, hook, "m" description] & & L \\
              & G_K \ar[rr, "\parclassid{u}" description, pos=0.25] \ar[dl, hook', "g_L" description] & &  T(K) \ar[dl, "l'" description, hook'] \\
            G_L \ar[rr, "\alpha" description] && L' \ar[from=uu,crossing over, hook, "t_L" description, pos=0.75]
        \end{tikzcd}
    \end{center}
    in which the bottom square is a pushout along an $\MM$-morphism, the back faces are pullbacks and the top square is a pushout. By $\MM$-adhesivity, the bottom square is an $\MM$-VK square. From this it follows that the front face is a pullback square.
    Thus Diagram~\eqref{eq:dpo:to:pbpostrong:diagram} defines a \pbpostrong step.
    
    Direction $\supseteq$: We are given a \pbpostrong step
    \begin{equation}
    \label{eq:pbpostrong:step:to:dpo:step}
        % https://tikzcd.yichuanshen.de/#N4Igdg9gJgpgziAXAbVABwnAlgFyxMJZARgBoAGAXVJADcBDAGwFcYkQAZEAX1PU1z5CKAEwVqdJq3YBpHnxAZseAkQDM4mgxZtEIAErz+yoUTLEJ26XoDiAfS69jg1aNIWtU3SHtynigRVhZA0PSR12e0N-JRdgsRFLL3YAFQAKGQBKIwCTVxJSRM8IvQ4AchzYoKJydySSzkrA0xRaovDrRu4JGCgAc3giUAAzACcIAFskMhAcCCRa2fosRnYACwgIAGsmHBBiztWaRnoAIxhGAAVm1xAsMGxYHLHJhZo5pA0llfXNrf2Ot4psczhdrnlhHcHlgnv4XlNEDMPogACzvZarPQbbYAqzeZgAk7nK43SH3R5sOHjBFI+aIMSA9ijQn3bxwCCMGEAtYwehQJBgZiMRjon56SBgNgg4nguLsckwykKeFIBnIgCsB28AHcATgMewJVKQDy+QKhSKQESwaT5dDYcrqUg0bM6ZrGXo+nZDNKbRC7RTnk7Ue86V99WKTX9drjkp6HITQST-XoFQ6RsGAGyhpAAdi17EYFVFmKj2xjvuTctT9qVGdeIddSGzHpAzGLVqTsuqNcDVIbtM+x1ZhoIxtN-PFY71BqnksTMttvcVsYaAB01xN6Dg1qMJsBLgAhbhBgc5xDuvHsDdbnd7g8AeRPw-nc-HvMn4GnJdH88r3ZaKE+0dM8m0QFtzjASc1EWK89A3GB9TsOR-yXICVx-LE-lPBFwzpfMrRHPR2U5fkC3gtcmDQNZ6AXP1q3Qp4aAnc1hUwr9537BEAA5z3DWcy3+ciQGBTtFxTRi6xAFVEAATnPAi4NmBNUIktNjQjUtsX+LikF4sD5MI18QCgCBmFORh3zNN8Z0jI0cOdBThJvbdd33I9nzE+ie0klljPs5iP1-DSBPs3TEH05FFLjEAXLvdzjz87wApNIKbPYlLrSrHz1J4ShuCAA
        \begin{tikzcd}[column sep=40]
                                                                                                             & L \arrow[d, "m" description, hook]  & K \arrow[l, "l" description, hook'] \arrow[d, "u" description, hook] \arrow[r, "r" description] \arrow[ld, "\mathrm{PB}", phantom] \arrow[rd, "\mathrm{PO}", phantom]  & R \arrow[d, "w" description] \\
        L \arrow[r, "m" description, hook] \arrow[d, equals] \arrow[rd, "\mathrm{PB}", phantom] & G_L \arrow[d, "\alpha" description] & G_K \arrow[r, "g_R" description,pos=0.6] \arrow[l, "g_L" description, hook'] \arrow[d, "u'" description] \arrow[ld, "\mathrm{PB}", phantom]                                                                         & G_R                          \\
        L \arrow[r, "t_L" description, hook]                                                                       & L'                                  & T(K) \arrow[l, "l'" description, hook']  \arrow[from=uu, "\eta_K" description, hook, bend left=49, pos=0.75, crossing over]                                                                                                                                                                                       &                             
        \end{tikzcd}
    \end{equation}
    where the pullback squares can be represented as the commutative cube
    \begin{center}
        \begin{tikzcd}[row sep={32,between origins},column sep={32,between origins}]
            & K  \ar[rr, "u" description, hook] \ar[dd, equals] \ar[dl, "l" description, hook'] & &  G_K \ar[dd, "u'" description] \ar[dl, "g_L" description, hook'] \\
            L \ar[rr, crossing over, "m" description, pos=0.25, hook] \ar[dd, equals] & & G_L \\
              & K \ar[rr, hook, "\eta_K" description, pos=0.27] \ar[dl, "l" description, hook'] & &  T(K) \ar[dl, "l'" description, hook'] \\
            L \ar[rr, "t_L" description, hook] && L' \ar[from=uu,crossing over, "\alpha" description, pos=0.75]
         \end{tikzcd} .
    \end{center}
    
    By Lemma~\ref{lemma:on:u}, we know that the back face is a pullback square. Because the floor is a pushout, and the vertical faces are all pullbacks, it follows that the top face of the cube is a pushout, using the fact that pushouts are stable under pullback in rm-quasiadhesive categories. Thus the top row of Diagram~\eqref{eq:pbpostrong:step:to:dpo:step} defines a DPO step.
    \qed
\end{proof}

\begin{theorem}
    \label{thm:pbpostrong:models:others}
    Let $\CC$ be a quasitopos, and let matches $m$ be regular monic. Then:
    \begin{center}
    \begin{tikzpicture}[default,nodes={rectangle,inner sep=3mm},baseline=(l.base)]
        \node (t) {\pbpostrong};
        \node (l) at (t.west) [anchor=east,yshift=-3.5mm]  
          {$\text{SqPO} \models \text{AGREE}$};
        \node (r) at (t.east) [anchor=west,xshift=-2mm,yshift=-3.5mm]  
          {$\text{DPO}$};
        \node (b) at ($(t)+(0,-7mm)$) {PBPO};
        \node at ($(l.north east)!0.5!(t.south west)$) [rotate=35] {$\models$};
        \node at ($(r.north west)!0.5!(t.south east) + (-.5mm,0mm)$) [rotate=180-35] {$\models$};
        \node at ($(t)!0.5!(b)$) [rotate=90] {$\models$};
    \end{tikzpicture}
    \end{center}
\end{theorem}
\begin{proof}
    By SqPO $\models$ AGREE~\cite[Theorem 2]{corradini2015agree}, Corollary~\ref{corr:pbpostrong:models:pbpo},
    Corollary~\ref{corr:pbpostrong:models:agree} and Theorem~\ref{thm:pbpostrong:models:dpo}.
\qed
\end{proof}

Moreover, we have the following negative result for the other directions.

\begin{proposition}
\label{prop:pbpostrong:cannot:be:modeled:by:sets:of:rules}
    In $\Graph$, which is a quasitopos, not every \pbpostrong{} rule can be modeled by a set of $\FF \in \{ \mathrm{AGREE}, \mathrm{PBPO}, \mathrm{DPO} \}$ rules.
\end{proposition}
\begin{proof}
    The \pbpostrong{} rule
    \begin{center}
        % diagram for the proof to Proposition 74

{% scope
\newcommand{\nodexa}{\vertex{x_1}{cblue!20}}%
\newcommand{\nodexb}{\vertex{x_2}{cblue!20}}%
\newcommand{\nodey}{\vertex{y}{cgreen!20}}%
\newcommand{\nodez}{\vertex{z}{cred!10}}%
\newcommand{\nodeu}{\vertex{u}{cpurple!25}}%
\newcommand{\nodex}{\vertex{x}{cblue!20}}%
\begin{center}\vspace{0ex}
  \scalebox{\rulescale}{
  \begin{tikzpicture}[->,node distance=12mm,n/.style={}]
    \graphbox{$L$}{0mm}{0mm}{35mm}{8mm}{-4mm}{-4mm}{
      \node [npattern] (x)
      {\nodex};
    }
    \graphbox{$K$}{36mm}{0mm}{35mm}{8mm}{-4mm}{-4mm}{
    }
    \graphbox{$R$}{72mm}{0mm}{35mm}{8mm}{-8mm}{-4mm}{
    }
    \graphbox{$L'$}{0mm}{-9mm}{35mm}{8mm}{-4mm}{-4mm}{
      \node [npattern] (x)
      {\nodex};
      \node [nset] (y) [right of=x] {\nodey};
    }
    \graphbox{$K'$}{36mm}{-9mm}{35mm}{8mm}{-4mm}{-4mm}{
      \node [] (x) {};
      \node [nset] (y) [right of=x] {\nodey};
    }
    \transparentgraphbox{$R'$}{72mm}{-9mm}{35mm}{8mm}{-8mm}{-4mm}{
      \node [] (x)
      {};
      \node [nset] (y) [right of=x] {\nodey};
    }
  \end{tikzpicture}
  }\vspace{0ex}
\end{center}%
}% scope
    \end{center}
    serves to prove all three claims. It specifies the deletion of a single node in a host graph without edges. AGREE and DPO cannot model this rule even with an infinite set of rules, because they cannot express that the host graph cannot have edges. Although PBPO can express this constraint, it cannot prevent rules from mapping the entire host graph onto $x$, deleting all nodes at once.
    \qed 
\end{proof}

\section{Category \Graphleq}
\label{sec:category:graph:lattice}

Unless one employs a meta-notation or restricts to unlabeled graphs, as we did in Section~\ref{sec:pbpostrong}, it is sometimes impractical to use \pbpostrong in the category \Graph. The following example illustrates the problem.
\medskip

\noindent
\begin{minipage}[c]{0.78\textwidth}
    \begin{example}\label{ex:pbpostrong:cumbersome:in:graph}
        Suppose the set of labels is $\labels = \{0,1\}$. To be able to injectively match pattern
        $
        L = % https://tikzcd.yichuanshen.de/#N4Igdg9gJgpgziAXAbVABwnAlgFyxMJZABgBpiBdUkANwEMAbAVxiRGJAF9T1Nd9CKAIzkqtRizYdOYmFADm8IqABmAJwgBbJGRA4ISEeOatEIIVwqcgA
        \begin{tikzcd}[column sep=0.5cm]
        0 \arrow[r, "1"] & 0
        \end{tikzcd}
        $
        in any context, one must inject it into the type graph $L'$ shown on the right
        in which every dotted loop represents two edges (one for each label), and every dotted non-loop represents four edges (one for each label, in either direction). 
        For general $\labels$, to allow any context, one needs to include $|\labels|$ additional vertices in $L'$, and $|\labels|$ complete graphs over $V_{L'}$.
    \end{example}
\end{minipage}\hfill%
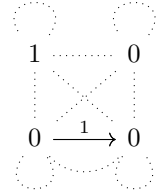
\begin{minipage}[c]{0.2\textwidth}
    \hfill
    % https://tikzcd.yichuanshen.de/#N4Igdg9gJgpgziAXAbVABwnAlgFyxMJZABgBoBGAXVJADcBDAGwFcYkRiQBfU9TXfIRTkK1Ok1btOPPtjwEiI4mIYs2iDt14gMcwUTLKaqyRvLcxMKAHN4RUADMAThAC2SMiBwQkI8WvZzGRBnNw8ab19g0PdEAGYIn0RPAAsYeih2SDA2GkYsHPYoCBwcKy1HF1jPSOS8iAg0A1IHJjgYMUZ6ACMYRgAFfnkhECcsaxScEBo0jKyCXJB8wo1i0vLoqt9E7aWGppQyVsZ2zp6+wb0FDTGJqZn0zI1sxeX1EDWyzM2wxAAmHaIPxvIolL7TECzJ7gBYVEJbeKA4EFd6fcoPObPWE-WIArxJTwg1Zg9GQx7zQo48L43a9MBPAC0ABYAJx5FGg9aZDHQl5wmJIBI0-7slYfEncsmYmGU7QCxHCoWMfaKAAcRzaHTy5wGQ30N3GkwhRPFXIhUIpbCpIuFeOVjTVGpOWqWOsuAmuo0N9yWHOJZp5losXCAA
    \begin{tikzcd}
    1 \arrow[d, no head, dotted] \arrow[rd, no head, dotted] \arrow[r, no head, dotted] \arrow[no head, dotted, loop, distance=2em, in=125, out=55] & 0 \arrow[d, no head, dotted] \arrow[ld, no head, dotted] \arrow[no head, dotted, loop, distance=2em, in=125, out=55] \\
    0 \arrow[r, "1"] \arrow[r] \arrow[no head, dotted, loop, distance=2em, in=305, out=235] \arrow[r, no head, dotted, bend right=49]               & 0 \arrow[no head, dotted, loop, distance=2em, in=305, out=235]                                                      
    \end{tikzcd}\hspace*{-4mm}%
\end{minipage}
\medskip

Beyond this example, and less easily alleviated with meta-notation, in \Graph it is impractical or  impossible to express rules that involve (i)~arbitrary labels (or classes of labels) in the application condition; (ii)~relabeling; or (iii)~allowing and 
capturing arbitrary subgraphs (or classes of subgraphs) around a match graph. As we will discuss in Section~\ref{sec:discussion}, these features have been non-trivial to express in general for algebraic graph rewriting approaches.

We define a category which allows flexibly addressing all of these issues. Recall the definition of a complete lattice (Definition~\ref{definition:complete:lattice}).

\begin{definition}[\GraphLattice{(\labels,\leq)}]
    For a complete lattice $(\labels, \leq)$, we define the category \GraphLattice{(\labels,\leq)}, where objects are graphs labeled from $\labels$, and  arrows are graph premorphisms $\phi : G \to G'$ that satisfy $\lbl_G(x) \leq \lbl_{G'}(\phi(x))$ for all $x \in V_G \cup E_G$.
\end{definition}

In terms of graph structure, the pullbacks and pushouts in \GraphLattice{(\labels,\leq)} are the usual pullbacks and pushouts in $\Graph$. The only difference is that the labels that are identified by the cospan of the pullback (resp.\ span of the pushout) are replaced by their meet (resp.\ join).

\begin{remark}[Fuzzy Graph Rewriting]
    The idea to label graphs using labels that form a complete lattice, for the purpose of rewriting, is not new. To the best of our knowledge, Mori and Kawahara were the first to propose this~\cite{mori1995fuzzy}, using a single pushout construction. There also exists a series of papers by Parasyuk and Yershov, in which fuzzy graph transformations are studied using single pushouts~\cite{parasyuk2006categorical} and double pushouts~\cite{parasyuk2007categorical, parasyuk2008transformational}. Because \pbpostrong{} rules have both a pattern span and a type span (unlike in the single and double pushout approaches), both lower and upper bounds can be specified on fuzzy values. Moreover, that fuzzy graphs lend themselves well for general relabeling purposes has not yet been observed.
\end{remark}

\begin{remark}
    Analogous to the situation for $\labels$-fuzzy sets (Example~\ref{example:L:fuzzy:set}), we have proven very recently that \GraphLattice{(\labels,\leq)} is a quasitopos if $\labels$ is a complete Heyting algebra~\cite{rosset2023fuzzy}.
\end{remark}

One very simple but useful complete lattice is the following.

\begin{definition}[Flat Lattice]
    Let $\labels^{\bot,\top} = \labels \uplus \{ \bot, \top \}$. We define the \emph{flat lattice} induced by $\labels$ as the poset $(\labels^{\bot,\top}, {\leq})$,
    in which 
    $\bot < l < \top$ for all $l \in \labels$ are the only non-trivial relations. Here, we refer to $\labels$ as the \emph{base label set}.
\end{definition}

One feature flat lattices provide is a kind of ``wildcard element'' $\top$.
\medskip

\noindent
\begin{minipage}[c]{0.65\textwidth}
    \begin{example}[Wildcards]
        Using flat lattices, $L'$ of Example~\ref{ex:pbpostrong:cumbersome:in:graph} can be fully expressed for any base label set $\mathcal{L} \ni 0,1$ as shown on the right (node identities are omitted).
        The visual syntax and naming shorthands of PGR~\cite{overbeek2020patch} (or variants thereof) could be leveraged to simplify the notation further.
    \end{example}
\end{minipage}\hfill%
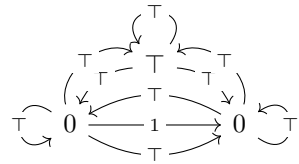
\begin{minipage}[c]{0.31\textwidth}
    \hfill
    % https://tikzcd.yichuanshen.de/#N4Igdg9gJgpgziAXAbVABwnAlgFyxMJZABgBoBGAXVJADcBDAGwFcYkRiQBfU9TXfIRQAmCtTpNW7Tjz7Y8BIuVLFxDFm0QgAOtpwQ03cTCgBzeEVAAzAE4QAtkjIh9SZRI3tyIGo3oAjGEYABX4FIRAsMGxYbl4QWwckURcIZN8IAyUADjIrJjgYcT9AkLDBdiiYthp1KS1dfUNZBLtHRBTXRGc6zR09LN8AoND5Cq0qrFiaQLAoJABaAGZiFsT2zrTEd172RsGQEpHyxQnoqZqQWfnEFbW2pxoulOukADYAdlrJPv3DIdKowEp0i51i9yS2yeWxSRzKYxBk2mHnq-SaPiuMDmi0+EPazi67jhQPClTBl12DQG-0x2MQy1W8XWbmhj1pNwZ308VPRAOOCIiSLYeJZqVFjEyhhQAE48gUinz4cDBeSMZS0VkRd1WdrDpKiMIPnLGIVisMlaSztU1T89tSjFwgA
    \begin{tikzcd}[row sep=0.3cm, column sep=0.6cm]
                                                                                                                                                                                          & \top \arrow["\top" description, loop, distance=2em, in=125, out=55] \arrow[ld, "\top" description, bend right] \arrow[rd, "\top" description, bend left] &                                                                                                                                                           \\
    0 \arrow[rr, "1" description] \arrow[ru, "\top" description, bend left=67] \arrow[rr, "\top" description, bend right] \arrow["\top" description, loop, distance=2em, in=215, out=145] &                                                                                                                                                          & 0 \arrow[lu, "\top" description, bend right=67] \arrow[ll, "\top" description, bend right] \arrow["\top" description, loop, distance=2em, in=35, out=325]
    \end{tikzcd}
\end{minipage}
\medskip

As the following example illustrates, the expressive power of a flat lattice stretches beyond wildcards: it also enables relabeling of graphs. (Henceforth, we will depict a node $x$ with label $u$ as $x^u$.)

\begin{example}[Relabeling]
\label{example:hard:overwriting}
  As vertex labels we employ the flat lattice induced by the set $\{\, a,b,c,\ldots \,\}$, and assume edges are unlabeled for notational simplicity.
  The diagram 
    {% scope

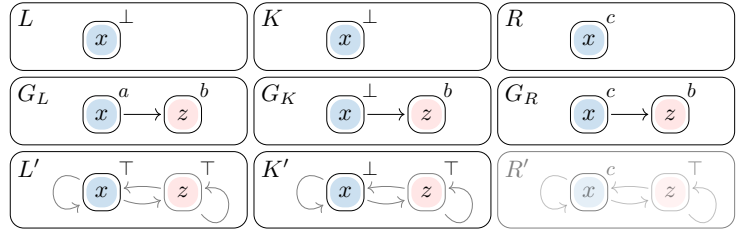
\newcommand{\nodexa}{\vertex{x_1}{cblue!20}}
\newcommand{\nodexb}{\vertex{x_2}{cblue!20}}
\newcommand{\nodexc}{\vertex{x_3}{cblue!20}}
\newcommand{\nodexd}{\vertex{x_4}{cblue!20}}

\newcommand{\nodea}{\vertex{a}{cgreen!20}}
\newcommand{\nodeb}{\vertex{b}{cpurple!25}}

\newcommand{\nodeaa}{\vertex{a_1}{cgreen!20}}
\newcommand{\nodeab}{\vertex{a_2}{cgreen!20}}
\newcommand{\nodeba}{\vertex{b_1}{cpurple!25}}
\newcommand{\nodebb}{\vertex{b_2}{cpurple!25}}

\newcommand{\nodex}{\vertex{x}{cblue!20}}

\newcommand{\nodez}{\vertex{z}{cred!10}}

\begin{center}
  \scalebox{\rulescale}{
  \begin{tikzpicture}[->,node distance=12mm,n/.style={}]
    \graphbox{$L$}{0mm}{0mm}{35mm}{10mm}{-4mm}{-5.5mm}{
      \node [npattern] (x)
      {\nodex};
       \annotate{x}{$\bot$};
    }
    \graphbox{$K$}{36mm}{0mm}{35mm}{10mm}{-4mm}{-5.5mm}{
      \node [npattern] (x)
      {\nodex};
      \annotate{x}{$\bot$};
    }
    \graphbox{$R$}{72mm}{0mm}{35mm}{10mm}{-4mm}{-5.5mm}{
      \node [npattern] (x)
      {\nodex};
       \annotate{x}{$c$};
    }
    \graphbox{$G_L$}{0mm}{-11mm}{35mm}{10mm}{-4mm}{-5.5mm}{
      \node [npattern] (x)
      {\nodex};
       \annotate{x}{$a$};
       \node [npattern] (z) [right of=x] {\nodez};
       \annotate{z}{$b$};
        \draw [epattern] (x) to node {} (z);
    }
    \graphbox{$G_K$}{36mm}{-11mm}{35mm}{10mm}{-4mm}{-5.5mm}{
      \node [npattern] (x)
      {\nodex};
      \annotate{x}{$\bot$};
      \node [npattern] (z) [right of=x] {\nodez};
       \annotate{z}{$b$};
        \draw [epattern] (x) to node {} (z);
    }
    \graphbox{$G_R$}{72mm}{-11mm}{35mm}{10mm}{-4mm}{-5.5mm}{
      \node [npattern] (x)
      {\nodex};
      \annotate{x}{$c$};
      \node [npattern] (z) [right of=x] {\nodez};
       \annotate{z}{$b$};
        \draw [epattern] (x) to node {} (z);
    }
    \graphbox{$L'$}{0mm}{-22mm}{35mm}{11.3mm}{-4mm}{-6mm}{
      \node [npattern] (x)
      {\nodex};
       \annotate{x}{$\top$};
        \draw [eset,loop=180,looseness=3] (x) to node {} (x);
        \node [nset] (z) [right of=x] {\nodez};
       \annotate{z}{$\top$};
        \draw [eset] (x) to [bend right=20] node {} (z);
         \draw [eset] (z) to [bend right=10] node {} (x);
         \draw [eset,loop=-20,looseness=3] (z) to node {} (z);
    }
    \graphbox{$K'$}{36mm}{-22mm}{35mm}{11.3mm}{-4mm}{-6mm}{
      \node [npattern] (x)
      {\nodex};
       \annotate{x}{$\bot$};
        \draw [eset,loop=180,looseness=3] (x) to node {} (x);
        \node [nset] (z) [right of=x] {\nodez};
       \annotate{z}{$\top$};
        \draw [eset] (x) to [bend right=20] node {} (z);
         \draw [eset] (z) to [bend right=10] node {} (x);
         \draw [eset,loop=-20,looseness=3] (z) to node {} (z);
    }
    \transparentgraphbox{$R'$}{72mm}{-22mm}{35mm}{11.3mm}{-4mm}{-6mm}{
      \node [npattern] (x)
      {\nodex};
       \annotate{x}{$c$};
        \draw [eset,loop=180,looseness=3] (x) to node {} (x);
        \node [nset] (z) [right of=x] {\nodez};
       \annotate{z}{$\top$};
        \draw [eset] (x) to [bend right=20] node {} (z);
         \draw [eset] (z) to [bend right=10] node {} (x);
         \draw [eset,loop=-20,looseness=3] (z) to node {} (z);
    }
  \end{tikzpicture}
  }
\end{center}

}% scope
  \noindent displays a rule ($L,L',K,K',R$)  for overwriting an arbitrary vertex's label with $c$, in any context. The middle row is an application to a host graph $G_L$.
\end{example}

Example~\ref{example:hard:overwriting} demonstrates how (i)~labels in $L$ serve as lower bounds for matching, (ii)~labels in $L'$ serve as upper bounds for matching, (iii)~labels in $K'$ can be used to decrease matched labels (in particular, $\bot$ ``instructs'' to ``erase'' the label by overwriting it with $\bot$, and $\top$ ``instructs'' to preserve labels), and (iv)~labels in $R$ can be used to increase labels. First erasing a label in $K'$ and then increasing it using another label in $R$ effectively establishes an arbitrary relabeling from $G_L$ to $G_R$.

Complete lattices also support modeling sorts.

\begin{example}[Sorts]
    Let $p_1, p_2,\ldots \in \mathbb{P}$ be a set of processes and $d_1,d_2,\ldots \in \mathbb{D}$ a set of data elements. Assume a complete lattice over labels $\mathbb{P} \cup \mathbb{D} \cup \{ \mathbb{P}, \mathbb{D}, \rhd, @ \}$, arranged as in the diagram
    \[
        \forall i \in \mathbb{N}:\\[.5ex]
        % https://tikzcd.yichuanshen.de/#N4Igdg9gJgpgziAXAbVABwnAlgFyxMJZABgBoBGAXVJADcBDAGwFcYkQAdDgW3pwAsARoOAAFAL4hxpdJlz5CKMgCZqdJq3ZoA+likyQGbHgJFyFNQxZtEnHnyEiAIpOmzjCs6VU0rm21C6+u7ypijKFr4aNnYATvxQwYZyJorI5gDMltHsXIIQOElGoWnmxNnWuRw4EGhFKZ4oGZHqlbYAAlJqMFAA5vBEoABmsRDcSOYgNUjEbiAjY0jNUxBIynML44gArDTTiAAsNIIwYImIALQZswabSLsrSxujWw-75M+LiGSPiABsny2EV+ANuLyQRxBgPue1WiAA7MdTucrjdhuCEbCkH8kWcZuJKOIgA
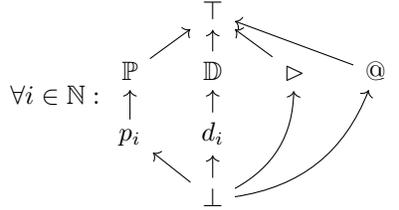
        \begin{tikzcd}[column sep=5mm,row sep=0.3cm]
                              & \top                                                                       &                 &               \\
        \mathbb{P} \arrow[ru] & \mathbb{D} \arrow[u]                                                       & \rhd \arrow[lu] & @ \arrow[llu] \\
        p_i \arrow[u]         & d_i \arrow[u]                                                              &                 &               \\
                              & \bot \arrow[ruu, bend right] \arrow[u] \arrow[lu] \arrow[rruu, bend right] &                 &              
        \end{tikzcd} .
    \]
    Moreover, assume that the vertices $x,y,\ldots$ in the graphs of interest are labeled with a $p_i$ or $d_i$, and that edges are labeled with a $\rhd$ or $@$. In such a graph,
    \begin{itemize}
        \item 
          an edge $x^{d_i} \xrightarrow{@} y^{p_j}$ encodes that process $p_j$ holds a local copy of datum $d_i$ ($x$ will have no other connections); and
        \item 
          a chain of edges $x^{p_i} \xrightarrow{\rhd} y^{d_k} \xrightarrow{\rhd} z^{d_l} \xrightarrow{\rhd} \cdots \xrightarrow{\rhd} u^{p_j}$ encodes a directed FIFO channel from process $p_i$ to process $p_j \neq p_i$, containing a sequence of elements $d_k, d_l, \ldots$. An empty channel is modeled as $x^{p_i} \xrightarrow{\rhd} u^{p_j}$.
    \end{itemize}
    Receiving a datum through an incoming channel (and storing it locally) can be modeled using the following rule:
        {% scope

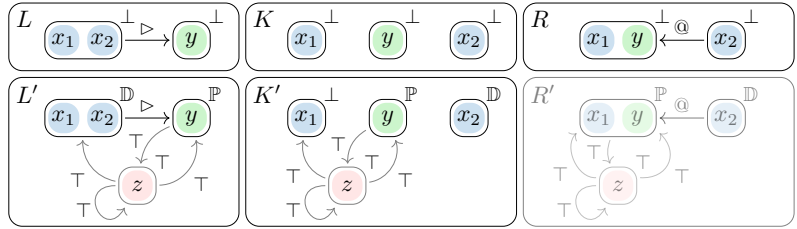
\newcommand{\nodexa}{\vertex{x_1}{cblue!20}}
\newcommand{\nodexb}{\vertex{x_2}{cblue!20}}
\newcommand{\nodexc}{\vertex{x_3}{cblue!20}}
\newcommand{\nodexd}{\vertex{x_4}{cblue!20}}

\newcommand{\nodea}{\vertex{a}{cgreen!20}}
\newcommand{\nodeb}{\vertex{b}{cpurple!25}}

\newcommand{\nodeaa}{\vertex{a_1}{cgreen!20}}
\newcommand{\nodeab}{\vertex{a_2}{cgreen!20}}
\newcommand{\nodeba}{\vertex{b_1}{cpurple!25}}
\newcommand{\nodebb}{\vertex{b_2}{cpurple!25}}

\newcommand{\nodex}{\vertex{x}{cblue!20}}

\newcommand{\nodez}{\vertex{z}{cred!10}}

\newcommand{\nodey}{\vertex{y}{cgreen!20}}

\begin{center}
  \scalebox{\rulescale}{
  \begin{tikzpicture}[->,node distance=12mm,n/.style={}]
    \graphbox{$L$}{0mm}{0mm}{34mm}{10mm}{-6mm}{-5.5mm}{
      \node [npattern] (x)
      {\nodexa \ \nodexb};
       \annotate{x}{$\bot$};
       \node [npattern] (y) [right of=x,xshift=4mm] {\nodey};
       \annotate{y}{$\bot$};
       \draw [epattern] (x) to node [above,label,inner sep=0.5mm] {$\rhd$} (y);
    }
    \graphbox{$L'$}{0mm}{-11mm}{34mm}{22mm}{-6mm}{-6mm}{
      \node [npattern] (x)
      {\nodexa \ \nodexb};
       \annotate{x}{$\mathbb{D}$};
        \node [npattern] (y) [right of=x, xshift=4mm] {\nodey};
       \annotate{y}{$\mathbb{P}$};
       \draw [epattern] (x) to node [above,label,inner sep=0.5mm] {$\rhd$} (y);
       % context node
       \node [nset] (z) at ($(x)!0.5!(y)$) [yshift=-10mm] {\nodez};
       \annotate{z}{$\top$};
       \draw [eset,loop=-150,looseness=3] (z) to node [left,label,inner sep=0.5mm] {$\top$} (z);
       \draw [eset] (z) to [bend left=40] node [below left,label,inner sep=0.5mm] {$\top$} (x);
       \draw [eset] (z) to[bend right=50] node [below right,label,inner sep=0.5mm] {$\top$} (y);
       \draw [eset] (y) to[bend right=30] node [ left,label,inner sep=0.5mm] {$\top$} (z);
    }
    \graphbox{$K'$}{35mm}{-11mm}{40mm}{22mm}{-11mm}{-6mm}{
      \node [npattern] (x)
      {\nodexa};
       \annotate{x}{$\bot$};
        \node [npattern] (y) [right of=x] {\nodey};
       \annotate{y}{$\mathbb{P}$};
       % context node
       \node [nset] (z) at ($(x)!0.5!(y)$) [yshift=-10mm] {\nodez};
       \annotate{z}{$\top$};
       \draw [eset,loop=-150,looseness=3] (z) to node [left,label,inner sep=0.5mm] {$\top$} (z);
       \draw [eset] (z) to [bend left=40] node [below left,label,inner sep=0.5mm] {$\top$} (x);
       \draw [eset] (z) to[bend right=50] node [below right,label,inner sep=0.5mm] {$\top$} (y);
       \draw [eset] (y) to[bend right=30] node [ left,label,inner sep=0.5mm] {$\top$} (z);
       % extra element
       \node [npattern] (x2) [right of=y] {$\nodexb$};
       \annotate{x2}{$\mathbb{D}$};
    }
    \graphbox{$K$}{35mm}{0mm}{40mm}{10mm}{-11mm}{-5.5mm}{
      \node [npattern] (x)
      {\nodexa};
       \annotate{x}{$\bot$};
        \node [npattern] (y) [right of=x] {\nodey};
       \annotate{y}{$\bot$};
       % extra element
       \node [npattern] (x2) [right of=y] {$\nodexb$};
       \annotate{x2}{$\bot$};
    }
    \graphbox{$R$}{76mm}{0mm}{40mm}{10mm}{-6mm}{-5.5mm}{
      \node [npattern] (x)
      {\nodexa \ \nodey};
       \annotate{x}{$\bot$};
       % extra element
       \node [npattern] (x2) [right of=x, xshift=4mm] {$\nodexb$};
       \annotate{x2}{$\bot$};
       \draw [epattern] (x2) to node [above,label,inner sep=0.5mm] {$@$} (x);
    }
    \transparentgraphbox{$R'$}{76mm}{-11mm}{40mm}{22mm}{-6mm}{-6mm}{
      \node [npattern] (x)
      {\nodexa \ \nodey};
       \annotate{x}{$\mathbb{P}$};
       % context node
       \node [nset] (z) [yshift=-10mm] {\nodez};
       \annotate{z}{$\top$};
       \draw [eset,loop=-150,looseness=3] (z) to node [left,label,inner sep=0.5mm] {$\top$} (z);
       \draw [eset] (z) to [bend left=70] node [below left,label,inner sep=0.5mm] {$\top$} (x);
       \draw [eset] (z) to[bend right=70,looseness=1.5] node [below right,label,inner sep=0.5mm] {$\top$} (x);
       \draw [eset] (x) to[bend right=20] node [ left,label,inner sep=0.5mm] {$\top$} (z);
       % extra element
        \node [npattern] (x2) [right of=x, xshift=4mm] {$\nodexb$};
       \annotate{x2}{$\mathbb{D}$};
       \draw [epattern] (x2) to node [above,label,inner sep=0.5mm] {$@$} (x);
    }
  \end{tikzpicture}
  }
\end{center}

}% scope
    \noindent The rule illustrates how sorts can improve readability and provide type safety. For instance, the label $\mathbb{D}$ in $L'$ prevents empty channels from being matched. More precisely, always the last element $d$ of a non-empty channel is matched. $K'$ duplicates the node holding $d$: for duplicate $x_1$, the label is forgotten but the connection to the context retained, allowing it to be fused with $y$; and for $x_2$, the connection is forgotten but the label retained, allowing it to be connected to $y$ as an otherwise isolated node.
\end{example}

Finally, a very powerful feature provided by the coupling of \pbpostrong and \GraphLattice{(\labels,\leq)} is the ability to model a general notion of variable, used for matching (possibly disjoint) parts of the context. This is achieved by using multiple context nodes in $L'$ (i.e., nodes not in the image of $t_L$).

\newcommand{\upcxtnode}{\mathcal{C}}

\begin{example}[Variables]
\label{example:variables}
    The rule
    $f(g(x),y) \to h(g(x), g(y), x)$
    on ordered trees
    can be precisely modeled in \pbpostrong by the rule
        {% scope

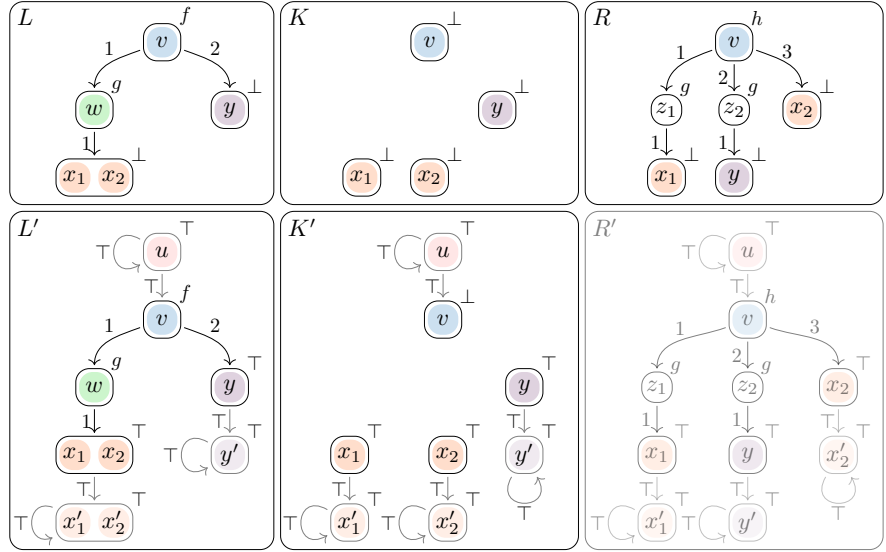
\newcommand{\noder}{\vertex{r}{cred!20}}
\newcommand{\nodex}{\vertex{v}{cblue!20}}
\newcommand{\nodey}{\vertex{w}{cgreen!20}}
\newcommand{\nodeu}{\vertex{u}{cred!10}}
\newcommand{\nodeva}{\vertex{x_1}{corange!20}}
\newcommand{\nodevb}{\vertex{x_2}{corange!20}}
\newcommand{\nodevap}{\vertex{x'_1}{corange!10}}
\newcommand{\nodevbp}{\vertex{x'_2}{corange!10}}
\newcommand{\nodew}{\vertex{y}{cpurple!25}}
\newcommand{\nodewp}{\vertex{y'}{cpurple!12}}

\begin{center}
  \scalebox{\rulescale}{
  \begin{tikzpicture}[->,node distance=12mm,n/.style={}]
    \graphbox{$L$}{0mm}{0mm}{39mm}{30mm}{3mm}{-6mm}{
      \node [npattern] (x) [short] {\nodex}; \annotate{x}{$f$};
      \node (c) [below of=x,short] {};
      \node [npattern] (y) [left of=c,short] {\nodey}; \annotate{y}{$g$};
      \node [npattern] (w) [right of=c,short] {\nodew}; \annotate{w}{$\bot$};
      \node [npattern] (v) [below of=y,short]
      {\nodeva \ \nodevb}; \annotate{v}{$\bot$};
      
      \draw [epattern] (x) to[out=-160,in=90] node [above,label,inner sep=1mm] {$1$} (y);
      \draw [epattern] (y) to node [left,label] {$1$} (v);
      \draw [epattern] (x) to[out=-20,in=90] node [above,label,inner sep=1mm] {$2$} (w);
    }
    \graphbox{$K$}{40mm}{0mm}{44mm}{30mm}{0mm}{-6mm}{
      \node [npattern] (x) [short] {\nodex}; \annotate{x}{$\bot$};
      \node (c) [below of=x,short] {};
      \node (y) [left of=c,short] {}; 
      \node [npattern] (w) [right of=c,short] {\nodew}; \annotate{w}{$\bot$};

      \node [npattern] (v1) [below of=y,short] {\nodeva}; \annotate{v1}{$\bot$};
      \node [npattern] (v2) [right of=v1,short] {\nodevb}; \annotate{v2}{$\bot$};
    }
    \graphbox{$R$}{85mm}{0mm}{44mm}{30mm}{0mm}{-6mm}{
      \node [npattern] (x) [short] {\nodex}; \annotate{x}{$h$};
      \node [npattern] (z2) [below of=x,short] {$z_2$}; \annotate{z2}{$g$};
      \node [npattern] (z1) [left of=z2,short] {$z_1$}; \annotate{z1}{$g$};
      \node [npattern] (v2) [right of=z2,short] {\nodevb}; \annotate{v2}{$\bot$};
      \node [npattern] (w) [below of=z2,short] {\nodew}; \annotate{w}{$\bot$};
      \node [npattern] (v1) [below of=z1,short] {\nodeva}; \annotate{v1}{$\bot$};

      \draw [epattern] (x) to[out=-160,in=90] node [above,label,inner sep=.5mm] {$1$} (z1);
      \draw [epattern] (x) to node [left,label,inner sep=0.5mm] {$2$} (z2);
      \draw [epattern] (x) to[out=-20,in=90] node [above,label,inner sep=0.5mm] {$3$} (v2);
      \draw [epattern] (z1) to node [left,label,inner sep=0.5mm] {$1$} (v1);
      \draw [epattern] (z2) to node [left,label,inner sep=0.5mm] {$1$} (w);
    }
    \graphbox{$L'$}{0mm}{-31mm}{39mm}{50mm}{3mm}{-6mm}{
      \node [nset] (u) [short] {\nodeu}; \annotate{u}{$\top$};
      \node [npattern] (x) [below of=u,short] {\nodex}; \annotate{x}{$f$};
      \node (c) [below of=x,short] {};
      \node [npattern] (y) [left of=c,short] {\nodey}; \annotate{y}{$g$};
      \node [npattern] (w) [right of=c,short] {\nodew}; \annotate{w}{$\top$};
      \node [npattern] (v) [below of=y,short] {\nodeva \ \nodevb}; \annotate{v}{$\top$};
      \node [nset] (v') [below of=v,short] {\nodevap \ \nodevbp}; \annotate{v'}{$\top$};
      \node [nset] (w') [below of=w,short] {\nodewp}; \annotate{w'}{$\top$};
      
      \draw [eset] (u) to node [left,label] {$\top$} (x);
      \draw [eset,loop=180,looseness=3] (u) to node [left,label] {$\top$} (u);
      \draw [epattern] (x) to[out=-160,in=90] node [above,label,inner sep=1mm] {$1$} (y);
      \draw [epattern] (y) to node [left,label] {$1$} (v);
      \draw [eset] (v) to node [left,label] {$\top$} (v');
      \draw [eset,thinloop=180,looseness=2.5] (v') to node [left,label] {$\top$} (v');
      \draw [epattern] (x) to[out=-20,in=90] node [above,label,inner sep=1mm] {$2$} (w);
      \draw [eset] (w) to node [left,label] {$\top$} (w');
      \draw [eset,loop=180,looseness=3] (w') to node [left,label] {$\top$} (w');
    }
    \graphbox{$K'$}{40mm}{-31mm}{44mm}{50mm}{2mm}{-6mm}{
      \node [nset] (u) [short] {\nodeu}; \annotate{u}{$\top$};
      \node [npattern] (x) [below of=u,short] {\nodex}; \annotate{x}{$\bot$};
      \node (c) [below of=x,short] {};
      \node (y) [left of=c] {};
      \node [npattern] (w) [right of=c] {\nodew}; \annotate{w}{$\top$};
      \node [nset] (w') [below of=w,short] {\nodewp}; \annotate{w'}{$\top$};
      \node [npattern] (v1) [below of=y,short,xshift=-1.8mm] {\nodeva}; %x1
      \annotate{v1}{$\top$};
      \node [npattern] (v2) [right of=v1, xshift=2.5mm] {\nodevb}; \annotate{v2}{$\top$};
      \node [nset] (v1') [below of=v1,short] {\nodevap}; %x_1'
      \annotate{v1'}{$\top$};
      \node [nset] (v2') [below of=v2,short] {\nodevbp}; \annotate{v2'}{$\top$};
      
      \draw [eset] (u) to node [left,label] {$\top$} (x);
      \draw [eset] (v1) to node [left,label] {$\top$} (v1');
      \draw [eset,loop=180,looseness=3] (v1') to node [left,label] {$\top$} (v1');
      \draw [eset] (v2) to node [left,label] {$\top$} (v2');
      \draw [eset,loop=180,looseness=3] (v2') to node [left,label] {$\top$} (v2');
      \draw [eset,loop=180,looseness=3] (u) to node [left,label] {$\top$} (u);
      \draw [eset] (w) to node [left,label] {$\top$} (w');
      \draw [eset,loop=270,looseness=3] (w') to node [below,label] {$\top$} (w');
    }
    \transparentgraphbox{$R'$}{85mm}{-31mm}{44mm}{50mm}{2mm}{-6mm}{
      \node [nset] (u) [short] {\nodeu}; \annotate{u}{$\top$};
      \node [npattern] (x) [below of=u,short] {\nodex}; \annotate{x}{$h$};
      \node [npattern] (z2) [below of=x,short] {$z_2$}; \annotate{z2}{$g$};
      \node [npattern] (z1) [left of=z2,xshift=-1.4mm] {$z_1$}; \annotate{z1}{$g$};
      \node [npattern] (v2) [right of=z2,xshift=1.4mm] {\nodevb}; \annotate{v2}{$\top$};
      \node [npattern] (w) [below of=z2,short] {\nodew}; \annotate{w}{$\top$};
      \node [nset] (w') [below of=w,short] {\nodewp}; \annotate{w'}{$\top$};
      \node [npattern] (v1) [below of=z1,short] {\nodeva}; \annotate{v1}{$\top$};
      \node [nset] (v1') [below of=v1,short] {\nodevap}; \annotate{v1'}{$\top$};
      \node [nset] (v2') [below of=v2,short] {\nodevbp}; \annotate{v2'}{$\top$};
      
      \draw [eset] (u) to node [left,label] {$\top$} (x);
      \draw [eset] (v1) to node [left,label] {$\top$} (v1');
      \draw [eset,loop=180,looseness=3] (v1') to node [left,label] {$\top$} (v1');
      \draw [eset] (v2) to node [left,label] {$\top$} (v2');
      \draw [eset,loop=270,looseness=3] (v2') to node [below,label] {$\top$} (v2');
      \draw [eset,loop=180,looseness=3] (u) to node [left,label] {$\top$} (u);
      \draw [eset] (w) to node [left,label] {$\top$} (w');
      \draw [eset,loop=180,looseness=3] (w') to node [left,label] {$\top$} (w');

      \draw [epattern] (x) to[out=-160,in=90] node [above,label,inner sep=.5mm] {$1$} (z1);
      \draw [epattern] (x) to node [left,label,inner sep=0.5mm] {$2$} (z2);
      \draw [epattern] (x) to[out=-20,in=90] node [above,label,inner sep=0.5mm] {$3$} (v2);
      \draw [epattern] (z1) to node [left,label,inner sep=0.5mm] {$1$} (v1);
      \draw [epattern] (z2) to node [left,label,inner sep=0.5mm] {$1$} (w);
    }
  \end{tikzpicture}
  }
\end{center}

}% scope
    \noindent if one restricts the set of rewritten graphs to straightforward representations of trees: nodes are labeled by symbols, and edges are labeled by $n \in \mathbb{N}$, the position of its target (argument of the symbol).
\end{example}

\newcommand{\encode}[1]{\mathcal{E}(#1)}

In another paper~\cite{overbeek2021from}, we developed the idea of Example~\ref{example:variables} further, and showed that any linear term rewriting system $R$ can be faithfully represented by a \pbpostrong graph rewrite system encoding $\encode{R}$, with the additional property that $R$ terminates iff $\encode{R}$ terminates on finite graphs~\cite[Theorem 62]{overbeek2021from}.

\section{Discussion}
\label{sec:discussion}

We discuss our rewriting (Section~\ref{sec:discussion:rewriting}) and relabeling (Section~\ref{sec:discussion:relabeling}) contributions in turn.

\subsection{Rewriting}
\label{sec:discussion:rewriting}

Other graph rewriting approaches that bear certain similarities to \pbpostrong{} (see also the discussion in~\cite{corradini2019pbpo}) include the double-pullout graph rewriting approach by Kahl~\cite{kahl2010amalgamating}; the cospan SqPO approach by Mantz~\cite[Section 4.5]{mantz2014phd}; and the recent drag rewriting framework by Dershowitz and Jouannaud~\cite{dershowitz2019drags}. Double-pullout graph rewriting also uses pullbacks and pushouts to delete and duplicate parts of the context (extending DPO), but the approach is defined in the context of collagories~\cite{kahl2011collagories}, and to us it is not yet clear in what way the two approaches relate. Cospan SqPO can be understood as being almost dual to SqPO: rules are cospans, and transformation steps consists of a pushout followed by a final pullback complement. An interesting question is whether \pbpostrong{} can also model cospan SqPO. Drag rewriting is a non-categorical approach to generalizing term rewriting, and like \pbpostrong{}, allows relatively fine control over the interface between pattern and context, thereby avoiding issues related to dangling pointers and the construction of pushout complements. Because drag rewriting is non-categorical and drags have inherently more structure than graphs, it is difficult to relate \pbpostrong{} and drag rewriting precisely. These could all be topics for future investigation.

Let us note that the combination of \pbpostrong{} and \GraphLattice{(\labels, \leq)} does not provide a strict generalization of Patch Graph Rewriting (PGR)~\cite{overbeek2020patch}, our conceptual precursor to \pbpostrong{} (Section~\ref{sec:introduction}). This is because patch edge endpoints that lie in the context graph can be redefined in PGR (e.g., the direction of edges between context and pattern can be inverted), but not in \pbpostrong{}. 
Beyond that, \pbpostrong{} is more general and expressive. Therefore, at this point we believe that the most distinguishing and redeeming feature of PGR is its visual syntax, which makes rewrite systems much easier to define and communicate. In order to combine the best of both worlds, our aim is to define a similar syntax for (a suitable restriction of) \pbpostrong{} in the future.

\subsection{Relabeling}
\label{sec:discussion:relabeling}

The coupling of \pbpostrong{} and \GraphLattice{(\labels, \leq)} allows relabeling and modeling sorts and variables with relative ease, and does not require a modification of the rewriting framework. Most existing approaches study these topics in the context of DPO, where the requirement to ensure the unique existence of a pushout complement requires restricting the method and proving non-trivial properties:
\begin{itemize}
    \item 
      Parisi-Presicce et al.~\cite{parisi1986graph} limit DPO rules $L \leftarrow K \to R$ to ones where $K \to R$ is monic (meaning merging is not possible), and where some set-theoretic consistency condition is satisfied. Moreover, the characterization of the existence of rewrite step has been shown to be incorrect~\cite{habel2002relabelling}, supporting the idea that pushout complements are not easy to reason about.
    \item 
      Habel and Plump~\cite{habel2002relabelling} study relabeling using the category of partially labeled graphs. They allow non-monic morphisms $K \to R$, but they nonetheless add two restrictions to the definition of a DPO rewrite rule. Among others, these conditions do not allow hard overwriting arbitrary labels as in Example~\ref{example:hard:overwriting}. Moreover, the pushouts of the DPO rewrite step must be restricted to pushouts that are also pullbacks.
      Finally, unlike the approach suggested by Parisi-Presice et al., Habel and Plump's approach does not support modeling notions of sorts and variables.
      
      Our result that \pbpostrong{} can model DPO in $\Graph$ extends to this relabeling approach in the following sense: given a DPO rule over graphs partially labeled from $\labels$ that moreover satisfies the criteria of~\cite{habel2002relabelling}, we conjecture that there exists a \pbpostrong{} rule in \GraphLattice{(\labels^{\bot,\top}, \leq)} that models the same rewrite relation when restricting to graphs totally labeled over the base label set $\labels$.
\end{itemize}
Later publications largely appear to build on the approach~\cite{habel2002relabelling} by Habel and Plump. For example, Schneider~\cite{schneider2005changing} gives a non-trivial categorical formulation; Hoffman~\cite{hoffman2005graph} proposes a two-layered (set-theoretic) approach to support variables; and Habel and Plump~\cite{habel2012adhesive} generalize their approach to $\mathcal{M}, \mathcal{N}$-adhesive systems (again restricting $K \to R$ to monic arrows). 

Independent of Habel and Plump's approach, there also exists an approach for DPO by Kahl, in which attributed graphs are considered as coalgebras~\cite{kahl2014graph}, providing support for relabeling operations. Relating Kahl's approach to ours is outside the scope of the present paper.

The transformation of attributed structures has been explored in a very general setting by Corradini et al.~\cite{corradini2019pbpo}, which involves an elegant comma category construction and suitable restrictions of the PBPO notions of rewrite rule and rewrite step. We leave relating their and our approach to future work.

\subsubsection*{Acknowledgments} We thank Nicolas Behr for inspiring discussions on quasitoposes and \pbpostrong, and for pointing us to some important references; Andrea Corradini for providing guidance on AGREE and PBPO, and for useful feedback on the conference version of this paper; Clemens Grabmayer for helpful commentary on the subsumption result; and Michael Shulman, who identified the sufficient conditions for materialization for us~\cite{shulman2021mathoverflow}. We are also grateful to anonymous reviewers for many useful suggestions and corrections.

The authors received funding from the Netherlands Organization for Scientific Research (NWO) under the Innovational Research Incentives Scheme Vidi (project.\ No.\ VI.Vidi.192.004).

\bibliographystyle{unsrt}
\bibliography{main}

\begin{thebibliography}{10}

\bibitem{overbeek2020patch}
R.~Overbeek and J.~Endrullis.
\newblock Patch graph rewriting.
\newblock In {\em Proc.\ Conf.\ on Graph Transformation ({ICGT})}, volume 12150
  of {\em LNCS}, pages 128--145. Springer, 2020.

\bibitem{corradini2019pbpo}
A.~Corradini, D.~Duval, R.~Echahed, F.~Prost, and L.~Ribeiro.
\newblock The {PBPO} graph transformation approach.
\newblock {\em J.\ Log.\ Algebraic Methods Program.}, 103:213--231, 2019.

\bibitem{wyler1991lecture}
O.~Wyler.
\newblock {\em Lecture Notes on Topoi and Quasitopoi}.
\newblock World Scientific Publishing Co., 1991.

\bibitem{lack2004adhesive}
S.~Lack and P.~Soboci\'{n}ski.
\newblock Adhesive categories.
\newblock In {\em Proc.\ Conf.\ on Foundations of Software Science and
  Computation Structures (FOSSACS)}, volume 2987 of {\em LNCS}, pages 273--288.
  Springer, 2004.

\bibitem{corradini2015agree}
A.~Corradini, D.~Duval, R.~Echahed, F.~Prost, and L.~Ribeiro.
\newblock {AGREE} -- algebraic graph rewriting with controlled embedding.
\newblock In {\em Proc.\ Conf.\ on Graph Transformation (ICGT)}, volume 9151 of
  {\em LNCS}, pages 35--51. Springer, 2015.

\bibitem{johnstone2007quasitoposes}
P.~T. Johnstone, S.~Lack, and P.~Sobocinski.
\newblock Quasitoposes, quasiadhesive categories and {Artin} glueing.
\newblock In {\em Algebra and Coalgebra in Computer Science ({CALCO})}, volume
  4624 of {\em LNCS}, pages 312--326. Springer, 2007.

\bibitem{corradini2020algebraic}
A.~Corradini, D.~Duval, R.~Echahed, F.~Prost, and L.~Ribeiro.
\newblock Algebraic graph rewriting with controlled embedding.
\newblock {\em Theor. Comput. Sci.}, 802:19--37, 2020.

\bibitem{ehrig1973graph}
H.~Ehrig, M.~Pfender, and H.~J. Schneider.
\newblock Graph-grammars: An algebraic approach.
\newblock In {\em Proc.\ Symp.\ on on Switching and Automata Theory (SWAT)},
  page 167–180. IEEE Computer Society, 1973.

\bibitem{overbeek2021pbpo}
R.~Overbeek, J.~Endrullis, and A.~Rosset.
\newblock Graph rewriting and relabeling with {PBPO$^{+}$}.
\newblock In {\em Proc.\ Conf.\ on Graph Transformation (ICGT)}, volume 12741
  of {\em LNCS}, pages 60--80. Springer, 2021.

\bibitem{overbeek2023tutorial}
R.~Overbeek and J.~Endrullis.
\newblock A {PBPO}$^{+}$ graph rewriting tutorial.
\newblock In {\em Proc.\ of Workshop on Computing with Terms and Graphs
  (TERMGRAPH)}, volume 377 of {\em {EPTCS}}, pages 45--63. Open Publishing
  Association, 2023.

\bibitem{mac1971categories}
S.~Mac~Lane.
\newblock {\em Categories for the Working Mathematician}, volume~5.
\newblock Springer Science \& Business Media, 1971.

\bibitem{awodey2006category}
S.~Awodey.
\newblock {\em {Category Theory}}.
\newblock Oxford University Press, 2006.

\bibitem{ehrig2006}
H.~Ehrig, K.~Ehrig, U.~Prange, and G.~Taentzer.
\newblock {\em Fundamentals of Algebraic Graph Transformation}.
\newblock Springer, 2006.

\bibitem{lowe1993algebraic}
M.~L{\"{o}}we.
\newblock Algebraic approach to single-pushout graph transformation.
\newblock {\em Theor. Comput. Sci.}, 109(1{\&}2):181--224, 1993.

\bibitem{corradini2006sesqui}
A.~Corradini, T.~Heindel, F.~Hermann, and B.~K{\"{o}}nig.
\newblock Sesqui-pushout rewriting.
\newblock In {\em Proc.\ Conf.\ on Graph Transformation (ICGT)}, volume 4178 of
  {\em LNCS}, pages 30--45. Springer, 2006.

\bibitem{behr2021concurrency}
Nicolas Behr, Russ Harmer, and Jean Krivine.
\newblock Concurrency theorems for non-linear rewriting theories.
\newblock In {\em Proc.\ Conf.\ on Graph Transformation (ICGT)}, volume 12741
  of {\em LNCS}, pages 3--21. Springer, 2021.

\bibitem{bauderon1995uniform}
M.~Bauderon.
\newblock A uniform approach to graph rewriting: The pullback approach.
\newblock In {\em Proc. of Graph-Theoretic Concepts in Computer Science},
  volume 1017 of {\em LNCS}, pages 101--115. Springer, 1995.

\bibitem{bauderon2001pullback}
M.~Bauderon and H.~Jacquet.
\newblock Pullback as a generic graph rewriting mechanism.
\newblock {\em Appl. Categorical Struct.}, 9(1):65--82, 2001.

\bibitem{cockett2002restrictionI}
J.~R.~B. Cockett and Stephen Lack.
\newblock Restriction categories {I:} categories of partial maps.
\newblock {\em Theor. Comput. Sci.}, 270(1-2):223--259, 2002.

\bibitem{cockett2003restrictionII}
J.~R.~B. Cockett and S.~Lack.
\newblock Restriction categories {II:} partial map classification.
\newblock {\em Theor. Comput. Sci.}, 294(1/2):61--102, 2003.

\bibitem{ehrig2010categorical}
H.~Ehrig, U.~Golas, and F.~Hermann.
\newblock Categorical frameworks for graph transformation and {HLR} systems
  based on the {DPO} approach.
\newblock {\em Bull. {EATCS}}, 102:111--121, 2010.

\bibitem{adamek2009joy}
J.~Ad{\'{a}}mek, H.~Herrlich, and G.~E. Strecker.
\newblock {\em Abstract and Concrete Categories - The Joy of Cats}.
\newblock Dover Publications, 2009.

\bibitem{johnstone2002sketches}
P.~T. Johnstone.
\newblock {\em Sketches of an Elephant: A Topos Theory Compendium, Volume 1}.
\newblock Oxford University Press, 2002.

\bibitem{heindel2010hereditary}
T.~Heindel.
\newblock Hereditary pushouts reconsidered.
\newblock In {\em Proc.\ Conf.\ on Graph Transformation (ICGT)}, volume 6372 of
  {\em LNCS}, pages 250--265. Springer, 2010.

\bibitem{garner2012axioms}
R.~Garner and S.~Lack.
\newblock On the axioms for adhesive and quasiadhesive categories.
\newblock {\em Theory and Applications of Categories}, 27(3):27--46, 2012.

\bibitem{goguen1967fuzzy}
J.~A. Goguen.
\newblock {$L$-fuzzy sets}.
\newblock {\em Journal of Mathematical Analysis and Applications},
  18(1):145--174, 1967.

\bibitem{stout1992thelogic}
L.~N. Stout.
\newblock {\em The Logic of Unbalanced Subobjects in a Category with Two Closed
  Structures}, pages 73--105.
\newblock Springer Netherlands, Dordrecht, 1992.

\bibitem{rosset2023fuzzy}
A.~Rosset, R.~Overbeek, and J.~Endrullis.
\newblock Fuzzy presheaves are quasitoposes.
\newblock {\em CoRR}, arXiv:2301.13067, 2023.

\bibitem{goldblatt1979topoi}
R.~Goldblatt.
\newblock {\em Topoi: the categorial analysis of logic}.
\newblock North-Holland, 1979.

\bibitem{corradini2019rewriting}
A.~Corradini, T.~Heindel, B.~K{\"{o}}nig, D.~Nolte, and A.~Rensink.
\newblock Rewriting abstract structures: Materialization explained
  categorically.
\newblock In {\em Proc.\ Conf.\ on Foundations of Software Science and
  Computation Structures (FOSSACS)}, volume 11425 of {\em LNCS}, pages
  169--188. Springer, 2019.

\bibitem{corradini2019rewritingarXiv}
A.~Corradini, T.~Heindel, B.~K{\"{o}}nig, D.~Nolte, and A.~Rensink.
\newblock Rewriting abstract structures: Materialization explained
  categorically.
\newblock {\em CoRR}, arXiv:1902.04809, 2019.

\bibitem{mori1995fuzzy}
Masao Mori and Yasuo Kawahara.
\newblock Fuzzy graph rewritings.
\newblock {\em \begin{CJK}{UTF8}{min}数理解析研究所講究録\end{CJK}},
  918:65--71, 1995.

\bibitem{parasyuk2006categorical}
I.~N. Parasyuk and S.~V. Yershov.
\newblock Categorical approach to the construction of fuzzy graph grammars.
\newblock {\em Cybernetics and Systems Analysis}, 42:570--581, 2006.

\bibitem{parasyuk2007categorical}
I.~N. Parasyuk and S.~V. Ershov.
\newblock Transformations of fuzzy graphs specified by {FD}-grammars.
\newblock {\em Cybernetics and Systems Analysis}, 43:266--280, 2007.

\bibitem{parasyuk2008transformational}
I.~N. Parasyuk and S.~V. Yershov.
\newblock Transformational approach to the development of software
  architectures on the basis of fuzzy graph models.
\newblock {\em Cybernetics and Systems Analysis}, 44:749--759, 2008.

\bibitem{overbeek2021from}
R.~Overbeek and J.~Endrullis.
\newblock From linear term rewriting to graph rewriting with preservation of
  termination.
\newblock In {\em Proc. Workshop on Graph Computational Models (GCM)}, volume
  350 of {\em {EPTCS}}, pages 19--34, 2021.

\bibitem{kahl2010amalgamating}
W.~Kahl.
\newblock Amalgamating pushout and pullback graph transformation in
  collagories.
\newblock In {\em Proc.\ Conf.\ on Graph Transformation (ICGT)}, volume 6372 of
  {\em LNCS}, pages 362--378. Springer, 2010.

\bibitem{mantz2014phd}
F.~Mantz.
\newblock {\em Coupled Transformations of Graph Structures applied to Model
  Migration}.
\newblock PhD thesis, University of Marburg, 2014.

\bibitem{dershowitz2019drags}
N.~Dershowitz and J.{-}P. Jouannaud.
\newblock Drags: {A} compositional algebraic framework for graph rewriting.
\newblock {\em Theor. Comput. Sci.}, 777:204--231, 2019.

\bibitem{kahl2011collagories}
W.~Kahl.
\newblock Collagories: Relation-algebraic reasoning for gluing constructions.
\newblock {\em J. Log. Algebraic Methods Program.}, 80(6):297--338, 2011.

\bibitem{parisi1986graph}
F.~Parisi{-}Presicce, H.~Ehrig, and U.~Montanari.
\newblock Graph rewriting with unification and composition.
\newblock In {\em Proc.\ Workshop on Graph-Grammars and Their Application to
  Computer Science}, volume 291 of {\em LNCS}, pages 496--514. Springer, 1986.

\bibitem{habel2002relabelling}
A.~Habel and D.~Plump.
\newblock Relabelling in graph transformation.
\newblock In {\em Proc.\ Conf.\ on Graph Transformation ({ICGT})}, volume 2505
  of {\em LNCS}, pages 135--147. Springer, 2002.

\bibitem{schneider2005changing}
H.~J. Schneider.
\newblock Changing labels in the double-pushout approach can be treated
  categorically.
\newblock In {\em Formal Methods in Software and Systems Modeling}, volume 3393
  of {\em LNCS}, pages 134--149. Springer, 2005.

\bibitem{hoffman2005graph}
B.~Hoffmann.
\newblock Graph transformation with variables.
\newblock In {\em Formal Methods in Software and Systems Modeling}, volume 3393
  of {\em LNCS}, pages 101--115. Springer, 2005.

\bibitem{habel2012adhesive}
A.~Habel and D.~Plump.
\newblock $\mathcal{M}$, $\mathcal{N}$-adhesive transformation systems.
\newblock In {\em Proc.\ Conf.\ on Graph Transformation ({ICGT})}, volume 7562
  of {\em LNCS}, pages 218--233. Springer, 2012.

\bibitem{kahl2014graph}
W.~Kahl.
\newblock Graph transformation with symbolic attributes via monadic coalgebra
  homomorphisms.
\newblock {\em Electron. Commun. Eur. Assoc. Softw. Sci. Technol.}, 71, 2014.

\bibitem{shulman2021mathoverflow}
M.~Shulman.
\newblock Subobject- and factorization-preserving typings.
\newblock MathOverflow.
\newblock URL:https://mathoverflow.net/q/381933 (version: 2021-01-22).

\end{thebibliography}

\end{document}